\documentclass[aps, pra,twocolumn, showpacs,%nofootinbib,
superscriptaddress,10pt,floatfix,longbibliography]{revtex4-2}

\pdfoutput=1

\usepackage[utf8]{inputenc}
\usepackage{braket}
\usepackage[caption=false]{subfig}
\usepackage{amsmath}
\usepackage{mathtools}
\usepackage{amssymb}
\usepackage{amsfonts}
\usepackage{graphicx}
\usepackage{hyperref}
\usepackage{mathrsfs}
\usepackage{bbold}
\usepackage[]{units}
\usepackage{bm}
\usepackage{braket}
\usepackage{color}
\usepackage{makecell}
\usepackage{enumitem}
\usepackage{upgreek}
\usepackage{physics}
\usepackage{amsmath,amsthm,amssymb}
\usepackage{xfrac}
\usepackage[dvipsnames]{xcolor}
\usepackage{hyperref}
\usepackage{braket}
\usepackage{bbm}
\usepackage{bm}
\usepackage{longtable}
\usepackage[bb=boondox]{mathalfa}
\usepackage{adjustbox}
\usepackage{array}
\usepackage{multirow}
\usepackage{tabularx}

\usepackage{float}

\usepackage{verbatim}

\usepackage{tikz}
\usetikzlibrary{shapes,arrows,patterns}

\newtheorem{lemma}{Lemma}
\newtheorem{corollary}{Corollary}

\newcommand{\cg}{{\mathrm{cg}}}

\renewcommand{\[}{\begin{equation}}
\renewcommand{\]}{\end{equation}}

\newcommand{\C}{{\mathcal{C}}}
\renewcommand{\P}{{P}}
\newcommand{\CR}{{\mathcal{C}_{\mathrm{R}}}}
\newcommand{\CM}{{\mathcal{C}_{\mathrm{M}}}}
\newcommand{\bR}{{\beta_{\mathrm{R}}}}
\newcommand{\bM}{{\beta_{\mathrm{M}}}}
\newcommand{\ham}{{H}}
\newcommand{\I}{{I}}
\newcommand{\R}{{\rho}}
\newcommand{\vN}{{\mathrm{vN}}}
\newcommand{\mean}[1]{\langle#1\rangle}
\newcommand{\ccdots}{{\!\cdot\!\cdot\!\cdot\!}}
\newcommand{\HS}{{\mathcal{H}}}
\newcommand{\CE}{{C_{\!E}}}
\newcommand{\tCE}{{\tilde C_{\!E}}}

\newcommand{\K}{{\mathcal{K}}}

\definecolor{rkrPurple}{HTML}{73024F}

\usetikzlibrary{decorations.markings,arrows,decorations.pathreplacing,calc,intersections,through,backgrounds}

\tikzset{->-/.style={decoration={markings,mark=at position #1 with {\arrow{>}}},postaction={decorate}}}
\tikzset{
	partial ellipse/.style args={#1:#2:#3}{
		insert path={+ (#1:#3) arc (#1:#2:#3)}
	}
}

\definecolor{myred}{RGB}{239,59,39}
\definecolor{myblue}{RGB}{43,80,163}

\newcommand{\redleft}{\begin{tikzpicture}[shift={(-.5,-0.4)}]
            \fill[myred] (0.125,0.140) circle (0.08cm);
            \draw[-] (0,0) -- (0.5,0);
            \draw[-] (0.25,0) -- (0.25,0.25);
            \draw[-] (0,0) -- (0,0.25);
        \end{tikzpicture}}
\newcommand{\blueleft}{\begin{tikzpicture}[shift={(-.5,-0.4)}]
            \fill[myblue] (0.125,0.140) circle (0.08cm);
            \draw[-] (0,0) -- (0.5,0);
            \draw[-] (0.25,0) -- (0.25,0.25);
            \draw[-] (0,0) -- (0,0.25);
        \end{tikzpicture}}
\newcommand{\blueright}{\begin{tikzpicture}[shift={(-.5,-0.4)}]
            \fill[myblue] (0.375,0.140) circle (0.08cm);
            \draw[-] (0,0) -- (0.5,0);
            \draw[-] (0.25,0) -- (0.25,0.25);
            \draw[-] (0.5,0) -- (0.5,0.25);
        \end{tikzpicture}}
\newcommand{\redright}{\begin{tikzpicture}[shift={(-.5,-0.4)}]
            \fill[myred] (0.375,0.140) circle (0.08cm);
            \draw[-] (0,0) -- (0.5,0);
            \draw[-] (0.25,0) -- (0.25,0.25);
            \draw[-] (0.5,0) -- (0.5,0.25);
        \end{tikzpicture}}
\newcommand{\blackleft}{\begin{tikzpicture}[shift={(-.5,-0.4)}]
            \fill[black] (0.125,0.140) circle (0.08cm);
            \draw[-] (0,0) -- (0.5,0);
            \draw[-] (0.25,0) -- (0.25,0.25);
            \draw[-] (0,0) -- (0,0.25);
        \end{tikzpicture}}
\newcommand{\blackright}{\begin{tikzpicture}[shift={(-.5,-0.4)}]
            \fill[black] (0.375,0.140) circle (0.08cm);
            \draw[-] (0,0) -- (0.5,0);
            \draw[-] (0.25,0) -- (0.25,0.25);
            \draw[-] (0.5,0) -- (0.5,0.25);
        \end{tikzpicture}}

\begin{document}

\title{
Work and entropy of mixing in isolated quantum systems\\
}

\author{Budhaditya Bhattacharjee}
\email{budhadityab@ibs.re.kr}
\affiliation{Center for Theoretical Physics of Complex Systems, Institute for Basic Science (IBS), Daejeon - 34126, Korea}

\author{Rohit Kishan Ray}
\email{rkray@vt.edu}
\affiliation{Center for Theoretical Physics of Complex Systems, Institute for Basic Science (IBS), Daejeon - 34126, Korea}
\affiliation{Department of Material Science and Engineering, Virginia Tech, Blacksburg, VA 24061, USA}

\author{Dominik \v{S}afr\'{a}nek}
\email{dsafranekibs@gmail.com}
\affiliation{Center for Theoretical Physics of Complex Systems, Institute for Basic Science (IBS), Daejeon - 34126, Korea}

\date{\today}% It is always \today, today,

\begin{abstract}
The mixing of two different gases is one of the most common natural phenomena, with applications ranging from $\text{CO}_2$ capture to water purification. Traditionally, mixing is analyzed in the context of local thermal equilibrium, where systems exchange energy with a heat bath. Here, we study mixing in an isolated system with potentially non-equilibrium initial states, characterized solely by macroscopic observables. We identify the entropy of mixing as a special case of observational entropy within an observer-dependent framework, where both entropy and extractable work depend on the resolution of measurement. This approach naturally resolves the Gibbs mixing paradox in quantum systems: while an observer experiences a discontinuous increase in entropy upon learning of the existence of two particle types, this knowledge does not provide an advantage in work extraction if the types of particles remain operationally indistinguishable in their measurements. Finally, we derive a Landauer-like bound on the difference in energy extracted by two observers, where an ‘observational temperature’ emerges, determined by the accessible information. These results provide a foundation for rigorously determining the energy required to unmix in non-equilibrium settings and extend beyond quantum systems, offering insights into the thermodynamics of isolated classical gases.
\end{abstract}

\maketitle

\section{Introduction}

Mixing two types of particles is one of the most common occurrences in nature. Whether it is the down-to-earth act of putting milk into our coffee or the release of $\text{CO}_2$ into the atmosphere, the underlying physical laws remain the same. Such processes also have practical implications—most notably in determining how much energy is needed to separate two compounds once mixed. For gases, this minimal energy has a fundamental bound, revealing that capturing $\text{CO}_2$ \emph{before} it is released is far more efficient than attempting to remove it afterward~\cite{carbon2020}. A similar principle applies to other modern environmental challenges, such as desalination~\cite{KARAGIANNIS2008water} or filtering microplastics~\cite{chandra2024microplastics} and hormones~\cite{gonsioroski2020endocrine} from our water supply, where proactive measures can significantly reduce both costs and ecological impacts.

In classical physics, where gases are assumed to be in thermal equilibrium at a well-defined temperature, the energy cost of mixing is well understood. However, in systems that do not equilibrate or are not connected to a thermal bath the usual assumptions break down. Consequently, it becomes unclear how much energy is required to separate the gases. This challenge is especially pronounced in isolated quantum systems, which are now routinely realized and studied in experimental settings~\cite{gong2021experimental, bernon2013, leonard2023}.

These questions are also fundamentally observer-dependent: to study how much energy it costs to separate the two gases, one must first define what it means ``to separate.'' Typically, this notion carries a connotation of spatial separation, such as one type of gas being confined within a box while the other remains outside. Even in this case, however, the energy cost depends on the size of the box, which is determined by the observer or the agent attempting to separate the gases~\cite{blundell2010concepts}. This definition is inherently tied to what the agent can achieve, shaped by their specific capabilities and limitations in carrying out the separation~\cite{jaynes1992gibbs}.

This interplay between the observer’s experimental capabilities and knowledge lies at the heart of the Gibbs mixing paradox (GMP). Consider two observers, one who can distinguish the two gases and the other who cannot. Recognizing the difference, the first observer assigns a higher entropy to the observed state, based on the larger number of microstates in compliance with any macroscopic observation. The second observer perceives no distinction and assigns a lower entropy, asserting that all particles are of the same type. Crucially, the first observer’s knowledge provides an advantage---if the gases are spatially separated, they can extract energy by allowing the gases to mix (classically, through semi-permeable membranes). The second observer, unaware of the distinction, cannot exploit this opportunity. As E.T. Jaynes noted~\cite{jaynes1992gibbs}: ``We would observe, however, that the number of fish that you can catch is an `objective experimental fact'; yet it depends on how much `subjective' information you have about the behavior of fish.''

While the role of different observers is widely acknowledged in these discussions, their treatment remains somewhat ambiguous in the absence of a formal framework defining what constitutes an ``observer.'' Instead, the situation is typically described qualitatively, with explanations stating that one observer can distinguish the two gases while the other cannot. From this premise, the appropriate definition of equilibrium entropy is either postulated or argued for, but it is not systematically deduced from the original statement. This lack of rigorous derivation leaves the connection between the observer's capabilities and the entropy definition somewhat imprecise. Does ``being able to distinguish'' imply the ability to probe each particle individually and determine its type, or does it refer to probing a specific region of space to ascertain the proportion of particles of each type within it? These two interpretations are fundamentally different and lead to distinct outcomes, both in terms of the observer's knowledge and the resulting thermodynamic description of the system.

This also highlights another aspect of the paradox: the abrupt, discontinuous change in entropy that occurs when particles are distinguishable, even if differences between them are minimal. The mere potential for distinction leads to a significantly larger associated entropy. This apparent paradox arises from treating entropy as a fundamental property of the physical system. However, if entropy is instead understood as an informational measure---reflecting the observer’s perceived ignorance about the system and defining the limits of their ability to extract energy---the paradox dissolves. In this interpretation, entropy is relational, shaped by the observer's knowledge and capabilities rather than an intrinsic characteristic of the system.

In this paper, we formalize the concept of an observer and their capabilities in precise mathematical terms. This systematic approach removes the need for overly intuitive reasoning. Instead, it allows us to directly approach the question with higher clarity and precision. Thereby it allows us to quantify the additional energy that an observer with a higher resolution can extract from a system compared to an observer with a lower resolution. As an application, we analyze the difference in extracted energy between two observers in the Gibbs mixing scenario, in the generic 1D lattice model with fermions of two colors. Additionally, we focus on isolated quantum systems and compute the difference in energy that an observer can extract before and after the two gases mix, thereby establishing a lower bound on the energy required to separate them again.

Interestingly, this result bears a strong mathematical resemblance to Landauer's bound, with a key difference: instead of the bath temperature, an ``observational temperature'' emerges, defined by the observer's experimental capabilities and limitations. Furthermore, since equivalent reasoning applies to classical systems, the insights gained extend beyond the quantum case. The framework allows us to discuss the energy required to separate isolated classical gases out of equilibrium, broadening the applicability of the results.

First, we review the background including mixing and observational entropy (Sec.~\ref{sec:literature}). Then, we derive the difference between extracted energy by two observers (Sec.~\ref{sec:energydiff}), with a specific result resembling Landauer's bound for observers that have similar capabilities (Sec.~\ref{sec:energydiffsim}), followed by the difference in extracted energy by the same observer at different times (Sec.~\ref{sec:difference_in_times}). In Sec.~\ref{sec:mixingentropy}, we analyze entropy in the Gibbs mixing scenario in quantum systems for an observer that can distinguish the particles (Sec.~\ref{sec:mixingentropyrick}) and three versions of an observer who cannot (Sec.~\ref{sec:mixingentropymorty1},~\ref{sec:mixingentropymorty2},~\ref{sec:mixingentropymorty3}), and compare them (Sec.~\ref{sec:entropydifference}). We also analyze the change in entropy for these observers when an isolated quantum system evolves~\ref{sec:timeevolution}, showing that for all observers with various degrees of ignorance, the observed change in entropy is the same. Then we combine the previous results to show the difference between extracted energy in the Gibbs mixing scenario (Sec.~\ref{sec:workdiffmixing}). In particular, we outline the assumptions under which we prove that an observer unaware of the two particle types extracts the same amount of energy as one who is aware, thereby resolving the Gibbs mixing paradox. We also include two simulations: first, the difference between two observers, and second, the difference for the same observer but at different times. 
Finally, we conclude (Sec.~\ref{sec:conclusion}).

\begin{figure}
    \centering
    \includegraphics[width=0.45\textwidth]{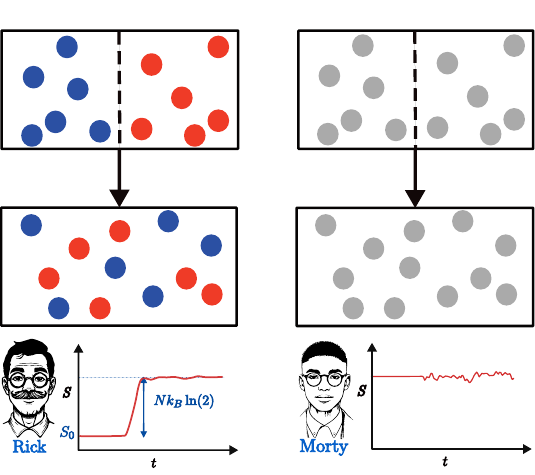}
    \caption{\label{fig:mgmp_schematic} A schematic representation of the Gibbs mixing paradox. Rick can distinguish the particles in the two partitions based on their colors, whereas Morty cannot. As a result, when the partition is removed, Rick observes a net increase in entropy due to the mixing of two distinct gases, while Morty, unable to perceive the distinction, does not observe the same change in entropy.}\label{fig:gibbssketch}
\end{figure}

\section{Background}\label{sec:literature}

Consider two ideal gases—initially separated by a partition, both with identical volume, pressure, and temperature—that are allowed to mix. The Gibbs Mixing Paradox (GMP) arises because of the differing observer abilities. An Observer Rick, who can distinguish the gases notices an entropy increase from mixing. Morty, on the other hand, cannot distinguish them and thus attributes no entropy change. This discontinuity depends on the gases being distinguishable, challenging classical thermodynamic notions of entropy (see Fig. \ref{fig:gibbssketch}). Note here, that we consider distinguishability between the gases rather than their constituent individual particles, which may be fundamentally indistinguishable in quantum physics. In second quantization, this corresponds to being able to distinguish different spatial modes or different values of spin.

\textbf{Gibbs Mixing Paradox.}
Jaynes~\cite{jaynes1992gibbs} framed entropy as a subjective measure tied to an observer’s coarse-graining of microstates. For Rick, reverting the system to its original state requires work proportional to the entropy change; for Morty, simply reinserting the partition suffices. This perspective anchors entropy in information. Building on Jaynes' perspective, many authors have explored the GMP from both philosophical and physical standpoints~\cite{tatarin_1999_entropy, allahverdyan_2006_explanation, ben-naim_2007_socalled, lin_2008_gibbs, maslov_2008_solution, swendsen_2008_gibbs, cheng_2009_thermodynamics, enders_2009_gibbs, nagle_2010_defense, peters_2013_demonstration, saunders_2018_gibbs, darrigol_2018_gibbs, dieks_2018_gibbs, swendsen_2018_probability, ihnatovych_2023_gibbs, baker_2024_how, lairez_2024_thermostatistics}. These approaches can be further categorized. Philosophical approaches focus on particle identity and the foundation of entropy. For example, \citet{maslov_2008_solution} emphasizes contextual particle distinguishability, while \citet{ihnatovych_2023_gibbs} critiques classical entropy additivity. From an information-theoretic perspective, \citet{lairez_2024_thermostatistics} aligns with Jaynes. Operational approaches resolve the discontinuity pragmatically. \citet{allahverdyan_2006_explanation} replace coarse-graining with ergotropy (extractable work), showing it decreases continuously as gases become indistinguishable. \citet{peters_2013_demonstration} attributes entropy changes to uncertainty in particle locations. Also, in agreement with Jaynes, \citet{yadin_2021_mixing} reconcile observer-dependent work extraction with respect to a single heat bath in a first-quantization description. With a setup similar to the Hong-Ou-Mandel experiment, they argue that Morty can also extract work, contrary to the predictions of the classical GMP.

\textbf{Mixing Beyond GMP.} Mixing has also been studied independently. \citet{yoshida_2022_work} quantify free energy changes in small quantum systems, while \citet{takakura_2019_entropy} explore entropy in generalized probabilistic theories. \citet{sasa_2022_quasistatic} analyze quasistatic work in quantum thermal processes. These works highlight broader thermodynamic implications of mixing beyond classical paradoxes.

\textbf{Observational entropy.} Von Neumann argued that von Neumann entropy does not accurately represent thermodynamic entropy, as it assigns zero entropy to all pure states and fails to capture the information loss that arises when only macroscopic degrees of freedom are observed. As an alternative, he introduced macroscopic entropy~\cite{vonNeumann1929translation,von1955mathematical}, which uses state-counting similar to Boltzmann entropy. This concept has recently gained renewed interest and been generalized to observational entropy~\cite{safranek2019a,safranek2019b,safranek2021brief,SW21,safranek2021generalized,buscemi2022observational,bai2024observational}, leading to a range of applications~\cite{riera2020finite,strasberg2021clausius,sreeram2023witnessing,sreeram2024dichotomy,chakraborty2024sample,nagasawa2024generic,Xuereb2024,meier2025emergence,schindler2025unification}. 

For a projective measurement defining a complete set of orthogonal projectors, known as coarse-graining $\C=\{\P_i\}$, observational entropy is defined as
\[
S_\C=-\sum_i p_i \ln p_i+ \sum_i p_i \ln V_i.
\]
Here, $p_i=\tr [\P_i \R]$ represents the probability of the system being in macrostate $i$, while $V_i=\tr [\P_i]$ denotes the volume of the macrostate, i.e., the number of microstates (pure states) contained within it. Each projector also corresponds to a subspace, defined by its support. Thus, we can equivalently write that a coarse-graining corresponds to a partition of the Hilbert space into subspaces, $\HS=\bigoplus_i \HS_i$, where $V_i=\dim \HS_i$. The first term in observational entropy corresponds to the Shannon entropy, capturing the uncertainty in identifying the macrostate to which the system belongs. The second term represents the mean Boltzmann entropy, quantifying the residual uncertainty in determining the specific microstate due to the coarse-grained nature of the measurement.

\begin{figure}[t!]
    \centering
    \begin{tabular}{c c}
        \rotatebox{90}{~~~quantum} &~
        \includegraphics[width=0.92\linewidth]{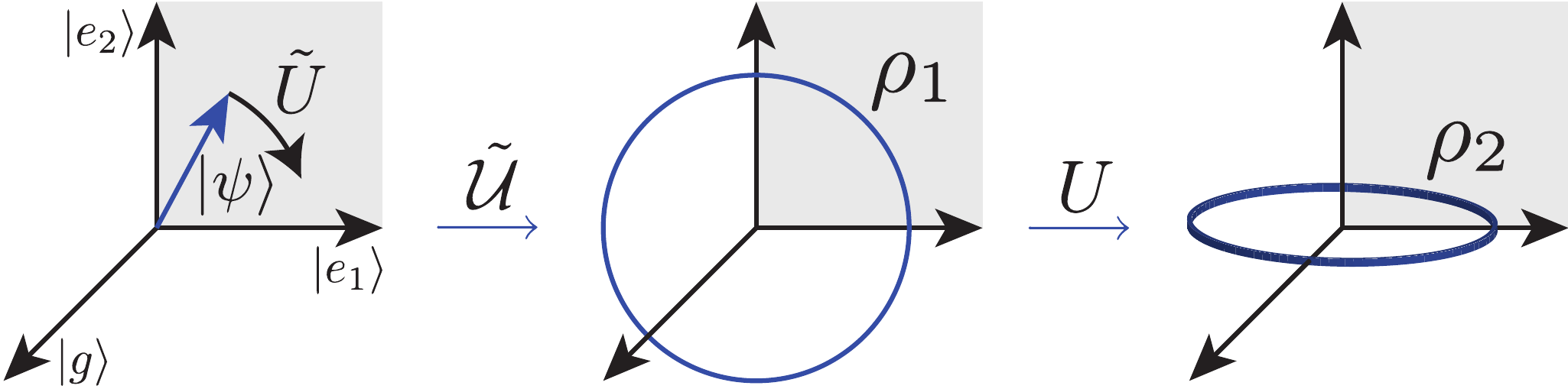}
    \end{tabular}
    
    \vspace{10pt}

    \begin{tabular}{c c}
        \rotatebox{90}{~~~~classical} &~
        \includegraphics[width=0.92\linewidth]{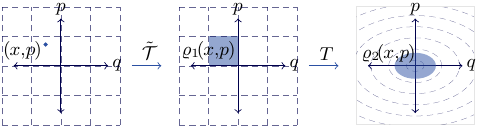}
    \end{tabular}
    
    \caption{Sketch of the extraction protocol in quantum (top panels) and classical (bottom panels) isolated systems. First, a coarse measurement is performed, which provides probability of finding the system in a given macrostate. In the quantum case, we sketch the Hilbert space with the measurement basis given by a set of projectors. In the example, we have chosen the case where the observer cannot distinguish the two excited state from each other but can distinguish them from the ground state, corresponding to coarse-graining
    $\C=\{\ketbra{g}{g}, \ketbra{e_1}{e_1}+\ketbra{e_2}{e_2}\}$. The first projector projects onto the x axis and the second to the y-z plane. We start from an initial state $\ket{\psi}=\frac{1}{\sqrt{3}}\ket{e_1}+\sqrt{\frac{2}{3}}\ket{e_2}$ which is unknown to the observer, represented by a vector lying in y-z plane. First, a random unitary $\tilde U$ is applied, which effectively spreads this state over a macrostate (i.e., on average, as described by operation $\mathcal{\tilde U}=\int\tilde{U}\bullet \tilde{U}^\dag d\mu(\tilde U)$). The resulting state $\R_1=\frac{1}{2}(\ketbra{e_1}{e_1}+\ketbra{e_2}{e_2})$ is known to the observer, represented by a circle in the y-z plane. Finally, an extraction unitary $U$ is applied, which lowers the energy of the state, extracting it in the process, yielding $
    \R_2=\frac{1}{2}(\ketbra{g}{g}+\ketbra{e_1}{e_1})$ (represented by a circle lying in the x-y plane). In the classical case, we sketch the phase space with the coarse-graining given by dividing it into cells (dashed grid). Starting from some unknown initial state, given by a point in phase space (blue dot), a random canonical transformation is applied, spreading the state over the macrostate (blue box). Then, a canonical extraction transformation is applied which reduces its energy (blue elliptic disk). This represents the example of a classical harmonic oscillator for which the energy shells form ellipses, with the smallest energy concentrated in the middle.}
    \label{fig:extraction}
\end{figure}

\section{Difference in the work extracted between two observers}\label{sec:energydiff}

In this section, we derive the difference in work extracted by two observers with different observational capabilities.

Consider two observers, Rick and Morty, with two measuring devices of different precision. These measuring devices correspond to coarse-grainings $\CR$ and $\CM$, We define the difference in observational entropies as
\[\label{eq:deltaS}
\Delta S\equiv S_\CM-S_{\CR},
\]
We assume that the measurement of Rick has a higher precision as measured by lower observational entropy, meaning that $\Delta S\geq 0$. For what follows, we require this assumption only for the particular quantum state from which the observers will extract energy, while for other states this inequality may be reversed. This will make all other quantities of type $\Delta$ introduced below positive. 

The work extracted by performing a joint extraction on a large number of copies of the system, which is characterized by the measurement probabilities in a single coarse-grained basis $\C$, is given by observational ergotropy~\cite{safranek2023work},
\[\label{eq:extracted_work}
W_\C^{\infty}=\tr[\ham(\R-\R_{\beta})].
\]
There, $\R_{\beta}=\exp(-\beta \ham)/Z$ is the thermal state with inverse temperature $\beta=1/kT$ implicitly defined by the von Neumann entropy, as a solution to $S_{\vN}(\R_{\beta})=S_\C$, and $k$ is the Boltzmann constant. We call temperature $T$ obtained by solving this equation an \emph{observational temperature}. Expression $\tr[\ham\R]$ represents the initial energy of the state from which the energy is extracted, while $\tr[\ham\R_{\beta}]$ is the average energy of the final state after the extraction protocol is employed.

\begin{table*}
    \centering
    \begin{tabular}{l|c|c}
        Type of system & Heat capacity & Work difference in units of $kT\Delta S$\\
        \hline
        Ideal gas & $\frac{3}{2}N k$ & $1-\frac{1}{3}\frac{\Delta S}{N}+\frac{2}{27}\left(\frac{\Delta S}{N}\right)^2$ \\
        % Liquid Helium near the lambda point & $\begin{cases}
        %     \frac{A^-}{\alpha} t^{-\alpha}+B^-, & T<T_\lambda\\
        %     \frac{A^+}{\alpha} \abs{t}^{-\alpha}+B^-, & T>T_\lambda
        % \end{cases}$  & \\
        Debye model at low temperatures & $324 N k \left(\frac{T}{T_D}\right)^3$   & $1-\frac{1}{2}\frac{\Delta S }{324 N }\frac{T_D^3}{T^3}-\frac{1}{3}\left(\frac{\Delta S }{324 N }\frac{T_D^3}{T^3}\right)^2$ \\
        Quantum Critical Metals & $\frac{\pi}{6} m N  T \left(1+\frac{A}{N T^{1/3}}\right)$   & $1-\frac{3 k}{\pi m}\frac{\Delta S}{N T\left(1+\frac{A}{N T^{1/3}}\right)}+\frac{2k^2}{\pi^2 m^2\left(1+\frac{N T^{1/3}}{A}\right)}\left(\frac{\Delta S}{N T\left(1+\frac{A}{N T^{1/3}}\right)}\right)^2$ \\
        Liquid Helium above the lambda point &  $N k (A+B t^{-\alpha})$ & $1-\frac{1}{2}\frac{\Delta S}{N(A+B t^{-\alpha})}+\frac{1}{6}\left(1+\frac{(1+t) \alpha}{(1+\frac{A}{B}t^{\alpha})t}\right)\left(\frac{\Delta S}{N(A+B t^{-\alpha})}\right)^2$ \\
        s-wave superconductor at low temperatures & $A N k \sqrt{T}e^{-\Delta/kT}$ & $1-\frac{1}{2}\frac{\Delta S e^{\Delta/kT}}{A N \sqrt{T}}+\frac{1}{12}\left(1-\frac{2\Delta}{k T}\right)\left(\frac{\Delta S e^{\Delta/kT}}{A N \sqrt{T}}\right)^2$ \\
    \end{tabular}
    \caption{Work difference as per Eq.~\ref{eq:energy_diff2} in units of $kT\Delta S$ for a variety of classical and quantum gases. $A$, $B$, $\Delta$, $m$ and $T_D$ are constants that depend on the specific model and $T$ denotes observational temperature. For the liquid helium we defined $t=T/T_c - 1$ with $T> T_c$.     For liquid Helium, the second factor diverges as the observational temperature approaches the critical temperature. For others, the expansion diverges as temperature converges to absolute zero. As a result, the expansion holds only for very small entropy differences around these critical points. Additionally, this indicates that there is a critical amount of information, given by $S_
    \C=S_{\vN}(\R_{\beta_c})$, which the observer is required to possess in order to extract significantly more energy.}
    \label{tab:expansions}
\end{table*}

The work extraction protocol is partially random, meaning that one first employs a randomizing unitary which washes away the information about the degrees of freedom that cannot be obtained from the coarse measurement $\C$, and then the global work extraction unitary that is optimized for the average state. In the language of quantum reference frames, the first step is known as ``twirling'' and encodes the lack of a quantum reference frame~\cite{Bartlett2007}. Observational ergotropy, Eq.~\eqref{eq:extracted_work}, then represents the average extracted work upon repetition of this protocol many times. See Fig.~\ref{fig:extraction} for the sketch of this extraction protocol. For more details see Refs.~\cite{safranek2023work,safranek2023ergotropic,chakraborty2024sample}.

The difference in extracted energies by the two observers is equal to,
\[\label{eq:work_difference}
\Delta W \equiv W_{\CR}^{\infty} - W_{\CM}^{\infty}=\tr[\ham(\R_{\bM}-\R_{\bR})],
\]
where $S_{\vN}(\R_{\bR})=S_\CR$ and $S_{\vN}(\R_{\bM})=S_{\CM}$.

This formula holds completely generally, for any two coarse-grainings. The positive difference in entropies $\Delta S$ implies that work extracted $\Delta W$ is also positive.

The formula for the observational ergotropy fundamentally depends on the initial energy, which is unknowable to the observers given that they obtain only limited information by measuring in their basis, $\CR$ and $\CM$, respectively. It can, however, be estimated from these measurements~\cite{safranek_2023_measuring}. On the other hand, the formula for the difference no longer depends on this initial energy, only on the respective observational entropies, which allows us to study it in more detail.

Note that in some scenarios, especially when the initial state energy is small, execution of the protocol may cost more energy than is being extracted for both observers. The positive difference in extracted work $\Delta W$ then means that Rick has lost less energy than Morty by executing the protocol.

\subsection{Different observers with similar capabilities}\label{sec:energydiffsim}

We show that when the difference in observational entropies is small, the extracted work difference is given by an expression that strikingly resembles Landauer's principle.

When $\Delta S$ is small, then the difference between corresponding temperatures are small, too. We define inverse observational temperatures of Rick and Morty implicitly via $S_{\vN}(\R_{\bR})=S_\CR$ and $S_{\vN}(\R_{\bM})=S_{\CM}$, respectively. The temperature difference
\[\label{eq:temperature_difference}
\Delta \beta \equiv \bR-\bM,
\]
is positive for positive $\Delta S$, implying $T_R\leq T_M\equiv T$ (meaning that $T$ is the observational temperature related to the higher of the two observational entropies). Further, the temperature difference vanishes as $\Delta S$ vanishes.

We derive the work difference as an expansion in small $\Delta S$~\footnote{In the example of the mixing paradox below, in the initial situation of the macroscopic non-equilibirum state with equal-sized boxes and non-dilute scenario, $\Delta S$ is about a third of $S$, so this condition is not satisfied. Including higher order terms will yield better agreement with the exact formula, Eq.~\eqref{eq:work_difference}, and we employ empirical approaches to do that. For the unequal sized boxes, or sufficiently dilute gas, this condition will generally be satisfied.}. To do that, we first insert the identities for observational entropy below Eq.~\eqref{eq:work_difference} into Eq.~\eqref{eq:deltaS} and expand it in the powers of $\Delta \beta$ up to the third order coefficient. Then, we solve this equation for $\Delta \beta$ while ignoring the higher order terms, yielding an expansion $\Delta \beta(\Delta S)$. Second, we expand the work difference, Eq.~\eqref{eq:work_difference}, in terms of $\Delta \beta$. 

Inserting the first expansion into the second, we obtain an expression for the ergotropy difference,
\[\label{eq:energy_diff1}
\Delta W\!=\!kT\Delta S\!\left(1\!-\!\frac{(kT)^2}{2v}\Delta S\!+\!\frac{(3kT\!+\!X)(kT)^3}{6v^2}\Delta S^2\!+\ccdots\right)\!.
\]
Here, $v=-\frac{d\mean{E}}{d\beta}=\mean{E^2}-\mean{E}^2$ is the variance in energy of the thermal state corresponding to temperature $\beta$, and $X=\frac{d \ln v}{d\beta}=(-\mean{E^3}+3\mean{E}\mean{E^2}-2\mean{E}^3)/v$, a term which encodes the third moment in energy.

Realizing that the variance is related to the heat capacity $\CE=\partial E/\partial T$ through $v=-\frac{d\mean{E}}{d \beta}=\CE k T^2$, we obtain a more physically-relevant expression for the ergotropy difference,
\[\label{eq:energy_diff2}
\Delta W\!\!=\!kT\Delta S\!\left(\!1\!-\!\frac{1}{2}\frac{\Delta S}{\tCE}\!+\!\frac{1}{6}\!\left(1\!-\!\frac{\partial \ln \tCE}{\partial T} T\right)\!\left(\frac{\Delta S}{\tCE}\right)^{\!2}\!\!\!+\ccdots\right)\!.
\]
This represents our main result. $\tCE=\CE/k$ is the dimensionless heat capacity, which measures the total number of degrees of freedom that contribute to the heat capacity.

The zeroth order term strikingly resembles Landauer's principle about the minimal required energy to erase a bit of data. There are a few notable differences: 

First, this formula represents the maximal extractable work given the information at hand, which is quantified by observational entropy. The work is extracted by applying the extraction unitary, which is dictated by the measurement outcome. This is similar to the Szilard engine. Also there, the position of the particle is determined by the measurement. Based on the outcome, the weight is attached to the piston at the correct side and is lifted. The maximal extracted work for the Szilard engine has been extensively studied~\footnote{See reviews \cite{sagawa_2012_thermodynamics,parrondo_2014_thermodynamics,davies_2021_the} and papers for quantum~\cite{Berut_2012_experimental,koski_2014_experimental,peterson_2016_experimental,boyd_2016_maxwell,boyd_2017_transient,masuyama_2018_information,vanhorne_2020_single-atom} (Maxwell's Demon, Szilard Engine and Landauer Principle), classical~\cite{marathe_2010_cooling,tamir2018experimental,vaikuntanathan_2011_modelling}, and active matter \cite{sosuke_2013_information,malgaretti_2022_szilard,chor_2023_many-body} systems.}.

The difference to the Szilard engine is that this result applies to isolated systems, while the Szilard engine operates at a fixed temperature of the surrounding thermal bath $T$. Here, in contrast, the temperature $T$ is not a physical temperature of a thermal bath, but rather the ``informational'' temperature corresponding to the observational entropy. To be more precise, we identified the temperature $T$ as the higher of the two observational temperatures, see Eq.~\eqref{eq:temperature_difference}. Given that we assume the entropy difference $\Delta S$ to be small, any difference between the observational temperatures remains small as well. 

Physically, the observational temperature is the temperature of the resulting final thermal state when the randomized extraction protocol is employed. As already mentioned above, if the mean energy of such a thermal state is higher than the initial mean energy, employing the protocol results in a net energy investment. If it is lower, energy is extracted. A positive ergotropy difference then means that Rick has gained $\Delta W$ amount more energy (or lost less energy) than Morty when employing the protocol.

The first correction term is always negative, meaning that $k T \Delta S$ bounds the work difference $\Delta W$ from above. This term also implies that the work difference is the largest when $(kT)^2\ll v$, meaning that energy fluctuations are much larger than the thermal energy scale. Equivalently, these are the situations when the heat capacity grows very large or diverges. This happens, for example, for systems near a critical point or in a low-temperature regime with large degeneracy~\cite{yang2001low,lipa2003specific,lashley2003critical,liang_2015_heat,zhang2023free}.

The second-order correction can be both positive or negative, and it will generally depend on how the heat capacity scales with temperature. Considering power law $\tCE \propto T^\alpha$, the correction is negative for a superlinear scaling $\alpha > 1$, meaning that the number of degrees of freedom contributing to the internal energy grows faster than linearly with temperature. It is positive otherwise. A few examples of classical and quantum gases with varying heat capacities are illustrated in Table~\ref{tab:expansions}.

Finally, note that there was not much particularly quantum about this calculation, apart from the unitary extraction when employing the protocol. Observational entropy can be consistently defined also for classical systems~\cite{safranek2020classical}. Furthermore, classical equivalents of unitary operations are symplectic (canonical) transformations, which preserve all phase-space volume, the symplectic product between the two pase-space vectors, and the relationship between conjugate variables (the Poisson brackets), sketched in Fig.~\ref{fig:extraction}. One then infers that this result also applies identically to classical isolated systems of indistinguishable particles.

\subsection{Work extracted at different times}\label{sec:difference_in_times}

Equivalent calculations can also be performed for the same observer who extracts work at different times. Consider an initial state $\R_0$, describing, for example, particles not being mixed yet. After some time, the system evolves into state $\R_1$, e.g., a state where particles have mixed already. The ergotropy difference can, in that case, also describe how much the same observer can extract at different times. Defining two observational entropies and the corresponding thermal state $S_{\vN}(\R_{\beta_i})=S_\C(\R_i)$, $i=0,1$, the ergotropy difference is given by $\Delta W=\tr[\ham(\R_{\beta_1}-\R_{\beta_0})]$. The same approximate expression, Eq.~\eqref{eq:energy_diff2}, also holds.

This indicates that an observer at an earlier time, when the particles are separated, can extract $\Delta W$ more energy than at a later time when all particles are mixed. Then, one could argue, it must cost \emph{at least as much energy} to separate the two gases, meaning that the ergotropy difference also gives the lower bound on the amount of work the observer needs to perform to macroscopically separate these gases. This would hold generally, not only for mixing gases: to undo the macroscopic changes observed by an observer with an ability to do a macroscopic measurement $\C$, they must spend at least $\Delta W$ amount of energy to do so.

While this is intuitively reasonable, there is a caveat: the difference in extracted work assumes the same protocol is applied at both times. However, a different protocol might reverse the macroscopic changes at a lower cost. In principle, such a protocol must exist---if the full quantum state were known, the reversal could be performed without extra energy. Morally, we would argue that with limited information, no protocol should achieve this reversal at an average cost lower than $\Delta W$. A rigorous proof is left for future work.

\section{Observational entropy of quantum mixing}\label{sec:mixingentropy}

We derive the entropy difference $\Delta S$, which we can then input into the formula for the ergotropy difference, for a quantum version of the Gibbs mixing paradox. 

We consider observers illustrated in Fig.~\ref{fig:rick_oneD_main}: Rick, who is able to determine the total number of particles of each color in both the left and the right side of the box; and three clones of Morty, who can identify the total number of particles in each part of the box but not their color, but have varying degrees of knowledge about colors of particles in the full system. Morty 1: knows the total number of blue particles in the full system; Morty 2: does not know the total number of particles in the full system, but is aware that these two types exist; and Morty 3: is not even aware of the existence of the two types of colors and as a result, will associate incorrect volumes to the observed macrostates.

The ability to determine or identify is described by a specific coarse-graining representing a macroscopic measurement that goes into the definition of observational entropy that either of them associates to the observed system.

In the following calculations, we consider both the blue and the red particles to be either fermions or hard-core bosons, which will affect the counting, however, a very similar derivation can also be done with bosons. The results for these will coincide when the gas is sufficiently dilute.

We consider the total Hilbert space to be that of any possible configuration of any number of blue and red particles on the lattice of $L$ sites -- i.e., spatial modes in second quantization. Here, color could be a spin degree of freedom (immutable within the observers capability; evolved with a spin-preserving Hamiltonian). The actual size of the total Hilbert space will not actually matter, as long as every state that we consider can be contained within it. For instance, when considering Hamiltonians that preserve the total number of particles of each type, we can restrict to a Hilbert space that respects this symmetry. However, both of these approaches will give the same results since these come from counting the dimensions of the macrostates and not the dimension of the full Hilbert space.

For simplicity, we assume that both sides are of the same lattice size, $w$. The non-symmetric case, of $L_A$ and $L_B$ being the size of the left and the right side, respectively, can be obtained by substituting $\ln w \rightarrow a\ln L_A+(1-a) \ln L_B$ into the results below ($a$ is defined below).

\begin{figure}[t!]
    \centering
    \includegraphics[width=0.8\linewidth]{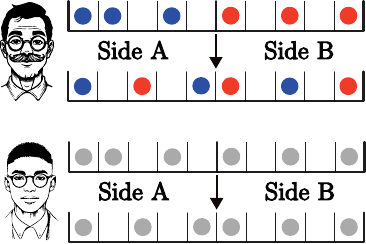}
    \caption{The schematic for counting the states for Rick. In this illustration, the number of sites on each side is $\omega=5$. Following the notation in the text, we can see that the number of blue particles is $N_+ = 3$, the number of red particles is $N_-=3$, the number of particles on the left side is $N_A=3$, and the number of particles on the right $N_B=3$. Colors mix as the system evolves and Rick observes an increase in entropy.}
    \label{fig:rick_oneD_main}
\end{figure}

\subsection{Rick}\label{sec:mixingentropyrick}

Rick can observe the total number of particles of each color on both the left and the right side of the box; see Fig.~\ref{fig:rick_oneD_main} (top panel). The non-symmetric case is derived in Appendix~\ref{app:rick_nonsymmetric}.

Rick's macrostate is given by four numbers, the total number of blue and red particles, on the left and on the right side. This macrostate is mathematically described by a subspace, 
\[\label{eq:HS}
\HS^{\mathrm{R}}=\underset{\blueleft}{\HS_+^A}\otimes \underset{\redleft}{\HS_-^A}\otimes\underset{\blueright}{\HS_+^B}\otimes \underset{\redright}{\HS_-^B}\ .
\]
The Hilbert spaces on the right correspond to blue (+) particles on the left side, red (-) on the left, blue on the right, and red on the right, from left to right. 

To simplify calculations, we introduce an equivalent set of different numbers: $i$, the number of blue particles on the left side; $N_+$, the total number of blue particles; $N_A$, the total number of particles on the left; and $N$, the total number of particles. From these, the total number of red particles, $N_-$, and the total number of particles on the right, $N_B$, is determined by $N_+ +N_-=N_A+N_B=N$.

The volume of this macrostate, parametrized by the tuple $(i,N_+,N_A,N)$, is given by
\[\label{eq:volume_rick}
V^{\mathrm{R}}=\underset{\blueleft}{\mqty(w\\i)}\underset{\redleft}{\mqty(w\\N_A-i)}\underset{\blueright}{\mqty(w\\N_+-i)}\underset{\redright}{\mqty(w\\N-N_+-(N_A-i))}.
\]
The binomial coefficients are the dimensions of the respective Hilbert spaces in Eq.~\eqref{eq:HS}, from left to right. We also denote this volume by a subscript ``init'' since we assume that when modelling time evolution, the system starts fully in this macrostate.

Assuming a diluted case, in which $w$ is significantly larger than other variables, we can apply Stirling's approximation, as detailed in Appendix~\ref{app:rick}. We derive
\[\label{eq:rick_final_main}
\ln V^{\mathrm{R}}
=N(\ln w-\ln N+S(q)+1)-\frac{1}{2}\ln \Big((2\pi N)^4 \prod_{j=1}^4 q_j \Big).
\]
Here, $q_1=p_i$, $q_2=a-p_i$, $q_3=r-p_i$, and $q_4=1-r-a+p_i$, $S(q)=-\sum_i q_i\ln q_i$ denotes the Shannon entropy. We defined $p_i=i/N$ as the proportion of blue particles on the left (out of all particles). We also define proportion of all particles on the left, $a=N_A/N$, and proportion of the blue particles out of all $r=N_+/N$. 

The last term of Eq.~\ref{eq:rick_final_main} represents logarithmic corrections. In those, the possible divergences when some $q_j=0$ result from the use of Stirling's approximation. Upon closer inspection of these boundary cases one finds that the general formula can still be used by subtracting $-\frac{1}{2}\ln(2\pi N q_j)$ for each $q_j=0$, see the end of Appendix~\ref{app:rick}.

\subsection{Morty 1}\label{sec:mixingentropymorty1}

Morty 1 can observe three macroscopic degrees of freedom: the total number of particles on the left $N_A$, the total number of particles on the right $N-N_A$, and the total number of blue particles, $N_+$.

We again have $N_+ +N_-=N_A+N_B=N$. Without loss of generality, we define $N_A$ and $N_+$ to be the smaller of the two, i.e., $N_A\leq N_B$ and $N_+\leq N_-$. This guarantees that $N\geq N_A+N_+$ for the expression to be well defined.
\[\label{eq:HSmorty}
\HS^{\mathrm{M}_1}=\bigoplus_{i=0}^{\min\{N_+,N_A\}}\HS^{\mathrm{R}}.
\]

The volume corresponding this macrostate, parametrized by a tuple $(N_+,N_A,N)$, is given by the sum over all smaller macrostates (of Rick) that lead to the same macroscopic observation for Morty,
\[\label{eq:general_volume_morty2_main}
V^{\mathrm{M}_1}=\sum_{i=0}^{\min\{N_+,N_A\}}V^{\mathrm{R}}
\]

Assuming the dilute limit, when taking the logarithm, the dominant term will be only the largest term of the sum, while all others will contribute only as a logarithmic correction. To compute this correction, we use the Laplace method which involves approximating the binomials with a Gaussian around the maximum and approximating the sum with an integral. Finally, we obtain
\[
\begin{split}
    \ln V^{\mathrm{M}_1}&=N\left(\ln w-\ln N +S(a)+S(r)+1\right)\\
    &-\frac{3}{2}\ln(2\pi N).
\end{split}
\]
Recalling definitions $a=N_A/N$ and $r=N_+/N$, $S(a)=-a\ln a -(1-a)\ln (1-a)$ is the binary Shannon entropy.

See Appendix~\ref{app:morty} for the derivation and Appendix~\ref{app:morty1_nonsymmetric} for the non-symmetric case.

\subsection{Morty 2}\label{sec:mixingentropymorty2}

Morty 2 can observe two macroscopic degrees of freedom: the total number of particles on the left, $N_A$ and the total number of particles on the right, $N-N_A$. However, he is aware that two types of particles exist and could be in the system.

His macrostates are parametrized by a couple $(N_A,N)$. They are given by subspaces with corresponding volumes defined as their dimension,
\[\label{eq:HSmorty_half_ignorant}
\begin{split}
\HS^{\mathrm{M}_2}&=\bigoplus_{N_+=0}^{N}\HS^{\mathrm{M}_1},\\
V^{\mathrm{M}_2}&=\sum_{N_+=0}^{N}V^{\mathrm{M}_1}
=2\sum_{N_+=0}^{N/2}\sum_{i=0}^{\min\{N_+,N_A\}}V^{\mathrm{R}}.
\end{split}
\]
The right hand side of the volumes was split into two parts, due to the assumption of $N_+\leq N_-$ in Eq.~\eqref{eq:HSmorty}.

Using the Laplace method on two variables, $i$ and $N_+$, we obtain
\[
\begin{split}
\ln V^{\mathrm{M}_2}&=N\left(\ln w-\ln N +S(a)+\ln 2+1\right)-\ln(4\pi N).
\end{split}
\]
(See Appendix~\ref{app:morty2}.) One can observe that up to the logarithmic correction, the leading order is the same as for Morty 1, with $r=1/2$.

\subsection{Morty 3 - the ignorant}\label{sec:mixingentropymorty3}

Finally, we consider the case of ``ignorant'' Morty, who not only cannot distinguish the two colors of particles, but is also not aware that these two colors of particles even exist. As a result, he also does not know how many particles of each type are in the system that he is observing. The only information he can access is the number of particles on the left and right sides, respectively.

This means that when calculating the associated observational entropy, he will assign different volumes to the macrostates than those corresponding to the actual number of microstates they contain, as given by Eq.~\eqref{eq:HSmorty_half_ignorant}.

From his perspective, the macrostate corresponding to having 
$N_A$ particles on the left and 
$N-N_A$ particles on the right will have a volume given by
\[
V^{\mathrm{M}_3}=\underset{\blackleft}{\mqty(w\\N_A)}\underset{\blackright}{\mqty(w\\N-N_A)}.
\]
We obtain
\[\label{eq:morty3entropy}
\begin{split}
\ln V^{\mathrm{M}_3} &= N (\ln w-\ln N +S(a)+1)\\
&-\frac{1}{2}\ln ((2\pi N)^2 a(1-a)).
\end{split}
\]
(See Appendix~\ref{app:morty3}.) The macrostate volume associated to the same observation is smaller than that of Morty 1 and Morty 2, by a factor of $N S(r)$ and $N \ln 2$, respectively. This is because Morty 3 has an incorrect judgement on the underlying physical Hilbert space. It is not because his uncertainty \emph{is}, in reality, smaller, but because he thinks it is smaller. The factors are precisely those related to the existence of two types of particles, which cannot be distinguished locally. In the next section, we show, however, that under some reasonable assumptions, he will extract the same amount of energy as Morty 2 does when he executes the extraction protocol with his limited information. This amount is further consistent with his computed value of entropy. For ease of comparison, see Table \ref{tab:observers} listing the abilities of each of the observers. 
\begin{table*}[ht!]
    \centering
    \begin{tabular}{l|c|c|l}
        Observer & Can access & Cannot access & Hilbert space\\
        \hline
        & & & \\
        Rick & $i$, $N_+$, $N_A$, $N_B$ & --- & $\HS^\mathrm{R} = \HS_+^A\otimes\HS_-^A\otimes\HS_+^B\otimes\HS_-^B$ \\
        & & & \\
        Morty 1 & $N_+$, $N_A$, $N_B$   & $i$ &  $\HS^{\mathrm{M_1}} = \bigoplus_{i=0}^{\min\{N_+,N_A\}}\HS^{\mathrm{R}}$\\
        & & & \\
        Morty 2 & $N_A$, $N_B$, $^\ast$   & $i$, $N_+$ &  $\HS^{\mathrm{M}_2}=\bigoplus_{N_+=0}^{N}\HS^{\mathrm{M}_1}$\\
        & & & \\
        Morty 3 &  $N_A$, $N_B$ & \begin{tabular}{c}
             $i$, $N_+$   \\
             (not aware) 
        \end{tabular} & $\HS^{\mathrm{M}_3}=
        \begin{cases}
        \HS^{A}\otimes \HS^{B}, &\text{(perceived)}\\
        \HS^{\mathrm{M}_2},&\text{(actual)}
        \end{cases}$ \\
    \end{tabular}
    \caption{Different type of observer for the purpose of this work. We have tabulated the various macroscopic observables accessible (or inaccessible) to each of the observers, and their effective Hilbert spaces. Use of $^\ast$ besides the observables of Morty 2 is to remind the reader that Morty 2 has knowledge of the existence of two colors, but he is unable to measure them (see the text in~\ref{sec:mixingentropyrick}-\ref{sec:mixingentropymorty3}). The labels denote: $i$ the number of blue particles on the left; $N_+$ the total number of blue particles; $N_A$ and $N_B$, the total number of particles on the left and the right, respectively.}
    \label{tab:observers}
\end{table*}

\subsection{Entropy difference between observers}\label{sec:entropydifference}

Consider a state which is contained within one of the macrostates of Rick, i.e., it has a definite macro-observable values $(i,N_+,N_A,N)$. For such a state, only one of the probabilities coming into the definition of observational entropy is equal to one while all others are zero. The same holds for Morty. Consequently, ignoring logarithmic terms, for Morty 1 we have
\[
\begin{split}
\Delta S_1&=S_{\CM_1}-S_{\CR}=\ln V^{\mathrm{M}_1}-\ln V^{\mathrm{R}}\\
&=N(S(a)+S(r)-S(q)).
\end{split}
\]
For Morty 2, we have
\[
\begin{split}
\Delta S_2&=S_{\CM_2}-S_{\CR}=\ln V^{\mathrm{M}_2}-\ln V^{\mathrm{R}}\\
&=N(S(a)+\ln 2-S(q)).
\end{split}
\]

When there is the same number of blue and red particles, $r=1/2$, while all blue are on the left and all red are on the right side, $a=p_i=1/2$, we obtain the long-expected result
\[\label{eq:symmetric_ent_diff}
\Delta S_1 = \Delta S_2 = N \ln 2,
\]
which holds for both Morty 1 and 2. Later, we will insert these entropy differences into the formula for ergotropy difference, Eq.~\eqref{eq:energy_diff2}.

It would not be fair to compare Morty 3’s entropy with Rick’s, who clearly has a deeper understanding of the system. Rick may assign a higher entropy to the same state, but only because he accounts for internal degrees of freedom that Morty 3, due to his incorrect judgment, does not recognize. As noted above and elaborated below, this discrepancy does not prevent Morty 3 from correctly predicting the amount of extractable work---though it will generally be lower than what Rick could extract, given his fuller knowledge.

\subsection{Difference in entropy over time}\label{sec:timeevolution}

In the following, we consider the initial state to be in one of the macrostates of Rick having a definite macro-observable values $(i,N_+,N_A,N)$. Further, we assume that the Hamiltonian preserves the total number of particles of each color and that the evolution will cause the particles to uniformly spread.

\textbf{Rick. } Assuming that the evolution will cause the particles to uniformly spread, in the long-time limit, the particles will uniformly fill the entire accessible Hilbert space, with probabilities proportional to the size of each macrostate. Mathematically, this means $p_l^{\mathrm{R}}=V_{l}^{\mathrm{R}}/V$, where $l=(i,N_+,N_A,N)$ denotes the label of each macrostate, $V_{l}^{\mathrm{R}}$ is defined in Eq.~\eqref{eq:volume_rick}, and the dimension of the entire accessible Hilbert space is,
\[
V=\mqty(2w\\N_+)\mqty(2w\\N-N_+).
\]
This means that the final observational entropy equals
\[\label{eq:final_entropy_rick}
\begin{split}
S_{\C_R}^{\mathrm{fin}}&=-\sum_l(V_{l}^{\mathrm{R}}/V)\ln\frac{V_{l}^{\mathrm{R}}/V}{V_{l}^{\mathrm{R}}}=\ln V\\
&= N(\ln 2w-\ln N+S(r)+1)-\frac{1}{2}\ln (2\pi N)^2r(1\!-\!r).
\end{split}
\]

Ignoring the logarithmic terms, this gives us the difference between the final and the initial entropy,
\[
\Delta S^{\mathrm{R}}=S_{\C_R}^{\mathrm{fin}}-S_{\C_R}=\ln \frac{V}{V^{\mathrm{R}}}= N(\ln 2+ S(r)-S(q)).
\]

For the standard scenario of Gibbs mixing paradox depicted in Fig.~\ref{fig:rick_oneD_main}, in which all blue particles start on the left, while the same number of the red particles starts on the right, corresponding to the initial configuration $p_i=1$ and $r=a=1/2$, we have $S(q)=\ln 2$ and $S(r)=\ln 2$. This means that Rick will observe growth in entropy
\[\label{eq:ricks_w_diff}
\Delta S^{\mathrm{R}}= N \ln 2,
\]
corresponding to the standard entropy of mixing, as the system thermalizes.

\textbf{Morty. } Assuming that the evolution will cause the particles to uniformly spread, for all clones of Morty 1 and 2, in the long-time limit the final entropy will be the same as Rick's, Eq.~\eqref{eq:final_entropy_rick}, since they share the same Hilbert space. For Morty 3, it will be given by the reduced Hilbert space with dimension $V_{\mathrm{red.}}=\binom{2w}{N}$, 
\[
S_{\CM_3}^{\mathrm{fin}}=\ln V_{\mathrm{red.}}=N(\ln 2w - \ln N+1)-\frac{1}{2}\ln 2\pi N.
\]

Ignoring logarithmic terms, this gives entropy growth,
\[\label{eq:correct_Morty_entropy_increase}
\begin{split}
\Delta S^{\mathrm{M}_1}&=S_{\CM_1}^{\mathrm{fin}}-S_{\CM_1}=\ln \frac{V}{V^{\mathrm{M}_1}}= N(\ln 2-S(a)),\\
\Delta S^{\mathrm{M}_2}&=\ln \frac{V}{V^{\mathrm{M}_2}} = N(\ln 2-S(a)),\\
\Delta S^{\mathrm{M}_3}&=\ln \frac{V_{\mathrm{red.}}}{V^{\mathrm{M}_3}}= N(\ln 2-S(a)).
\end{split}
\]

In the standard Gibbs mixing scenario, the initial state starts with configuration $a=1/2$. This means that none of the clones will not observe a change in entropy,
\[
\Delta S^{\mathrm{M}_1}=\Delta S^{\mathrm{M}_2}=\Delta S^{\mathrm{M}_3}=0.
\]

Thus, irrespective of the correct or incorrect microstate counting, the difference in entropy of the final and the initial state is the same for all. For most applications in thermodynamics it does not matter where the initial--baseline entropy is; all that matters is the entropy difference. Thus, there is fundamentally no Gibbs paradox, considering that Morty 1 and 2, who have a proper understanding of the underlying Hilbert space, and Morty 3, who does not, will observe the same increase in entropy.

Finally, we take a look at the difference in entropy assigned by different versions of Morty,
\[\label{eq:morty_time_independent}
\begin{split}
S_{\CM_1}-S_{\CM_3}&=N S(r),\\
S_{\CM_2}-S_{\CM_3}&=N \ln 2.
\end{split}
\]
For a system that contains only one type of particle, we have $r=0$, Morty 3's judgment becomes correct, and entropies associated by Morty 1 and 3 will coincide, recovering the correct limit. Since the system is conserving the number of particles of each type, the proportion of types $r$ is time-independent. In other words, for the entire time evolution, entropies associated to the system by the clones of Morty differ only by a time-independent value.

\begin{figure}[t!]
    \centering
    \includegraphics[width=1\linewidth]{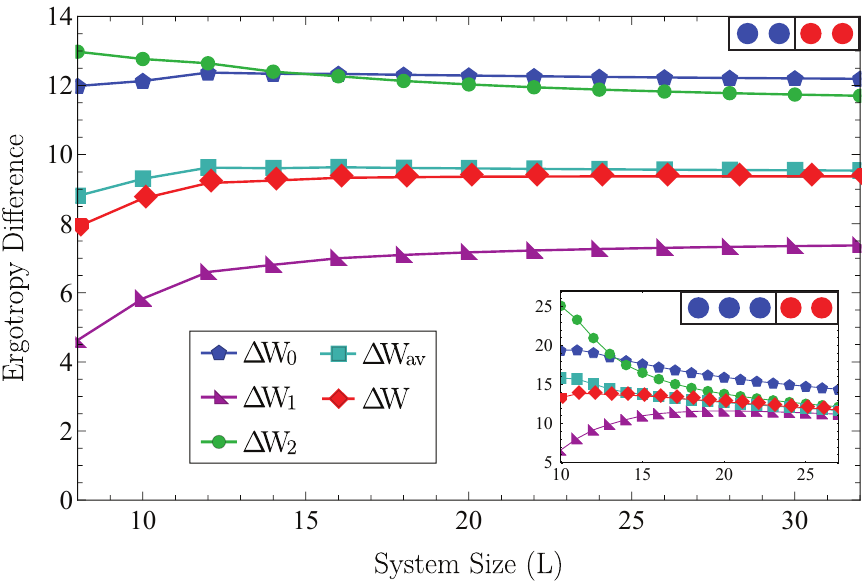}
    \caption{Ergotropy difference between Rick, who can distinguish the particles, and Morty 1, who cannot, for variable system sizes. We fix $k=1$. We compare the zeroth ($\Delta W_0=T\Delta S$), the first ($\Delta W_1=T\Delta S(1+\Delta S/2\CE)$), and the second ($\Delta W_2$, Eq.~\eqref{eq:energy_diff2}) order expansion of ergotropy difference, with the exact value $\Delta W$, Eq.~\eqref{eq:work_difference}. $\Delta W_{\mathrm{av}}=(\Delta W_1+\Delta W_2)/2$ is the average value, which approximates the exact value the best. This is reasonable considering that the alternating signs of the expansion imply that the real value must fall somewhere in between $\Delta W_1$ and $\Delta W_2$. As another heuristic approach, we employ a ``smart'' average by treating the expansion as a geometric series in the higher-order terms. This leads to $\Delta W_{\mathrm{av},2}\equiv \Delta W_2-w_2\frac{q}{1+q}$. The coefficients are defined as $w_1=\Delta W_1-\Delta W_0$, $w_2=\Delta W_2-\Delta W_1$, $q=w2/w1$. We show the case of symmetric box sizes, where $N_A=N_B=2$ and $L_A=L_B=L/2$, and the non-symmetric case (inset), where $N_A=3$, $N_B=2$, $L_A=6,\dots, 17$, $L_B=\lfloor\frac{N_B}{N_A}L_A\rfloor$, and $L=L_A+L_B$. All blue particles are on the left side, while all red ones are on the right. The average value closely matches the exact value. Additionally, the non-symmetric case shows clear convergence, confirming the expansion of the ergotropy difference.}
    \label{fig:diff_static}
\end{figure}

\begin{figure}[t!]
    \centering
    \includegraphics[width=1\linewidth]{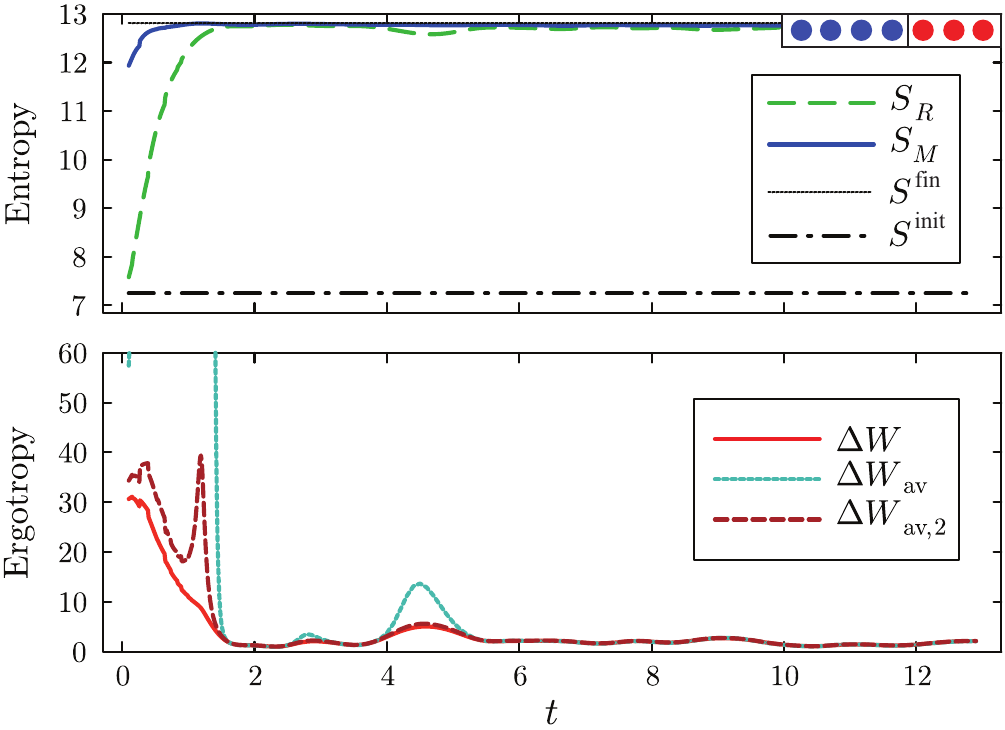}
    \caption{Observational ergotropy and entropy evolution through the mixing of a quantum gas. Evolved with Hamiltonian, Eq.~\eqref{eq:hamiltonian}, and initial configuration $L_A=8$, $L_B=6$, $N_{A,B}=L_{A,B}/2$, with all blue (red) particles starting on the left (right). We plot the exact ergotropy difference and two approximate values: the average defined in Fig.~\eqref{fig:diff_static} and the ``smart'' average, defined in caption \ref{fig:diff_static}. The smart average approximates the the exact value well for small values of $\Delta S$. Notice a small peak (dip) in ergotropy (entropy) around value $t=4.2$, illustrating that the ergotropy is inversely related to the entropy difference (this follows from Eq.~\eqref{eq:extracted_work}). The initial and final entropy values are given by Rick's initial configuration, $S^{\mathrm{init}}=\ln \binom{8}{4}\binom{6}{3}$, and by the dimension of the accessible Hilbert space, $S^{\mathrm{fin}}=\ln \binom{14}{4}\binom{14}{3}$, respectively.}
    \label{fig:diff_evolution}
\end{figure}

\section{Work difference in quantum mixing}\label{sec:workdiffmixing}

Inserting the entropy difference, Eq.~\eqref{eq:symmetric_ent_diff}, into Eq.~\eqref{eq:energy_diff2}, we obtain the difference between the amount of work Rick and Morty can extract in both quantum and classical Gibbs mixing scenario,
\[\label{eq:energy_diff_gibbs}
\Delta W\!\!=\! NkT \ln \!2\!\left(\!1\!-\!\frac{\ln \!2}{2\tilde  c_E}\!+\!\frac{1}{6}\!\left(1\!-\!\frac{\partial \ln \tilde c_E}{\partial T} T\right)\!\left(\frac{\ln \!2}{\tilde c_E}\right)^{\!2}\!\!\!+\ccdots\right)\!.
\]
where $\tilde c_E= \CE/(N k)$ is the dimensionless specific heat. We reiterate that temperature $T$ is not a temperature of a thermal bath. Instead, it is the observational temperature given by $S(\R_T)=S_{\C_M}$. For an ideal gas described in Table~\ref{tab:expansions}, we obtain $\Delta W=NkT\ln 2(1-\frac{1}{3}\ln 2+\frac{2}{27}(\ln 2)^2+\cdots)\approx 0.805\times NkT \ln 2$. 
%{\color{red}Here add the discussion about delta S}

We observe that the formula for Rick's entropy difference at different times, Eq.~\eqref{eq:ricks_w_diff}, is the same as the aforementioned entropy difference.
Following the arguments of Section~\ref{sec:difference_in_times}, this shows that Rick will extract $\Delta W$ less energy if he chooses to perform the extraction at a later time, as opposed to immediately. %This also indicates that it must cost at least $\Delta W$ above for the Rick to macroscopically separate the gases.

While different types of Morty associate different entropy, this difference remains constant throughout the evolution, Eq.~\eqref{eq:morty_time_independent}. Consequently, the increase in entropy will be the same for all of them, Eq.~\eqref{eq:correct_Morty_entropy_increase}. The extracted work difference depends only on this entropy increase. There is no advantage for either Morty 1 or 2; they will extract the same amount of energy.

Morty 3 requires much more care. This is because he lacks a correct understanding of the underlying Hilbert space, resulting from his unawareness of the existence of the two particles, due to his inability to distinguish between different colors (which is his defining feature). This we explained in Section~\ref{sec:mixingentropymorty3}. As a result, he further misjudges the actual unitary he applies to the system when attempting to extract work, misjudges the measurement he performs, and would infer a density matrix from quantum tomography that deviates from the true one. However, in Appendix~\ref{sec:morty3work_extraction} we prove that under some reasonable assumptions, he will extract the same amount of energy as Morty 2. This includes three consistency assumptions: one on the density matrix, stating that the actual density is macroscopically consistent with the one he would infer by performing any measurement that is unable to distinguish the colors of the particles, which we call an admissible measurement. Second, on the measurements, saying that the probability of distribution outcomes he predicts is the same as the actual one when he performs any admissible measurement. Third, on the unitary operations, saying that the actual unitary measurement is consistent with what is expected and detectable by admissible measurements. Furthermore, we have the super-selection rule assumptions, which state that it is not possible to have superpositions between different types of particles. This means that states such as $\ket{0+}+\ket{-0}$ are not allowed, nor can they be created by applying either measurement or unitary operations, or by letting the system evolve naturally under the Hamiltonian. So far, these assumptions have been relatively weak. Finally, we have one non-trivial assumption that will not be true for many systems, namely that the Hamiltonian does not depend on the color of the particles but only on the configuration---describing whether the particle is or is not at a given lattice site, irrespective of its color. This means that every two states with the same configuration, such as $\ket{0+}$ and $\ket{0-}$ corresponding to $\ket{01}$, must have the same energy. This ensures that Morty 3 correctly understands the system's energy, given the constraints of what he is able to measure.

With these assumptions, Morty 3 extracts exactly the same amount of energy as Morty 2. Thus, the difference in extracted work between the two times will be the same.

To conclude, there is no advantage for either clone of Morty. In other words, it does not matter whether Morty knows or does not know whether there are two types of particles in the system rather than just one, or how many blue particles there are in total. In either case, he will extract the same amount of energy. This resolves the Gibbs mixing paradox: while absolute values of associated entropy might be different, there is no difference in terms of what different clones of Morty can practically do.

We provide two simulations of the ergotropy difference: Static, applied only to the initial state which is localized in one of Rick's macrostates but for varying system sizes, see Fig.~\ref{fig:diff_static}; and dynamic, where the system starts in an initial macrostate of Rick but is then evolved with Hamiltonian that preserves the total number of particles of each kind and distributes them uniformly in the long-time limit, see Fig.~\ref{fig:diff_evolution}.

The Hamiltonian is given by
\[\label{eq:hamiltonian}
\begin{split}
    H &= \sum_{i} t_1 c^{\dagger}_{i}c_{i + 1} + v_1 n_{i}n_{i + 1}\\ &+\sum_{i} t_2 c^{\dagger}_{i}c_{i + 2} + v_2 n_{i}n_{i + 2} + \text{h.c.} 
\end{split}
\]
with periodic boundary conditions. $c_i^\dag$ and $c_i$ represent creation and annihiliation operators at site $i$, respectively, and $n_i$ represents the number operator. The numerical computations are done in the chaotic phase, with $t_1 = v_1 = 1$ and $t_2 = v_2 = 0.96$.

To minimize the error, we compute the entropy exactly from the binomial coefficients, Eqs.~\eqref{eq:volume_rick} and~\eqref{eq:general_volume_morty2_main}. Further, we compute the ergotropy difference by inserting this entropy difference into the exact formula, Eq.~\eqref{eq:work_difference}, and then from varying levels of approximation using the perturbative expansion, Eq.~\eqref{eq:energy_diff2}.

\section{Conclusion}\label{sec:conclusion}
In this work, we investigated how two observers—each characterized by a distinct ability to distinguish macrostates—assign entropy to an isolated quantum or classical system, and how this observer-dependent entropy influences the energy they can extract. Focusing on the difference in extractable energy between the two observers, we found that its perturbative expansion closely resembles Landauer’s bound, with higher-order corrections determined by the ratio of entropy to heat capacity. Notably, rather than the usual temperature of a thermal bath, an “observational temperature” emerges, governed by the observers’ knowledge of the system.

In the case of critical systems, we identified a critical amount of knowledge---reached when the observational temperature equals the critical temperature---beyond which the amount of extractable work jumps discontinuously. We further demonstrated that the same formula for the extracted energy difference applies when considering a single observer extracting energy at different times.

We then applied these ideas to the Gibbs mixing paradox, in which one observer can distinguish two types of particles while another cannot. A closer examination revealed three distinct subtypes of the latter observer, each associating a different amount of entropy to the system. 

However, all of them observe the same entropy increase. For the two observers who are aware of the existence of two types of particles, this means that they extract the same amount of work. For the third, who is not aware of their existence, we assume consistency of their description, superselection rules, and independence of the Hamiltonian on different particle types. With these assumptions, we show that they extract exactly the same amount of energy as the second observer, who is aware of the existence of two particle types, but unable to distinghuish them. This resolves the Gibbs mixing paradox as follows. 
Even though an observer who discovers the presence of ``two colors'' of particles experiences a discontinuous jump in the entropy they assign to the system (counterintuitively, knowing more but associating higher entropy!), this new knowledge does not provide a practical advantage for work extraction if the particles remain indistinguishable in their measurements.

We further argued that the loss of extractable energy caused by mixing must also place a lower bound on the energy required to separate the particles, though there are subtleties in reversing the protocol that merit future investigation. Our analysis was restricted to projective measurements; generalizing to POVMs~\cite{buscemi2022observational,teixido2025entropic}, including non-commuting measurements~\cite{majidy2023noncommuting,lasek2024numerical} would be an interesting avenue for follow-up work. While free energy is the relevant quantity in standard mixing scenarios, we showed that “observational ergotropy” plays an analogous role in the isolated systems considered here. It would be worthwhile to explore whether these concepts overlap, potentially through studying local equilibrium states where a combination of the local particle number (as considered here) and the local energy coarse-grainings define a ``dynamical thermodynamic entropy~\cite{safranek_2021_brief}.'' Further, as dynamical thermodynamic entropy corresponds with thermodynamic entropy for systems in equilibrium, it would be interesting to compare the work difference derived from this work and that from the adiabatic expansion of ideal gas. This would serve as a check for the consistency of the derived theory. Another promising direction is to examine Fisher information associated with temperature measurements, given its relation to the energy variance that appears in the perturbative expansion of the extracted work.

Finally, we propose two experimental platforms which allow the manipulation of two conserved, distinguishable particle states to test these ideas in a quantum setting. First, we can use Bose-Einstein condensates (BECs), e.g. of $^{87}\text{Rb}$, where hyperfine states (as ``colors'') are spatially separated by applying a magnetic field to induce Zeeman splitting~\cite{myatt_1997_production, hall_1998_dynamics, wheeler_2004_spontaneous}. The evolution conserves the density of each type, and mixing can be created by modulating magnetic field~\cite{lin_2011_spin}. The second promising platform is a 2D optical lattice with Lithium atoms in hyperfine states to denote $\ket{\uparrow}$ and $\ket{\downarrow}$. Each lattice site holds either no atom, one atom of either spin, or two atoms of opposite spin. Using a spin-removal beam and site-resolved imaging~\cite{greif_2016_siteresolved,parsons_2016_siteresolved}, we can determine the spin number, type, and mixing at each site, enabling the computation of observational entropy and extractable work. To realize the work-extraction protocol in practice would require further modifications---such as developing controls for the random and extraction unitaries---or adapting the protocol to fit experimental constraints.

\medskip\noindent \textbf{Acknowledgments.} B.B., R.K.R.,  and D.\v{S}.~acknowledge the support from the Institute for Basic Science in Korea (IBS-R024-D1). Many thanks to Felix C. Binder for his insightful comments and questions, which have shaped the direction of this research and improved its overall presentation.

\appendix
\section{Derivation of the work difference between observers with similar capabilities}

First, we derive expansion of $\Delta \beta$ in terms of $\Delta S$. We do this by expanding Eq.~\eqref{eq:deltaS} as, 
\[\label{eq:fordeltabeta}
\frac{1}{6}S'''\Delta\beta^3+\frac{1}{2}S''\Delta\beta^2+S'\Delta\beta+\Delta S=0.
\]
There, the coefficients represent the derivative of the von Neumann entropies, $S^{(k)}\equiv\frac{d^kS_{\vN}(\R_{\beta})}{d\beta^k}$. Explicitly, we derive
\begin{align}
S'&=-\beta v\\
S''&=-v-\beta x\\
S'''&=-2x-\beta y,
\end{align}
where 
\begin{align}
v&=-\frac{d\mean{E}}{d\beta}=\mean{E^2}-\mean{E}^2,\\
x&=\frac{dv}{d\beta}=-\mean{E^3}+3\mean{E}\mean{E^2}-2\mean{E}^3,\\
y&=\frac{dx}{d\beta}=\mean{E^4}-4\mean{E}\mean{E^3}-3\mean{E^2}^2\\
&\ \ \ \ \ \ \ \ \;\!+12\mean{E}^2\mean{E^2}-6\mean{E}^4,
\end{align}
are the variance of the thermal state, the first derivative of the variance, and the second derivative, respectively. To derive this, we have used $S_{\vN}(\R_{\beta})=\ln Z +\beta \mean{E}$, $\frac{dZ}{d\beta}=-\mean{E} Z$, and $\frac{d\mean{E^k}}{d\beta}=-\mean{E^{k+1}}+\mean{E}\mean{E^{k}}$.

We rewrite Eq.~\eqref{eq:fordeltabeta} as
\[\label{eq:fordeltabeta2}
\frac{\beta(2X+\beta Y)}{6}b^3+\frac{1+\beta X}{2}b^2+b-\delta=0.
\]
We have defined $b=\Delta\beta/\beta$, $X=x/v$, $Y=y/v$, and $\delta=\Delta S/(v \beta^2)$. We look for the solution of in the small expansion of $\delta$, and derive
\[
\begin{split}
\Delta\beta\!=\!\beta{\delta}\left(1\!-\!\frac{1\!+\!\beta X}{2} {\delta}\!+\!\frac{3\!+\!4\beta X\!+\!3\beta^2 X^2\!-\!\beta^2 Y}{6} {\delta}^2\!+\ccdots\right).\\
\end{split}
\]

Similarly, we expand the work difference in the expansion of $\Delta \beta$ as,
\[
\Delta W=v\Delta\beta\left(1+\frac{1}{2}X\Delta\beta+\frac{1}{6}Y\Delta\beta^2+\cdots\right).
\]
There, we have used $\Delta W=\mean{E}_\beta-\mean{E}_{\bR}$ (note that $\beta\equiv\bM$), where $\mean{E}\equiv\mean{E}_\beta=\tr[\ham \R_\beta]$, and $\mean{E}_{\bR}=\mean{E}-v\Delta\beta-\frac{1}{2}x\Delta\beta^2-\frac{1}{6}y\Delta\beta^3+\ccdots$. Combining the above two equations, we obtain our final result in the expansion of $\delta$,
\[
\Delta W=v\beta \delta\left(1-\frac{\delta}{2}+\frac{3+X\beta}{6}\delta^2+\cdots\right).
\]
Inserting the substitutions back together with $\beta=1/(kT)$, we obtain
\[\label{eq:energy_diffapp}
\Delta W\!=\!kT\Delta S\!\left(1\!-\!\frac{(kT)^2}{2v}\Delta S\!+\!\frac{(3kT\!+\!X)(kT)^3}{6v^2}\Delta S^2\!+\ccdots\right)\!.
\]
which is Eq.~\eqref{eq:energy_diff1}. Note that the final result does not depend on the fourth moment of energy, $Y$, as its contributions cancel out.

The variance of the thermal state is connected to the heat capacity,
\[
v=-\frac{d\mean{E}}{d \beta}=-\frac{\partial \mean{E}}{\partial T}\frac{\partial T}{\partial \beta}=-\CE \frac{\partial \tfrac{1}{k\beta}}{\partial\beta}=\frac{\CE}{k \beta^2}=\CE k T^2,
\]
where $T$ is the temperature, and $\CE=\partial E/\partial T$ is the heat capacity.
Similarly, we derive an expression for $x$ as
\[
x=\frac{dv}{d\beta}=-\frac{\partial(\CE k T^2)}{\partial T} kT^2=-\frac{\partial \CE}{\partial T} (kT^2)^2-2\CE k^2T^3.
\]
This gives
\[
X=\frac{x}{v}=-\frac{1}{\CE}\frac{\partial \CE}{\partial T}kT^2-2 kT=-\frac{\partial \ln (\CE/k)}{\partial T} kT^2-2kT.
\]
Finally, we obtain Eq.~\eqref{eq:energy_diff2},
\[
\Delta W=\!kT\!\Delta S\left(1\!-\!\frac{k\Delta S}{2\CE}\!+\!\frac{1\!-\!\frac{\partial \ln (\CE/k)}{\partial T} T}{6}\left(\frac{k\Delta S}{\CE}\right)^2\!\!\!\!+\ccdots\right)\!.
\]

\section{Volumes of Rick and Morty macrostates}

\subsection{Rick}\label{app:rick}

Rick can observe both the number and type of particle in each part. His macrostate is given by four numbers: $i$, denoting the number of blue particles on the left; $ N_+$, the total number of blue particles; $ N_A$, the total number of particles on the left; and $N$, the total number of particles. From these, the total number of red particles, $N_-$, and the total number of particles on the right, $N_B$, can be derived using $N_+ +N_-=N_A+N_B=N$.

Volume this macrostate is given by
\[\label{eq:rick_volume_app}
V^{\mathrm{R}}=\underset{\blueleft}{\mqty(w\\i)}\underset{\redleft}{\mqty(w\\N_A-i)}\underset{\blueright}{\mqty(w\\N_+-i)}\underset{\redright}{\mqty(w\\N-N_+-(N_A-i))}.
\]
The Binomial coefficients correspond to counting the number of microstates for blue (+) particles on the left side, red (-) on the left, blue on the right, and red on the right, from left to right. 

We compute this volume by using Stirling's approximation,
\[
\begin{split}
    \ln\mqty(w\\k)&=w\ln w - k\ln k-(w-k)\ln(w-k)\\
    &+\frac{1}{2}\ln\frac{w}{2\pi k (w-k)},
\end{split}
\]
Assuming that $w$ is significantly larger than other parameters, we approximate terms such as
\[
\ln (w-k)\approx \ln w-\frac{k}{w}.
\]
This gives
\[\label{eq:comb1}
\begin{split}
\ln\mqty(w\\k)&=w\ln w - k\ln k-(w-k)\left(\ln w-\frac{k}{w}\right)\\
&-\frac{1}{2}\ln 2\pi k+\frac{1}{2}\frac{k}{w},\\
&=k\ln w - k\ln k+k-\frac{k^2}{w}-\frac{1}{2}\ln 2\pi k+\frac{1}{2}\frac{k}{w},\\
&\approx k\ln w - k\ln k+k-\frac{1}{2}\ln 2\pi k,
\end{split}
\]
where in the last term we neglected terms divided by $w$.

For any $0< p\leq 1$, we can conveniently rewrite Eq.~\eqref{eq:comb1} as 
\[\label{eq:comb2}
\ln \mqty(w\\p\;\! N)=pN(\ln w-\ln N-\ln p+1)-\frac{1}{2}\ln 2\pi p N.
\]

We define $p_i=i/N$ as the proportional amount of blue particles on the left out of all particles. We also define the proportion of the total particles on the left, $a=N_A/N$, and the proportion of the blue particles $r=N_+/N$.

We obtain
\[\label{eq:rick_final}
\ln V^{\mathrm{R}}
=N(\ln w-\ln N+S(q)+1)-\frac{1}{2}\ln \Big((2\pi N)^4 \prod_{j=1}^4 q_j \Big)
\]
where $q_1=p_i$, $q_2=a-p_i$, $q_3=r-p_i$, and $q_4=1-r-a+p_i$.

Note that for $a$ or $r$ equal to zero, the constant term is not defined. However, neither is Eq.~\eqref{eq:comb1} for $k=0$ due to the limits of the validity of Stirling's approximation. However, for any tiny but non-zero $a$ and $r$, the equation is valid and very precise. For example, when $a=1/N$ corresponding to $N_A=1$, the problematic term becomes
$\ln a(1-a)N =\ln (1-1/N)\approx 1/N$. Consequently, terms will be small compared to the dominant $N\ln w$ term.

To investigate this situation in detail, consider the situation in which $p_i=0$. As a result, the volume consists of three terms,
\[\label{eq:boundary_case}
V^{\mathrm{R}}=\mqty(w\\N_A)\mqty(w\\N_+)\mqty(w\\N-N_+-N_A)
\]
We analogously derive
\[
\begin{split}
\ln V^{\mathrm{R}}_{\mathrm{init}}(i=0)
&=N(\ln w-\ln N+S(q)+1)\\
&-\frac{1}{2}\ln \Big((2\pi N)^3 \prod_{j=2}^4 q_j \Big),
\end{split}
\]
where $q_2=a$, $q_3=r$, and $q_4=1-r-a$. Thus, the general form is still valid and equal to Eq.~\eqref{eq:rick_final}, considering that one takes out $-\frac{1}{2}\ln(2\pi N q_j)$ for each $q_j=0$.

\subsection{Morty 1}\label{app:morty}

Morty can observe only two macroscopic degrees of freedom, that is, the total number of particles on the left $N_A$, and the total number of particles on the right, which we denote as the total number of particles minus those on the left, $N-N_A$.

We again have $N_+ +N_-=N_A+N_B=N$. Without loss of generatity we define $N_A$ and $N_+$ to be the smaller of the two, i.e., $N_A\leq N_B$ and $N_+\leq N_-$. This guarantees that $N\geq N_A+N_+$ for the expression to be well defined. 
The volume corresponding the this macrostate is the sum over all smaller macrostates (of Rick) that lead to the same macroscopic observation for Morty,
\[\label{eq:general_volume_morty1app}
\begin{split}
&V^{\mathrm{M}_1}=\\
&\sum_{i=0}^{\min\{N_+,N_A\}}\!\!\!\underset{\blueleft}{\mqty(w\\i)}\underset{\redleft}{\mqty(w\\N_A\!-\!i)}\underset{\blueright}{\mqty(w\\N_+\!-\!i)}\underset{\redright}{\mqty(w\\N\!-\!N_+\!-\!(N_A\!-\!i))}\\
&=:\sum_{i=0}^{\min\{N_+,N_A\}}f(i).
\end{split}
\]

We find the largest terms of the sum assuming that $w$ is much larger than the rest. Then, using Stirling's approximation, Eq.~\eqref{eq:comb1}, we express $\ln f(i)$. Setting the derivative of this function to zero, we find $i_{\max}=\frac{N_+N_A}{N}$. Clearly, $0\leq i_{\max}\leq N_+,N_A$, so this value gives a valid maximum. Since logarithm is a strictly growing function, $i_{\max}$ also gives the maximum of $f(i)$.

Recalling definitions $a=N_A/N$ and $r=N_+/N$, we can elegantly write $i_{\max}=a r N$. Using Eq.~\eqref{eq:comb2}, we obtain the maximal term,
\[\label{eq:max_morty1}
\begin{split}
\ln f_{\max}&=N\left(\ln w-\ln N +S(a)+S(r)+1\right)\\
&-\frac{1}{2}\ln( (2\pi N)^4 a(1-a)r(1-r)).
\end{split}
\]

Next, we include the effect of the sum, which will contribute at most logarithmically in $N$. We use the Laplace method to approximate the binomial coefficients around the maximum and then estimate the sum by computing the integral. Mathematically, we have
\[
f(i)=f_{\max}\exp\left(-\frac{(i-i_{\max})^2}{2\sigma^2}\right)
\]
We compute the second derivative to compute the variance,
\[
\sigma^2=-\left(\left.\frac{d^2 \ln f}{di^2}\right|_{i=i_{\max}}\right)^{-1}
=Na(1-a)r(1-r).
\]
Finally, we obtain the sum,
\[\label{eq:general_volume_morty3}
\begin{split}
V^{\mathrm{M}_1}
&=\int_{-\infty}^{+\infty}f(i)=f_{\max}\sqrt{2\pi N a(1-a)r(1-r)}.
\end{split}
\]
Collecting all results, we obtain the macrostate volume,
\[
\begin{split}
    \ln V^{\mathrm{M}_1}&=N\left(\ln w-\ln N +S(a)+S(r)+1\right)\\
    &-\frac{3}{2}\ln(2\pi N).
\end{split}
\]

\subsection{Morty 2}\label{app:morty2}

The macrostate volume of Morty 2, parametrized by a couple $(N_A,N)$, is given by 
\[\label{eq:general_volume_morty2app}
V^{\mathrm{M}_2}=2\sum_{N_+=0}^{N/2}\sum_{i=0}^{\min\{N_+,N_A\}}f(i,N_+).
\]
where $f(i,N_+)$ is a function identical to that of $f(i)$ in Eq.~\eqref{eq:general_volume_morty1app} and equal to the Rick's macrostate volume, Eq.~\eqref{eq:rick_volume_app}.

The maximum of $f(i,N_+)$ is reached at $i_{\max}=aN/2$ and $R_{+\max}=N/2$, and given by
\[\label{eq:max_morty2}
\begin{split}
\ln f_{\max}&=N\left(\ln w-\ln N +S(a)+\ln 2+1\right)\\
&-\frac{1}{2}\ln( \frac{(2\pi N)^4}{4} a(1-a)).
\end{split}
\]
(Obtained by inserting $r=1/2$ into Eq.~\eqref{eq:max_morty2}.)

We got a Gaussian approximation
\[
f(i, N_+) \approx f_{\max}\exp\left( \frac{1}{2} \begin{pmatrix} \Delta i & \Delta N_+ \end{pmatrix} H \begin{pmatrix} \Delta i \\ \Delta N_+ \end{pmatrix}\right),
\]
where $\Delta i=i-i_{\max}$, $\Delta N_+=N_+-R_{+\max}$, and the Hessian is defined as,
\[
H = 
\left. 
\begin{pmatrix}
\displaystyle \frac{\partial^2 \ln f}{\partial i^2} & \displaystyle \frac{\partial^2 \ln f}{\partial i \partial N_+} \\[12pt]
\displaystyle \frac{\partial^2 \ln f}{\partial N_+ \partial i} & \displaystyle \frac{\partial^2 \ln f}{\partial N_+^2}
\end{pmatrix}
\right|_{(i_{\max}, R_{+\max})}
\!\!\!\!\!\!\!\!\!\!\!\!\!\!\!\!\!\!=\frac{4}{N}\begin{pmatrix}
\frac{-1}{a(1-a)} & \frac{1}{1-a} \\
\frac{1}{1-a} & \frac{-1}{a(1-a)} \\
\end{pmatrix}.
\]
Considering that the function is symmetric around the maximum $R_{+\max}$, we can approximate
\[
\begin{split}
V^{\mathrm{M}_2}&=2\int_{0}^{N/2}\int_{0}^{\min\{N_+,N_A\}}f(i,N_+)\ di\, dN_+.\\
&\approx \int_{-\infty}^{+\infty}\int_{-\infty}^{+\infty}f(i,N_+)\ di\, dN_+\\
&=f_{\max}\frac{2\pi}{\sqrt{\mathrm{det}(-H)}}=f_{\max} \frac{\pi N}{2}\sqrt{a(1-a)}.
\end{split}
\]
We obtain
\[
\begin{split}
\ln V^{\mathrm{M}_2}&=N\left(\ln w-\ln N +S(a)+\ln 2+1\right)-\ln(4\pi N).
\end{split}
\]

\subsection{Morty 3}\label{app:morty3}

Finally, we consider the case of ``ignorant'' Morty, who not only cannot distiguish the two colors of particles, but also is not aware that these two colors of particles even exist. As a result, he also does not know how many particles of each type are in the system that he is observing. The only information he can obtain is how many particles are on the left and on the right side.

This means that when calculating the associated observational entropy, he will associate different volumes to these macrostates than their values given by the actual number of microstates that are contained within, given by Eq.~\eqref{eq:general_volume_morty2app}.

From his point of view, the macrostate associated with having $i$ particles on the left and $N-i$ particles on the right is given by
\[
V^{\mathrm{M}_3}=\underset{\blackleft}{\mqty(w\\N_A)}\underset{\blackright}{\mqty(w\\N-N_A)}.
\]
Using Eq.~\eqref{eq:comb2}, we obtain
\[
\begin{split}
\ln V^{\mathrm{M}_3} &= N (\ln w-\ln N +S(a)+1)\\
&-\frac{1}{2}\ln ((2\pi N)^2 a(1-a))
\end{split}
\]

\subsection{Rick - non-symmetric case}\label{app:rick_nonsymmetric}

Finally, let us consider a non-symmetric case of Rick, where the macrostate volume is given by
\[
\begin{split}
&V^{\mathrm{R}}=\underset{\blueleft}{\mqty(L_A\\i)}\underset{\redleft}{\mqty(L_A\\N_A\!-\!i)}\underset{\blueright}{\mqty(L_B\\N_+\!-\!i)}\underset{\redright}{\mqty(L_B\\N\!-\!N_+\!-\!(N_A\!-\!i))}\\
&={\mqty(L_A\\p_i N)}{\mqty(L_A\\(a\!-\!p_i)N)}{\mqty(L_B\\(r\!-\!p_i)N)}{\mqty(L_B\\(1\!-\!r\!-\!a\!+\!p_i)N)}\\
&={\mqty(L_A\\q_1 N)}{\mqty(L_A\\q_2N)}{\mqty(L_B\\q_3 N)}{\mqty(L_B\\q_4 N)}.
\end{split}
\]
Using Eq.~\eqref{eq:comb2}, we find
\[\label{eq:rick_nonsym}
\begin{split}
\ln V^{\mathrm{R}}&=
q_1N(\ln L_A-\ln N-\ln q_1+1)-\frac{1}{2}\ln 2\pi N q_1\\
&+q_2N(\ln L_A-\ln N-\ln q_2+1)-\frac{1}{2}\ln 2\pi N q_2\\
&+q_3N(\ln L_B-\ln N-\ln q_3+1)-\frac{1}{2}\ln 2\pi N q_3\\
&+q_4N(\ln L_B-\ln N-\ln q_4+1)-\frac{1}{2}\ln 2\pi N q_4\\
&=N(a\ln L_A+(1\!-\!a)\ln L_B-\ln N+S(q)+1)\\
&-\frac{1}{2}\ln \Big((2\pi N)^4 \prod_{j=1}^4 q_j\Big).
\end{split}
\]

If the gas is uniformly distributed with all the blue particles are on the left and all the red particles on the right, we have $L_A=a L$, $L_B=(1-a)L$, $q_1=a$, $q_2=q_3=0$ (which also means $r=a$), and $q_4=1-a$. This gives
\[
\begin{split}
\ln V_{\mathrm{init}}^{\mathrm{R}}=N(\ln L-\ln N+1)-\frac{1}{2}\ln \Big((2\pi N)^2 a(1-a)\Big).
\end{split}
\]

When the particles are allow to mix they fill the space uniformly, which corresponds to
\[\label{eq:rick_final_nonsymmetric}
\begin{split}
S_{\C_R}^{\mathrm{fin}}&=\ln V =\ln {\mqty(L\\a N)}{\mqty(L\\(1-a)N)}\\
&=N(\ln L-\ln N+S(a)+1)\\
&-\frac{1}{2}\ln \Big((2\pi N)^2 a(1-a)\Big).
\end{split}
\]
The growth in entropy is the standard entropy of mixing,
\[
\begin{split}
\Delta S^{\mathrm{R}}&=S_{\C_R}^{\mathrm{fin}}-S_{\C_R}^{\mathrm{init}}=\ln \frac{V}{V_{\mathrm{init}}^R}= N S(a).
\end{split}
\]

\subsection{Morty 1 - non-symmetric case}\label{app:morty1_nonsymmetric}
\[
\begin{split}
&V^{\mathrm{M}_1}=\\
&\sum_{i=0}^{\min\{N_+,N_A\}}\!\!\!\underset{\blueleft}{\mqty(L_A\\i)}\underset{\redleft}{\mqty(L_A\\N_A\!-\!i)}\underset{\blueright}{\mqty(L_B\\N_+\!-\!i)}\underset{\redright}{\mqty(L_B\\N\!-\!N_+\!-\!(N_A\!-\!i))}\\
&=:\sum_{i=0}^{\min\{N_+,N_A\}}f(i)\approx f_{\max} \sqrt{2\pi \sigma^2}
\end{split}
\]
where $f_{\max}$ is given again by $i=\frac{N_+ N_A}{N}=a r N$, and the variance is given again by
\[
\sigma^2=-\left(\left.\frac{d^2 \ln f}{di^2}\right|_{i=i_{\max}}\right)^{-1}
=Na(1-a)r(1-r).
\]
We obtain $f_{\max}$ by inserting $p_i=a r$ into Eq.~\eqref{eq:rick_nonsym} and derive
\[
\begin{split}
\ln V^{\mathrm{M}_1}&=N\big(a\ln L_A+(1\!-\!a)\ln L_B-\ln N\\
&+S(a)+S(r)+1\big)-\frac{3}{2}\ln (2\pi N).
\end{split}
\]

For the same initial and final state as considered in Rick's non-symmetric case, we have
\[
\begin{split}
\ln V_{\mathrm{init}}^{\mathrm{M}_1}&=N\big(\ln L-\ln N+S(a)+1\big)-\frac{3}{2}\ln (2\pi N),
\end{split}
\]
while the final entropy is the same as Rick's, Eq.~\eqref{eq:rick_final_nonsymmetric}. Thus, the entropy stays almost constant during the time evolution, with only a logarithmic increase coming from the increased uncertainty about the exact number of particles on each side,
\[
\begin{split}
\Delta S^{\mathrm{M}_1}&=\ln \frac{V}{V_{\mathrm{init}}^{\mathrm{M}_1}}= \frac{3}{2}\ln \Big(2\pi N\Big)-\frac{1}{2}\ln \Big((2\pi N)^2 a(1-a)\Big)\\
&=\frac{1}{2}\ln \frac{2\pi N}{a(1-a)}.
\end{split}
\]

Analysis for the non-symmetric case for Morty 2 and 3 proceeds analogously, with the same results as the symmetric case with a substitution $\ln w\rightarrow a\ln L_A+(1\!-\!a)\ln L_B$.

\section{Morty 3 work extraction}~\label{sec:morty3work_extraction}

In this section, we investigate what it means when Morty 3 attempts to extract work using his incorrect judgment of the situation.

For that, we will have to assume some structure about what he thinks he is doing versus what he is actually doing, as his understanding of the situation is limited.

Let us summarize what we know about Morty 3. He believes there is only one type of particle in the system---he sees only black particles and is unaware that the system may in fact consist of a mixture of red and blue particles. A similar situation would be that of a physicist working with spin-1/2 particles before the discovery of spin.

As a result, he not only underestimates the system’s entropy—judging it by Eq.~\eqref{eq:morty3entropy}—but also misidentifies the Hilbert space the particles occupy. He assigns a state corresponding to what he believes to be black particles, whereas in reality, the system is in a different state composed of red and blue particles. Consequently, when he performs unitary operations to extract work, he believes they act only on black particles, while in fact they affect a mixture of red and blue particles.

Given this, we construct a model of the actual system that aligns with Morty 3’s incomplete description of reality, based on specific assumptions. Additionally, we must identify the underlying transformation that is truly implemented when he believes he is performing an operation. Our approach must be consistent—the actual transformation should reproduce the outcome he intends to achieve.

\subsection{Consistency, assumptions, and the resulting structure}

Mathematically, we define consistency as follows: for all measurements accessible to him, the outcome probabilities before and after his unitary operation must coincide with those generated by the actual state under our model transformation (see Eqs.~\eqref{eq:preUnitary_Measure_appC} and ~\eqref{eq:postUnitary_Measure_appC} below).

Consider that $\R$ and $U$ represent the actual state and transformation, while $\R_M$ and $U_M$, respectively, denote his corresponding perception of the same. Also, consider $\C=\{\P_i\}$ the actual measurement, while $\C_M=\{\P_i^M\}$ is the measurement he thinks he performs.

We begin by assuming the consistency of measurements $(a1)$. This assumption consists of two parts: First, we assume that Morty 3 cannot perform a measurement that distinguishes the colors of the particles in any way. In other words, the set of admissible measurements is restricted to those that cannot distinguish between the blue and the red particles. Second, we assume a physically motivated correspondence between the perceived and actual measurements:
\begin{align}
&(a1) &\C_M \longrightarrow \C.
\end{align}
For example, consider a two-site chain containing a single red particle, which appears black to Morty 3. From his perspective, the measurement of the particle’s position---whether it is on the left or right site---is given by
\[
\C_M=\{\ketbra{10}{10}, \ketbra{01}{01}\}.
\]
In reality, the particle could be either red or blue, corresponding to an extended Hilbert space. Thus, the measurement Morty performs effectively projects onto the subspaces where some particle---regardless of its internal type---is located on the left or right site.
\[\label{eq:morty3_measurement_initial}
\begin{split}
\C=\{&\ketbra{+0}{+0}+\ketbra{-0}{-0}, \\
&\ketbra{0+}{0+}+\ketbra{0-}{0-}\},
\end{split}
\]
This measurement can be expressed using the formalism of Morty 2, who has the correct understanding of the Hilbert space. It illustrates that both internal states, $
\ket{+0}$ and $
\ket{-0}$, which are indistinguishable to Morty 3, will be registered as the same outcome by him. This can be seen as follows: denoting the projector onto the "left" position as $
\P_L=\ketbra{+0}{+0}+\ketbra{-0}{-0}$ (seen on the left), we find that $p_L=
\bra{+0}\P_L\ket{+0}=\bra{-0}\P_L\ket{-0}=1$. We will formalize this assumption in more detail later and extend it to arbitrary projective measurements.

Second, we assume consistency of density matrices (assumption $(a2)$.). For any admissible perceived measurement $\C_M$, with corresponding actual measurement $\C$, the actual density matrix $\R$ must reproduce the same outcome distribution as the perceived density matrix $\R_M$,
\begin{align}\label{eq:preUnitary_Measure_appC}
&(a2) &\R\overset{\C}{\longrightarrow}\ &p_i \overset{\C_M}{\longleftarrow} \R_M,
\end{align}
where $p_i=\tr[\R \P_i]=\tr[\R_M \P_i^M]$.

Third, we assume consistency of unitary operations (assumption $(a3)$). If Morty 3 believes he is applying a unitary $U_M$, then the actual unitary $U$ applied must produce a state that is indistinguishable---under any measurement that cannot resolve particle colors---from the state he believes he has prepared,
\begin{align}\label{eq:postUnitary_Measure_appC}
&(a3) &U\R U^\dag\overset{\C}{\longrightarrow}\ &q_i \overset{\C_M}{\longleftarrow} U_M\R_M U_M^\dag, 
\end{align}
where $q_i=\tr[U\R U^\dag \P_i]=\tr[U_M\R_M U_M^\dag \P_i^M]$.

The next set of assumptions concerns the admissible states, allowed evolutions, and applicable unitary operators, in order to uphold superselection rules between different particle types. We restrict ourselves to the situation where there is a fixed and known total number of particles $N$, as assumed in the main text, implying the actual Hilbert space (which is the Hilbert space of Morty 2, who cannot distinguish the colors, but knows that two colors exist),
\begin{align}\label{eq:HSmorty_half_ignorant2}
&(a4) &\HS\equiv\HS^{\mathrm{M}_2}&=\bigoplus_{N_+=0}^{N}\HS_+^{N_+}\otimes \HS_-^{N-N_+}.
\end{align}
(Recall definition \eqref{eq:HSmorty_half_ignorant}.)

Further, we assume that the actual density matrix does not admit superpositions of states that involve differing number of particles of each type, such as $(\ket{0-}+\ket{+0})/\sqrt{2}$~\footnote{We denote $\ket{0-}$ to express that in the first position there are no particles, and in the second position there is a red particle and so on.}, which we can write as
\begin{align}\label{eq:actual_density}
&(a5) &\R=\bigoplus_{N_+=0}^N \lambda_{N_+} \R_{N_+} .
\end{align}
where $\sum_{N_+=0}^N \lambda_{N_+} =1$, where $\R_{N_+}$ is a density matrix---linear operator---acting on $\HS_+^{N_+}\otimes \HS_-^{N-N_+}$.
In other words, it can only admit classical mixtures between these subspaces, thereby forming a super-selection rule.

We further assume that the unitary that he performs cannot change the color of the particles, meaning that it acts independently on both the red and the blue particle Hilbert spaces, which is mathematically described as factorizing
\begin{align}\label{eq:actual_unitary}
&(a6) &U=\bigoplus_{N_+=0}^N U^{N,N_+}.
\end{align}
where every unitary $U^{N,N_+}$ acts only at the subspace $\HS_+^{N_+}\otimes \HS_-^{N-N_+}$. This also means that $U$ is block-diagonal in the full Hilbert space $\HS$. This ensures that the super-selection cannot be broken by Morty 3 applying a unitary. Intuitively, this ensures that Morty 3, by virtue of his ignorance, does not wield more operational ability than Rick. We constructed Rick's Hilbert space to be constrained by the fixed number of red and blue particles, which implicitly restricts the set of extraction unitaries he can perform. However, because Morty 3's ignorance grants access to a larger Hilbert space, an unconstrained unitary could, in principle, swap the particle colors. To preserve the intended hierarchy between the agents, we ensure that Rick remains the more capable one.

Furthermore, let $\ham_M$ denote the Hamiltonian used by Morty 3 to describe the system's evolution and to assign energy to its states. This Hamiltonian does not account for red and blue particles, as Morty 3 is unaware of their existence. In contrast, let $\ham$ denote the actual Hamiltonian, which includes a full description of particles of both colors.
We assume that the total (actual) Hamiltonian conserves the number of particles of each color,
\begin{align}\label{eq:actual_ham}
&(a7) &\ham=\bigoplus_{N_+=0}^N H^{N,N_+}.
\end{align}
This means that the blocks defined by $N$ and $ N_+$ act and evolve independently of each other. In practice, $(a6)$ and $(a7)$ impose conservation of the probability of finding the state in each subspace $\HS_+^{N_+}\otimes \HS_-^{N-N_+}$.

Finally, we assume that $(a8)$ the actual Hamiltonian is independent of particle color and consistent with Morty 3’s description. For example, if $\ket{011}+\ket{110}$ is an energy eigenvector of $\ham_M$, then $\ket{0+-}+\ket{-+0}$ must also be an eigenvector of $\ham$ with the same energy. In other words, the Hamiltonian depends only on the configuration---that is, the arrangement of particles on the lattice---not on their color. This ensures that Morty can correctly determine the energy of a state by inspecting its configuration alone. It is further assumed that Morty 3 possesses a Hamiltonian which, although incorrect in structure, assigns correct energies to all relevant states.

For this, it will be useful to define a linear mapping $\K$ between the Morty 2 and Morty 3 Hilbert spaces, $\HS$ and $\HS^{M}\equiv\HS^{\mathrm{M}_3}$, defined by its action on the basis vectors which erases the information about the type of particle, i.e., for example
\[
\K(\ket{00++0-})=\ket{001101}
\]
We call the vector on the right-hand side \emph{configuration} and on the left-hand side \emph{inner configuration}. The linearity means that for two inner configuration states $\tilde c_1,\tilde c_1\in \HS$, and a complex number $\alpha$ we have
\[
\K(\alpha \tilde c_1+\tilde c_2)=\alpha\K(\tilde c_1)+\K(\tilde c_1).
\]
Given a vector $c\in \HS^M$ , we define a one-dimensional Hilbert space $\HS_c^M=\mathrm{span}\{c\}$.
We define the pre-image of this Hilbert space as,
\[\label{eq:Hc_def}
\HS_c\equiv \K^{-1}(\HS_c^M)=\{\tilde c \in \HS | \K(\tilde c)\in \HS_c^M\}.
\]
For example, for $c_1=(\ket{01}+\ket{10})/\sqrt{2}$, we have
\[
\begin{split}
\HS_{c_1}=\mathrm{span}\{&(\ket{0+}+\ket{+0})/\sqrt{2},(\ket{0+}+\ket{-0})/\sqrt{2},\\&(\ket{0-}+\ket{+0})/\sqrt{2},(\ket{0-}+\ket{-0})/\sqrt{2}\}.
\end{split}
\]
For $c_2=(\ket{01}-\ket{10})/\sqrt{2}$, we have
\[
\begin{split}
\HS_{c_2}=\mathrm{span}\{&(\ket{0+}-\ket{+0})/\sqrt{2},(\ket{0+}-\ket{-0})/\sqrt{2},\\&(\ket{0-}-\ket{+0})/\sqrt{2},(\ket{0-}-\ket{-0})/\sqrt{2}\}.
\end{split}
\]

$\HS_c$ is a Hilbert space. To prove that it is a vector space, assume a linear combination $\ket{\psi}=\sum_i \alpha_i\ket{\phi_i}$, where $\ket{\phi_i}\in \HS_{c}$. Then from the linearity we have,
\[
\begin{split}
\K \ket{\psi}&=\K \big(\sum_i \alpha_i\ket{\phi_i}\big)=\sum_i \alpha_i\K(\ket{\phi_i})\\
&=c \sum_i \alpha_i\beta_i\in \HS_c^M.
\end{split}
\]
Additionally, it is equipped with an inner product induced from $\HS$, and it is finite. Thus, it is complete.

Note, however, that surprisingly also $(\ket{+0}+\ket{-0})\in\HS_{c_1}$ and $(\ket{0+}-\ket{0-})\in\HS_{c_2}$ are in the kernel, corresponding to zero vector in $\HS_{c_1}^M$ and $\HS_{c_2}^M$, respectively. There is only a finite number of such non-trivial vectors in the kernel, which is a set of measure zero, since it requires a very specific combination of parameters so that $\sum_i \alpha_i\beta_i=0$. This is because the vectors are always characterized by a combination that maps onto the same vector, such as $\ket{+-}-\ket{-+}$ or similar, which maps onto $\ket{11}-\ket{11}=0$. First, physical states always contain some noise, so this pathological situation should not occur in practice. Second, the situation where the state is $0$ corresponds to not having a system; thus, we will ignore it. The second observation is that although $c_1$ and $c_2$ are orthogonal, not every two vectors in $\ket{\psi_1}\in\HS_{c_1}$ and $\ket{\psi_1}\in\HS_{c_2}$ are. E.g., $(\ket{0+}+\ket{-0})/\sqrt{2}$ and $(\ket{0-}-\ket{-0})/\sqrt{2}$ have squared overlap $1/2$. Even if we restrict ourselves to a subspace $\HS_+^{N_+}\otimes \HS_-^{N-N_+}$, still, one can find two vectors created as preimages from orthogonal vectors which have non-zero overlap. If the preimages were, however, created in different subspaces $\HS_+^{N_+}\otimes \HS_-^{N-N_+}$, they are orthogonal by default.

Consider an orthonormal basis $\{c\}$, so that $\HS^M=\mathrm{span}\{c\}$, following the argument after Eq.~\label{eq:Hc_def}. Then the subspaces created from the pre-images combine into the entire Hilbert space,
\[
\HS=\bigoplus_c\HS_c.
\]
As we mentioned above, states from different $\HS_c$ are not orthogonal. However, they are still linearly independent. To prove this, one can assume $\ket{\psi}\in \HS_{c'}$, can be written as a linear combination of vectors from $\HS_{c}$, $c\neq c'$. However, if that is possible, then $\K(\ket{\psi})=\alpha c$ for some number $\alpha$, meaning that $\ket{\psi}\in \HS_c$, which is in contradiction with the assumption that $c\neq c'$. This means that the direct sum is well-defined. Additionally, one can always find an orthonormal basis of each $\HS_c$ so that each vectors from there are orthogonal to each other. This, however, doesn't necessarily imply that they will be orthogonal to vectors from another subspace, $\HS_{c'}$, unless there is some additional structure involved (such as when different subspaces are eigensubspaces of a Hermitian operator).

Assumption $(a8)$ says that the Hamiltonian $H^{N,N_+}$ does not depend on the inner configuration. Consider that Morty 3 has a Hamiltonian 
\[
H_M=\sum_E E \ketbra{E}{E},
\]
which judges the energy correctly. Energy eigenstates $\ket{E}$ are a superposition of different configurations. This means that all eigenstates of the actual Hamiltonian that map onto the same eigenstate of Morty 3's Hamiltonian must correspond to the same energy. In other words, combining this with assumption $(a7)$, Eq.~\eqref{eq:actual_ham}, the actual Hamiltonian is given by
\begin{align}\label{eq:Ham_subsystem}
%&(a6m) &
H^{N,N_+}=\sum_E E \P_E^{N,N_+},
\end{align}
where $\P_E^{N,N_+}$ is a projector onto a subspace that is the overlap of two vector spaces
\[
\HS_{\ket{E}} \cap \HS_+^{N_+}\otimes \HS_-^{N-N_+}.
\]
The overlap means that superselection rules still apply during evolution, i.e., superpositions of states that involve differing numbers of particles of each type, such as $(\ket{0-}+\ket{+0})/\sqrt{2}$, are not created while the system evolves.

Similarly, assuming that Morty thinks that he applies a unitary with spectral decomposition
\[\label{eq:UM}
U_M = \sum_u u \ketbra{u}{u},
\]
the unitary he actually applies will be given by
\[\label{eq:actual_U}
U^{N,N_+} = \sum_u u\, U_u^{N,N_+},
\]
inserted into~\eqref{eq:actual_unitary}. $U_u^{N,N_+}$ is an operator acting on
\[\label{eq:u_subspace}
\HS_+^{N_+}\otimes \HS_-^{N-N_+},
\]
which acts as a unitary on its subspace $\HS_{\ket{u}}$, while zero on its complement,
\[\label{eq:supportU}
U_u^{N,N_+}\equiv U_u^{N,N_+}\oplus 0_{\HS_{\ket{u}}^\perp}.
\]
The reasoning behind this is that only those internal states which are indistinguishable to Morty 2 can be freely mapped onto each other via unitary transformations—such changes will go unnoticed by him. From his perspective, the evolution appears exactly as if the unitary $U_M$ had been applied. If this were not the case, he could, in principle, perform a measurement revealing that the unitary did not act as expected, violating assumption $(a3)$ and contradicting Eq.~\eqref{eq:postUnitary_Measure_appC}. The structure of the unitary operator~\eqref{eq:actual_U} reflects the structure of \eqref{eq:UM}: the same way the rank-1 projector $\ketbra{u}$ acts as an identity unitary operator on the one-dimensional Hilbert space spanned by itself, it acts as a zero operator on all other vectors $\ket{u'}\perp \ket{u}$. The inner structure~\eqref{eq:supportU} reflects that.

This definition leads to,
\[
U_u^{N,N_+}U_{u'}^{N,N_+\dag}=
\begin{cases}
\I_{\HS_{\ket{u}}}\oplus 0_{\HS_{\ket{u}}^\perp}, & u=u',\\
0, & u\neq u'.
\end{cases}
\]
Using this, we can prove unitarity,
\[
\begin{split}
U^{N,N_+}U^{N,N_+\dag}&=\sum_u u\, U_u^{N,N_+} \sum_{u'} u'\, U_{u'}^{N,N_+}\\
&=\sum_{u,u'}\delta_{uu'}u u' \I = \sum_u \abs{u}^2 \I_{\HS_{\ket{u}}}\oplus 0_{\HS_{\ket{u}}^\perp}\\
&=\sum_u \I_{\HS_{\ket{u}}}\oplus 0_{\HS_{\ket{u}}^\perp}= \I,
\end{split}
\]
because $\abs{u}^2=1$. The meaning of $U_u^{N,N_+}$ is that it can mix all vectors within the subspace in which Morty 3 is unable to observe any changes, while outwardly appearing to act exactly as $U_M$ does.

Finally, if the density matrix describing the system in the Morty 3 framework is
\[
\R_M = \sum_r r \ketbra{r}{r},
\]
then the actual density matrix is given by inserting
\[\label{eq:actual_density_matrix}
\R_{N_+} = \sum_r r \R_r^{N,N_+}
\]
into Eq.~\eqref{eq:actual_density}. $\R_r^{N,N_+}$ is a linear operator acting on 
\[
\HS_+^{N_+}\otimes \HS_-^{N-N_+}.
\]
which acts as density matrix on its subspace $\HS_{\ket{r}}$ while zero on its complement,
\[\label{eq:inner_structure_Rr}
\R_r^{N,N_+}\equiv \R_r^{N,N_+}\oplus 0_{\HS_{\ket{r}}^\perp}.
\]
This construction ensures that outwardly the $\R$ appears exactly as $\R_M$, since all the integral degrees of freedom (described by $\R_r$) map onto the same external (described by $\ketbra{r}$) visible to Morty 3, while the external probabilities, $r$, which are observable by Morty 3, stay the same. If it had some other structure, it would be distinguishable by some measurement, and thus in contradition with $(a2)$, Eq.~\eqref{eq:preUnitary_Measure_appC}.

Finally, using the map $\K$ we can rephrase assumption $(a1)$ and derive the prescription for the measurement Morty 3 actually performs. Given Morty 3's $\C_M=\{\P_i^M\}$ description of the measurement he thinks he performs, the actual measurement he performs is
\begin{align}\label{eq:projector_real}
&(a1) &\C=\{\P_i\}.
\end{align}
There, $\P_i$ is the projector onto 
\[
\HS_{\P_i^M}:=\bigoplus_k \HS_{\ket{\psi_{ik}^M}},
\]
where we use the spectral decomposition into orthonormal projectors, $\P_i^M=\sum_k\ketbra{\psi_{ik}^M}$. Note that $\ket{\psi_{ik}}\in \HS_{\ket{\psi_{ik}^M}}$ and $\ket{\psi_{ik'}}\in \HS_{\ket{\psi_{ik'}^M}}$ are not necessarily orthogonal for $k\neq k'$, however, one can always employ an orthogonalization procedure from which $\P_i$ is explicitly formed. Further, subspaces $\HS_{\P_i^M}$ and $\HS_{\P_{i'}^M}$ are not necessarily orthogonal for $i\neq i'$, but they will be if they correspond to different eigenvalues of some normal operator, such as a Hamiltonian or a unitary operator. (We use this several times below.) We could additionally assume that Morty 3 is unable to project onto superpositions of inner configurations—analogous to our assumption on the Hamiltonian in Eq.~\eqref{eq:actual_ham}---to prevent the creation of forbidden superpositions, such as $\ket{+} + \ket{-}$, through measurement. However, this assumption is unnecessary to prove the statement we are trying to reach.

One can verify that this definition indeed corresponds to our earlier, less formal prescription given in Eq.~\eqref{eq:morty3_measurement_initial}. As an additional example, consider a rank-2 projector corresponding to a measurement that detects a total particle number of one.
\[
\P_1^M=\ketbra{10}{10}+\ketbra{01}{01}.
\]
The actual projector corresponding to this is
\[
\begin{split}
\P_1
&=\ketbra{+0}{+0}+\ketbra{0+}{0+}+\ketbra{-0}{-0}+\ketbra{0-}{0-}.
\end{split}
\]
Finally, consider a projector which projects onto a superposition of number states, e.g.,
\[
P_i^M=\ketbra{\psi^M}{\psi^M}\quad \ket{\psi^M}=a_1\ket{011}+a_2\ket{110}.
\]
When restricting ourselves to the assumptions $(a4-a7)$ which forbids superpositions such as $a_1\ket{0--}+a_2\ket{++0}$, the corresponding Hilbert space is,
\[
\HS_{\ket{\psi^M}}=\HS_1\oplus\HS_2\oplus\HS_3,
\]
where
\[
\begin{split}
\HS_1=\mathrm{span}\{&a_1\ket{0++}+a_2\ket{++0}\}\\
\HS_2=\mathrm{span}\{&a_1\ket{0--}+a_2\ket{--0}\}\\
\HS_3=\mathrm{span}\{&a_1\ket{0+-}+a_2\ket{+-0},\\
&a_1\ket{0+-}+a_2\ket{-+0},\\
&a_1\ket{0-+}+a_2\ket{+-0},\\
&a_1\ket{0-+}+a_2\ket{-+0}\}.\\
\end{split}
\]
and allowed states follow the superselection rule $\R=\sum_i\P_i\R\P_i$, where $\P_i$ projects onto $\HS_i$. As mentioned above, this is an example of where not all of these basis vectors are orthogonal, but can be made orthogonal by orthogonalization procedure. 

As a side note, this can be further generalized to any observable: consider Morty 3 has an observable, $O^M=\sum_o o \P_o^M$. The corresponding actual observable is $O=\sum_o o \P_o$, where $\P_o$ is the projector onto $\HS_{\P_o^M}$. The Hamiltonian, Eq.~\eqref{eq:Ham_subsystem}, already follows this pattern, with an additional assumption of Eq.~\eqref{eq:actual_ham}.

\subsection{Extracted work by Morty 3}

Now we have most of the ingredients for deriving Morty 3's actual extracted work. 

We are interested in whether Morty 3's estimate of the extracted work matches the actual extracted work, given the set of consistency assumptions $(a1$--$a3)$, the superselection rule stating that no superposition between different types of particles is allowed, and the assumption that such superselection sectors are preserved under both the dynamics and applied unitaries $(a4$--$a7)$. Finally, we assume that Morty 3 can correctly judge the system's energy $(a8)$, meaning that the Hamiltonian depends only on the configuration and not on the colors of the particles.

Morty 3's estimate of work, for a single copy of the state of the system (i.e., without simultaneous estimation from multiple copies) is~\cite{safranek_2023_work},
\[\label{eq:extracted_work_morty3app}
\begin{split}
W_{\C_M}&=\int_{\tilde U_M} \tr[\ham_M(\rho- U_M^{\mathrm{ext}}\tilde U_M\R\tilde U_M^\dag U_M^{\mathrm{ext}\dag})] \mu (\tilde U_M)\\
&=\tr[\ham_M(\R_M-\pi_M^\cg )],
\end{split}
\]
where
\[
\pi_M^\cg = U_M^{\mathrm{ext}}\R_M^\cg U_M^{\mathrm{ext}\dag},
\]
is a passive state to $\R_M^{\cg}=\sum_i \frac{p_i}{V_i^M}\P_i^M = \int \tilde U_M \R_M \tilde U_M^\dag \mu (\tilde U_M)$, where $\Pi_i^M$, where $\C_M=\{\P_i^M\}$, $V_i^M=\tr[\P_i^M]$. $\tilde U$ is a random unitary randomizing vectors in each subspace $\HS_i$, onto which $\P_i^M$ projects, picked with the Haar measure, and $U_M^{\mathrm{ext.}}$ is the extraction unitary which reorders the eigenvalues from the largest to the smallest and associates them to the energy eigenvectors from the lowest to the highest energy. $U_M=U_M^{\mathrm{ext}}\tilde U_M$ is the full unitary that Morty 3 thinks he applies on the system, consisting both of the random and the extraction unitary.

In reality, however, he performs a measurement in the basis $\C=\{\P_i\}$, where projectors $\P_i$ include all the different colors that lead to the same configurations as vectors is $\P_i$. This is mathematically described in Eq.~\eqref{eq:projector_real}, denoted $\C_M\rightarrow \C$. His actual state of the system is $\R_M \rightarrow \R$, Eq.~\eqref{eq:actual_density_matrix}. The random unitary he actually applies is $U_M\rightarrow U$, Eq.~\eqref{eq:actual_U}, which consists both of the random and the extraction unitary, $U=U^{\mathrm{ext}}\tilde U$, $\tilde U_M\rightarrow \tilde U$, $U_M^{\mathrm{ext}}\rightarrow U^{\mathrm{ext}}$. The actual Hamiltonian is $\ham_M\rightarrow\ham$, Eq.~\eqref{eq:Ham_subsystem}.

The actual energy he will extract on average is then
\[\label{eq:extracted_work_morty3actual}
\begin{split}
&W_{\C}=\int_{\tilde U}\tr[\ham(\R-U\R U^\dag)] \nu (\tilde U),
\end{split}
\]
where $\nu$ is a measure over unitaries $\tilde U$. We examine each term individually. First, however, we prove the following lemma, which will be useful later.

\begin{lemma}\label{lemma:overlap}
Let $\ket{\psi_M}$ an arbitrary but normalized vector in Morty 3's description, a normalized vector $\ket{\psi}\in \HS_{\ket{\psi_M}}\cap \HS_+^{N_+}\otimes \HS_-^{N-N_+}$, and $\P_E$ is a projector onto $\HS_{\ket{E}}\cap \HS_+^{N_+}\otimes \HS_-^{N-N_+}$. Then
\[
\tr[\P_E\ketbra{\psi}{\psi}]=\abs{\braket{E}{\psi_M}}^2.
\]
\end{lemma}
\begin{proof}
We expand $\ket{\psi_M}$ in the basis of $\ket{E}$,
\[
\ket{\psi_M}=\sum_{E} \alpha_{E}\ket{E},\quad \alpha_E=\braket{E}{\psi_M}.
\]
This means that whatever $\ket{\psi}$ is, it must have a form of
\[
\ket{\psi}=\sum_{E} \alpha_{E}\ket{\psi_E}.
\]
where $\ket{\psi_E} \in \HS_E$. This can be proven by contradiction by applying the operator $\K$ and using orthogonality of $\ket{E}$. Since $\ket{\psi_E}$ are eigenvectors of a Hermitian operator (Hamiltonian $\ham$) corresponding to different eigenvalues, they must be orthogonal to each other. More specifically, we have $\P_E\ket{\psi_{E'}}=0$ for any $E'\neq E$. This implies
\[\label{eq:crucial}
\P_E\ket{\psi}=\alpha_{E}\ket{\psi_E},
\]
from which follows,
\[
\tr[\P_E\ketbra{\psi}{\psi}]=\bra{\psi}\P_E^2\ket{\psi}=\abs{\alpha_{E}}^2=\abs{\braket{E}{\psi_M}}^2.
\]
\end{proof}

We also make one useful observation.
\begin{corollary}\label{corr:observation}
The Lemma~\eqref{lemma:overlap} also holds for any other normal operator acting on the actual Hilbert space, rather than Hamiltonian $\ham=\sum_E E \P_E$, such as a unitary operator $U=\sum_u u \P_u$.
\end{corollary}
\begin{proof}
Any normal operator has orthogonal eigenspaces, which the only central assumption in Eq.~\eqref{eq:crucial}.
\end{proof}

For the first term in Eq.~\eqref{eq:extracted_work_morty3actual}, we have
\[
\begin{split}
\int_{\tilde U}\tr[\ham\R] \nu (\tilde U)&=\tr[\ham\R]\\
&=\tr[\bigoplus_{N_+=0}^N H^{N,N_+}\bigoplus_{N_+=0}^N \lambda_{N_+} \R_{N_+}]\\
&=\sum_{N_+=0}^N \lambda_{N_+} \tr[H^{N,N_+}\R_{N_+}]\\
&=\sum_{N_+=0}^N \lambda_{N_+} \tr[\sum_E E \P_E^{N,N_+}\sum_r r \R_r^{N,N_+}]\\
&=\sum_{N_+=0}^N\sum_E \sum_r E r \lambda_{N_+} \tr[\P_E^{N,N_+} \R_r^{N,N_+}]\\
&=\sum_{N_+=0}^N\sum_E \sum_r E r \lambda_{N_+} \abs{\braket{E}{r}}^2\\
&=\sum_E \sum_r E r \abs{\braket{E}{r}}^2\\
&=\tr[\ham_M \R_M].
\end{split}
\]
We have used
\[\label{eq:trace_energy}
\tr[\P_E^{N,N_+} \R_r^{N,N_+}] = \abs{\braket{E}{r}}^2.
\]
This follows from Eq.~\eqref{eq:inner_structure_Rr}, which says that $\R_r^{N,N_+}$ has support only on $\HS_{\ket{r}}$, meaning that it has a spectral decomposition 
\[\label{eq:inner_density_matrix}
\R_r^{N,N_+}=\sum_j \beta_j^r \ketbra{\psi_j^{r,N,N_+}}{\psi_j^{r,N,N_+}},
\]
$\ket{\psi_j^{r,N,N_+}}\in \HS_{\ket{r}}$.  Applying Lemma~\ref{lemma:overlap} yields
\[
\begin{split}
\tr[\P_E^{N,N_+} \R_r^{N,N_+}] &\!=\! \sum_j \beta_j^r\tr[\P_E^{N,N_+}  \!\ketbra{\psi_j^{r,N,N_+}}{\psi_j^{r,N,N_+}\!}]\\
&\!=\!\sum_j \beta_j^r\tr[\P_E^{N,N_+}  \abs{\braket{E}{r}}^2] = \abs{\braket{E}{r}}^2.
\end{split}
\]

Before we can address the second term, we need to discuss the measure $\mu (\tilde U)$. From the definition of the protocol, Morty 3 chooses a Haar measure over the unitary group $\mu (\tilde U)$, where unitaries act on the Hilbert space $\HS_3$, describing a chain of $N$ colorless particles. In reality, he applies some other random unitary on the Hilbert space, with some arbitrary measure over the degrees of freedom he cannot observe. Thus, to define the corresponding measure, we look for any measure $\nu$ of the unitary group $\tilde U(\HS)=:X$ such that it reduces to the Haar measure $\mu$ for $\tilde U_M(\HS_3)=:Y$.

Let us define map $f:X \rightarrow Y$, using the prescription of Eq.~\eqref{eq:actual_U},
\[
\tilde U = \bigoplus_{N_+=0}^N \sum_{\tilde u} {\tilde u}\, \tilde U_{\tilde u}^{N,N_+}
\overset{f}{\longrightarrow} \tilde U_M=\sum_{\tilde u} {\tilde u} \ketbra{\tilde u}{\tilde u}.
\]

For a measurable subset $A \subset X$, we define the measure
\[\label{eq:full_measure}
\nu(A) = \int_Y \kappa(\tilde U_M, A \cap f^{-1}(\tilde U_M)) \, d\mu(\tilde U_M),
\]
where $f^{-1}(\tilde U_M)$ is a fiber and $\kappa(\tilde U_M, \bullet)$ is a conditional probability measure supported on that fiber. We require that
\[
\kappa(\tilde{U}_M, f^{-1}(\tilde{U}_M)) = 1.
\]
This conditional probability encodes an arbitrary choice over the unobservable degrees of freedom, while the normalization condition ensures that integrating over them reproduces the Haar measure. This can be directly confirmed by integrating $\nu$ over the preimage of a measurable subset $B \subset Y$,
\begin{align*}
\nu(f^{-1}(B)) 
&= \int_Y \kappa(\tilde U_M, f^{-1}(B) \cap f^{-1}(\tilde U_M)) \, d\mu(\tilde U_M) \\
&= \int_Y \mathbf{1}_B(\tilde U_M) \cdot \kappa(\tilde U_M, f^{-1}(\tilde U_M)) \, d\mu(\tilde U_M) \\
&= \int_B \underbrace{\kappa(\tilde U_M, f^{-1}(\tilde U_M))}_{=1} \, d\mu(\tilde U_M) \\
&= \int_B d\mu(\tilde U_M) = \mu(B).
\end{align*}
Here, we used the indicator function $\mathbf{1}_B : Y \to \{0, 1\}$, which selects whether a given fiber $f^{-1}(\tilde U_M)$ lies in the pre-image $f^{-1}(B)$,
\[
\mathbf{1}_B(y) =
\begin{cases}
1, & \text{if } y \in B, \\
0, & \text{if } y \notin B.
\end{cases}
\]
In summary, this construction ensures that the pushforward of $\nu$ under $f$ recovers the original Haar measure,
\[
f_* \nu = \mu.
\]

We will also need the following Lemma.
\begin{lemma}\label{lemma:Utransformed_PE}
Let $\ham=\sum_E E\P_E$ be an actual Hamiltonian corresponding to Hamiltonian $\ham_M=\sum_E E \ketbra{E}{E}$ in the Morty 3's frame, i.e., $\P_E$ projects onto $\HS_{\ket{E}}$. Let $\tilde \P_E:= U^\dag \P_E U$ be a projector of a transformed Hamiltonian $\tilde \ham = U^\dag \ham U=\sum_E E \tilde \P_E$. There,
\[
U=\sum_u u\, U_u\oplus 0_{\HS_{\ket{u}}^\perp}
\]
is an actual unitary operator corresponding to unitary operator $U_M=\sum u \ketbra{u}{u}$ in the Morty 3's frame, given by construction in Eq.~\eqref{eq:supportU}, where $U_u$ is a unitary operator acting on subspace $\HS_{\ket{u}}$. Then projectors $\tilde P_E$ project onto subspace $\HS_{U_M^\dag\ket{E}}$.
\end{lemma}

\begin{proof}
We explicitly express the definition,
\[
\HS_{U_M^\dag\ket{E}}=\big\{\ket{\psi}\, \big|\, \K (\ket{\psi}) = \alpha U_M^\dag\ket{E},\ \alpha \in \mathbb{C}\big\}.
\]
We have
\[
U_M^\dag\ket{E}=\sum_u u^* \ket{u}\braket{u}{E},
\]
where both $\ket{u}$ and $\ket{E}$ are normalized and form orthogonal bases. Due to linearity of $\K$, any vector $\ket{\psi}\in \HS_{U_M^\dag\ket{E}}$ must have a form
\[\label{eq:form_of_transformed}
\ket{\psi}=\alpha\sum_u u^* \braket{u}{E}\ket{\psi_u},
\]
where $\ket{\psi_u}\in \HS_{\ket{u}}$ is normalized so that $\K(\ket{\psi_u})=\ket{u}$. To show that, let us define projector onto $\HS_{\ket{u}}$ as $\P_u$. Then, according to Lemma~\ref{lemma:overlap} combined with Corollary~\ref{corr:observation},
\[
\tr[\P_u\ketbra{\psi_u}{\psi_u}]=\abs{\braket{u}{u}}^2=1.
\]
Considering $\P_u\ket{\psi_u}=\ket{\psi_u}$, we can rewrite this as,
\[
\abs{\braket{\psi_u}{\psi_u}}^2=1.
\]
Thus, for $\K(\ket{\psi_u})=\ket{u}$ to hold, $\ket{\psi_u}$ must be also normalized.

Consider a spectral decomposition into orthonormal eigenvectors,
\[
\P_E = \sum_k \ketbra{\psi_E^k}{\psi_E^k}
\]
By definition, we have
\[
\tilde \P_E = \sum_k U^\dag\ketbra{\psi_E^k}{\psi_E^k}U
\]
where all eigenvectors $U^\dag \ket{\psi_E^k}$ are orthogonal and normalized. Consider a single eigenvector of this operator,
\[
\begin{split}
U^\dag \ket{\psi_E^k}&=\Big(\sum_u u^*\, U_u\oplus 0_{\HS_{\ket{u}}^\perp}\Big)\ket{\psi_E^k}\\
&=\sum_u u^*\, U_u\P_u\ket{\psi_E^k}\\
&\overset{(1)}{=}\sum_u u^*\, \braket{u}{E}\frac{e^{i \phi_{uEk}}U_u\P_u\ket{\psi_E^k}}{\norm{\P_u\ket{\psi_E^k}}}\\
&=\sum_u u^*\, \braket{u}{E}\ket{\psi_{uE}^k},\\
\end{split}
\]
where $\ket{\psi_{uE}^k}\in \HS_u$ is a normalized vector. The step $(1)$ follows from the normalization: using Lemma~\ref{lemma:overlap} combined with Corollary~\ref{corr:observation}, we have
\[
\norm{U_u\P_u\ket{\psi_E^k}}^2=\tr[\P_u\ketbra{\psi_E^k}]=\abs{\braket{u}{E}}^2,
\]
and the complex phase $e^{i \phi_{uEk}}$ ensures that $\braket{u}{E}e^{i \phi_{uEk}}=\abs{\braket{u}{E}}$. Since every eigenvector $U^\dag \ket{\psi_E^k}$ of $\tilde P_E$ has a form of Eq.~\eqref{eq:form_of_transformed}, is normalized and orthogonal, and $\ket{\psi_{uE}^k}$ are normalized, they form an orthonormal basis of $\HS_{U_M^\dag\ket{E}}$. In other words, $\tilde P_E$ projects onto this subspace. 
\end{proof}

Finally, we address the second term,
\vspace{1cm}
\begingroup
\begin{widetext}
\allowdisplaybreaks
\begin{align}
&\int_{\tilde U}\tr[\ham U\R U^\dag] \nu (\tilde U)=\int_{\tilde U}\tr[\ham U^{\mathrm{ext}}\tilde U\R \tilde U^\dag U^{\mathrm{ext}\dag}] \nu (\tilde U) \notag\\
&=\int_{\tilde U}\tr \Bigg[\bigoplus_{N_+=0}^N H^{N,N_+} \bigoplus_{N_+=0}^N U^{N,N_+,\mathrm{ext}}\bigoplus_{N_+=0}^N \tilde U^{N,N_+} \bigoplus_{N_+=0}^N\lambda_{N_+} \R_{N_+} \bigoplus_{N_+=0}^N \tilde U^{N,N_+\dag} \bigoplus_{N_+=0}^N U^{N,N_+,\mathrm{ext}\dag} \Bigg]\nu (\tilde U)\notag \\
&=\int_{\tilde U}\sum_{N_+=0}^N \lambda_{N_+} \tr \Big[H^{N,N_+} U^{N,N_+,\mathrm{ext}} \tilde U^{N,N_+} \R_{N_+} \tilde U^{N,N_+\dag} U^{N,N_+,\mathrm{ext}\dag} \Big]\nu (\tilde U)\notag\\
&=\int_{\tilde U}\sum_{N_+=0}^N \lambda_{N_+} \tr \Big[\sum_E E \P_E^{N,N_+} U^{N,N_+,\mathrm{ext}} \tilde U^{N,N_+} \sum_r r \R_r^{N,N_+} \tilde U^{N,N_+\dag} U^{N,N_+,\mathrm{ext}\dag} \Big]\nu (\tilde U)\notag\\
&=\int_{\tilde U}\sum_{N_+=0}^N \lambda_{N_+} \sum_E \sum_r  E r \tr \Big[\tilde U^{N,N_+\dag}U^{N,N_+,\mathrm{ext}\dag}\P_E^{N,N_+} U^{N,N_+,\mathrm{ext}} \tilde U^{N,N_+}  \R_r^{N,N_+}   \Big]\nu (\tilde U)\notag\\
&=\int_{\tilde U}\sum_{N_+=0}^N \lambda_{N_+} \sum_E \sum_r  E r \tr \Big[\tilde\P_E^{N,N_+}  \R_r^{N,N_+}   \Big]\nu (\tilde U) \\
&\overset{(1)}{=} \int_{\tilde U} \sum_{N_+=0}^N\lambda_{N_+} \sum_E \sum_r  E r \abs{\bra{r}\tilde U_M^\dag U_M^{\mathrm{ext}\dag}\ket{E}}^2 \, \nu(\tilde U) \notag\\
&\overset{(2)}{=} \int_{\tilde U} \sum_E \sum_r  E r \abs{\bra{r}\tilde U_M^\dag U_M^{\mathrm{ext}\dag}\ket{E}}^2 \, \nu(\tilde U) \notag\\
&\overset{(3)}{=} \int_{\tilde U_M} \left( \int_{f^{-1}(\tilde U_M)} \sum_E \sum_r  E r \abs{\bra{r}\tilde U_M^\dag U_M^{\mathrm{ext}\dag}\ket{E}}^2 \, \kappa(\tilde U_M, d\tilde U) \right) \mu(\tilde U_M)\notag \\
&\overset{(4)}{=} \int_{\tilde U_M} \sum_E \sum_r  E r \abs{\bra{r}\tilde U_M^\dag U_M^{\mathrm{ext}\dag}\ket{E}}^2 \, \mu(\tilde U_M)\notag\\
&=\int_{\tilde U_M} \tr[\ham_M U_M^{\mathrm{ext}}\tilde U_M\R_M\tilde U_M^\dag U_M^{\mathrm{ext}\dag}] \mu (\tilde U_M).\notag
\end{align}
\end{widetext}
\endgroup
There we have defined coarse-graining
\[
\tilde\P_E^{N,N_+}=\tilde U^{N,N_+\dag}U^{N,N_+,\mathrm{ext}\dag}\P_E^{N,N_+} U^{N,N_+,\mathrm{ext}} \tilde U^{N,N_+},
\]
corresponding to Hamiltonian $\tilde H^{N,N_+}=\sum E \tilde\P_E^{N,N_+}$, where, according to Lemma~\ref{lemma:Utransformed_PE}, each $\tilde \P_E^{N,N_+}$ projects onto subspace 
\[
\HS_{U_M^\dag U_M^{\mathrm{ext}\dag}\ket{E}} \cap \HS_+^{N_+}\otimes \HS_-^{N-N_+}.
\]
This is because both $U$ and $U^{\mathrm{ext}}$ follow the same structure as required by the Lemma, so the statement follows by applying the Lemma twice. In step $(1)$, we have used this property together with Eq.~\eqref{eq:trace_energy}. In step $(2)$, we have summed the eigenvalues, $\sum_{N_+=0}^N\lambda_{N_+}=1$, because the summand no longer depends on $N_+$. In step $(3)$, we have explicitly substituted the measure $\nu$ using Eq.~\eqref{eq:full_measure}. In step $(4)$, we trivially integrated over the inner degrees of freedom, since the integrand does not depend on the fiber variables, and the conditional measure $\kappa$ is normalized on each fiber.

Inspecting Eqs.~\eqref{eq:extracted_work_morty3app} and~\eqref{eq:extracted_work_morty3actual}, we have proven that
\[
W_\C = W_{\C_M}.
\]
Thus, given the assumptions $(a1$--$a8)$, Morty 3 extracts exactly the amount of energy he predicts, without the knowledge of the existence of two types of particles.

The average work extracted per copy when extracting from multiple copies simultaneously, which is defined by~\cite{safranek_2023_work},
\[
\begin{split}
W_{\C}=\lim_{L\rightarrow \infty}\frac{1}{L}\int_{\tilde U} &\tr\big[\ham^L\rho^{\otimes L}\big]\\
-
&\tr\big[\ham^L U^{\mathrm{ext}}\tilde U^{\otimes L}\R^{\otimes L}\tilde U^{\otimes L\dag} U^{\mathrm{ext}\dag}\big] \mu (\tilde U),
\end{split}
\]
where $\ham^L=\ham\otimes \I\otimes \cdots+\I\otimes\ham\otimes \cdots$,
proceeds analogously.

\bibliography{main.bib}

\begin{thebibliography}{88}%
\makeatletter
\providecommand \@ifxundefined [1]{%
 \@ifx{#1\undefined}
}%
\providecommand \@ifnum [1]{%
 \ifnum #1\expandafter \@firstoftwo
 \else \expandafter \@secondoftwo
 \fi
}%
\providecommand \@ifx [1]{%
 \ifx #1\expandafter \@firstoftwo
 \else \expandafter \@secondoftwo
 \fi
}%
\providecommand \natexlab [1]{#1}%
\providecommand \enquote  [1]{``#1''}%
\providecommand \bibnamefont  [1]{#1}%
\providecommand \bibfnamefont [1]{#1}%
\providecommand \citenamefont [1]{#1}%
\providecommand \href@noop [0]{\@secondoftwo}%
\providecommand \href [0]{\begingroup \@sanitize@url \@href}%
\providecommand \@href[1]{\@@startlink{#1}\@@href}%
\providecommand \@@href[1]{\endgroup#1\@@endlink}%
\providecommand \@sanitize@url [0]{\catcode `\\12\catcode `\$12\catcode `\&12\catcode `\#12\catcode `\^12\catcode `\_12\catcode `\%12\relax}%
\providecommand \@@startlink[1]{}%
\providecommand \@@endlink[0]{}%
\providecommand \url  [0]{\begingroup\@sanitize@url \@url }%
\providecommand \@url [1]{\endgroup\@href {#1}{\urlprefix }}%
\providecommand \urlprefix  [0]{URL }%
\providecommand \Eprint [0]{\href }%
\providecommand \doibase [0]{https://doi.org/}%
\providecommand \selectlanguage [0]{\@gobble}%
\providecommand \bibinfo  [0]{\@secondoftwo}%
\providecommand \bibfield  [0]{\@secondoftwo}%
\providecommand \translation [1]{[#1]}%
\providecommand \BibitemOpen [0]{}%
\providecommand \bibitemStop [0]{}%
\providecommand \bibitemNoStop [0]{.\EOS\space}%
\providecommand \EOS [0]{\spacefactor3000\relax}%
\providecommand \BibitemShut  [1]{\csname bibitem#1\endcsname}%
\let\auto@bib@innerbib\@empty
%</preamble>
\bibitem [{\citenamefont {Abd}\ \emph {et~al.}(2020)\citenamefont {Abd}, \citenamefont {Naji}, \citenamefont {Hashim},\ and\ \citenamefont {Othman}}]{carbon2020}%
  \BibitemOpen
  \bibfield  {author} {\bibinfo {author} {\bibfnamefont {A.~A.}\ \bibnamefont {Abd}}, \bibinfo {author} {\bibfnamefont {S.~Z.}\ \bibnamefont {Naji}}, \bibinfo {author} {\bibfnamefont {A.~S.}\ \bibnamefont {Hashim}},\ and\ \bibinfo {author} {\bibfnamefont {M.~R.}\ \bibnamefont {Othman}},\ }\bibfield  {title} {\bibinfo {title} {Carbon dioxide removal through physical adsorption using carbonaceous and non-carbonaceous adsorbents: A review},\ }\href {https://doi.org/https://doi.org/10.1016/j.jece.2020.104142} {\bibfield  {journal} {\bibinfo  {journal} {Journal of Environmental Chemical Engineering}\ }\textbf {\bibinfo {volume} {8}},\ \bibinfo {pages} {104142} (\bibinfo {year} {2020})}\BibitemShut {NoStop}%
\bibitem [{\citenamefont {Karagiannis}\ and\ \citenamefont {Soldatos}(2008)}]{KARAGIANNIS2008water}%
  \BibitemOpen
  \bibfield  {author} {\bibinfo {author} {\bibfnamefont {I.~C.}\ \bibnamefont {Karagiannis}}\ and\ \bibinfo {author} {\bibfnamefont {P.~G.}\ \bibnamefont {Soldatos}},\ }\bibfield  {title} {\bibinfo {title} {Water desalination cost literature: review and assessment},\ }\href {https://doi.org/https://doi.org/10.1016/j.desal.2007.02.071} {\bibfield  {journal} {\bibinfo  {journal} {Desalination}\ }\textbf {\bibinfo {volume} {223}},\ \bibinfo {pages} {448} (\bibinfo {year} {2008})},\ \bibinfo {note} {european Desalination Society and Center for Research and Technology Hellas (CERTH), Sani Resort 22–25 April 2007, Halkidiki, Greece}\BibitemShut {NoStop}%
\bibitem [{\citenamefont {Chandra}\ and\ \citenamefont {Walsh}(2024)}]{chandra2024microplastics}%
  \BibitemOpen
  \bibfield  {author} {\bibinfo {author} {\bibfnamefont {S.}~\bibnamefont {Chandra}}\ and\ \bibinfo {author} {\bibfnamefont {K.~B.}\ \bibnamefont {Walsh}},\ }\bibfield  {title} {\bibinfo {title} {Microplastics in water: Occurrence, fate and removal},\ }\href {https://doi.org/https://doi.org/10.1016/j.jconhyd.2024.104360} {\bibfield  {journal} {\bibinfo  {journal} {Journal of Contaminant Hydrology}\ }\textbf {\bibinfo {volume} {264}},\ \bibinfo {pages} {104360} (\bibinfo {year} {2024})}\BibitemShut {NoStop}%
\bibitem [{\citenamefont {Gonsioroski}\ \emph {et~al.}(2020)\citenamefont {Gonsioroski}, \citenamefont {Mourikes},\ and\ \citenamefont {Flaws}}]{gonsioroski2020endocrine}%
  \BibitemOpen
  \bibfield  {author} {\bibinfo {author} {\bibfnamefont {A.}~\bibnamefont {Gonsioroski}}, \bibinfo {author} {\bibfnamefont {V.~E.}\ \bibnamefont {Mourikes}},\ and\ \bibinfo {author} {\bibfnamefont {J.~A.}\ \bibnamefont {Flaws}},\ }\bibfield  {title} {\bibinfo {title} {Endocrine disruptors in water and their effects on the reproductive system},\ }\href {https://doi.org/10.3390/ijms21061929} {\bibfield  {journal} {\bibinfo  {journal} {International journal of molecular sciences}\ }\textbf {\bibinfo {volume} {21}},\ \bibinfo {pages} {1929} (\bibinfo {year} {2020})}\BibitemShut {NoStop}%
\bibitem [{\citenamefont {Gong}\ \emph {et~al.}(2021)\citenamefont {Gong}, \citenamefont {de~Moraes~Neto}, \citenamefont {Zha}, \citenamefont {Wu}, \citenamefont {Rong}, \citenamefont {Ye}, \citenamefont {Li}, \citenamefont {Zhu}, \citenamefont {Wang}, \citenamefont {Zhao}, \citenamefont {Liang}, \citenamefont {Lin}, \citenamefont {Xu}, \citenamefont {Peng}, \citenamefont {Deng}, \citenamefont {Bayat}, \citenamefont {Zhu},\ and\ \citenamefont {Pan}}]{gong2021experimental}%
  \BibitemOpen
  \bibfield  {author} {\bibinfo {author} {\bibfnamefont {M.}~\bibnamefont {Gong}}, \bibinfo {author} {\bibfnamefont {G.~D.}\ \bibnamefont {de~Moraes~Neto}}, \bibinfo {author} {\bibfnamefont {C.}~\bibnamefont {Zha}}, \bibinfo {author} {\bibfnamefont {Y.}~\bibnamefont {Wu}}, \bibinfo {author} {\bibfnamefont {H.}~\bibnamefont {Rong}}, \bibinfo {author} {\bibfnamefont {Y.}~\bibnamefont {Ye}}, \bibinfo {author} {\bibfnamefont {S.}~\bibnamefont {Li}}, \bibinfo {author} {\bibfnamefont {Q.}~\bibnamefont {Zhu}}, \bibinfo {author} {\bibfnamefont {S.}~\bibnamefont {Wang}}, \bibinfo {author} {\bibfnamefont {Y.}~\bibnamefont {Zhao}}, \bibinfo {author} {\bibfnamefont {F.}~\bibnamefont {Liang}}, \bibinfo {author} {\bibfnamefont {J.}~\bibnamefont {Lin}}, \bibinfo {author} {\bibfnamefont {Y.}~\bibnamefont {Xu}}, \bibinfo {author} {\bibfnamefont {C.-Z.}\ \bibnamefont {Peng}}, \bibinfo {author} {\bibfnamefont {H.}~\bibnamefont {Deng}}, \bibinfo {author} {\bibfnamefont {A.}~\bibnamefont {Bayat}}, \bibinfo {author} {\bibfnamefont
  {X.}~\bibnamefont {Zhu}},\ and\ \bibinfo {author} {\bibfnamefont {J.-W.}\ \bibnamefont {Pan}},\ }\bibfield  {title} {\bibinfo {title} {Experimental characterization of the quantum many-body localization transition},\ }\href {https://doi.org/10.1103/PhysRevResearch.3.033043} {\bibfield  {journal} {\bibinfo  {journal} {Phys. Rev. Res.}\ }\textbf {\bibinfo {volume} {3}},\ \bibinfo {pages} {033043} (\bibinfo {year} {2021})}\BibitemShut {NoStop}%
\bibitem [{\citenamefont {Bernon}\ \emph {et~al.}(2013)\citenamefont {Bernon}, \citenamefont {Hattermann}, \citenamefont {Bothner}, \citenamefont {Knufinke}, \citenamefont {Weiss}, \citenamefont {Jessen}, \citenamefont {Cano}, \citenamefont {Kemmler}, \citenamefont {Kleiner}, \citenamefont {Koelle},\ and\ \citenamefont {Fort{\'a}gh}}]{bernon2013}%
  \BibitemOpen
  \bibfield  {author} {\bibinfo {author} {\bibfnamefont {S.}~\bibnamefont {Bernon}}, \bibinfo {author} {\bibfnamefont {H.}~\bibnamefont {Hattermann}}, \bibinfo {author} {\bibfnamefont {D.}~\bibnamefont {Bothner}}, \bibinfo {author} {\bibfnamefont {M.}~\bibnamefont {Knufinke}}, \bibinfo {author} {\bibfnamefont {P.}~\bibnamefont {Weiss}}, \bibinfo {author} {\bibfnamefont {F.}~\bibnamefont {Jessen}}, \bibinfo {author} {\bibfnamefont {D.}~\bibnamefont {Cano}}, \bibinfo {author} {\bibfnamefont {M.}~\bibnamefont {Kemmler}}, \bibinfo {author} {\bibfnamefont {R.}~\bibnamefont {Kleiner}}, \bibinfo {author} {\bibfnamefont {D.}~\bibnamefont {Koelle}},\ and\ \bibinfo {author} {\bibfnamefont {J.}~\bibnamefont {Fort{\'a}gh}},\ }\bibfield  {title} {\bibinfo {title} {Manipulation and coherence of ultra-cold atoms on a superconducting atom chip},\ }\href {https://doi.org/10.1038/ncomms3380} {\bibfield  {journal} {\bibinfo  {journal} {Nature Communications}\ }\textbf {\bibinfo {volume} {4}},\ \bibinfo {pages} {2380} (\bibinfo
  {year} {2013})}\BibitemShut {NoStop}%
\bibitem [{\citenamefont {L{\'e}onard}\ \emph {et~al.}(2023)\citenamefont {L{\'e}onard}, \citenamefont {Kim}, \citenamefont {Rispoli}, \citenamefont {Lukin}, \citenamefont {Schittko}, \citenamefont {Kwan}, \citenamefont {Demler}, \citenamefont {Sels},\ and\ \citenamefont {Greiner}}]{leonard2023}%
  \BibitemOpen
  \bibfield  {author} {\bibinfo {author} {\bibfnamefont {J.}~\bibnamefont {L{\'e}onard}}, \bibinfo {author} {\bibfnamefont {S.}~\bibnamefont {Kim}}, \bibinfo {author} {\bibfnamefont {M.}~\bibnamefont {Rispoli}}, \bibinfo {author} {\bibfnamefont {A.}~\bibnamefont {Lukin}}, \bibinfo {author} {\bibfnamefont {R.}~\bibnamefont {Schittko}}, \bibinfo {author} {\bibfnamefont {J.}~\bibnamefont {Kwan}}, \bibinfo {author} {\bibfnamefont {E.}~\bibnamefont {Demler}}, \bibinfo {author} {\bibfnamefont {D.}~\bibnamefont {Sels}},\ and\ \bibinfo {author} {\bibfnamefont {M.}~\bibnamefont {Greiner}},\ }\bibfield  {title} {\bibinfo {title} {Probing the onset of quantum avalanches in a many-body localized system},\ }\href {https://doi.org/10.1038/s41567-022-01887-3} {\bibfield  {journal} {\bibinfo  {journal} {Nature Physics}\ }\textbf {\bibinfo {volume} {19}},\ \bibinfo {pages} {481} (\bibinfo {year} {2023})}\BibitemShut {NoStop}%
\bibitem [{\citenamefont {Blundell}\ and\ \citenamefont {Blundell}(2010)}]{blundell2010concepts}%
  \BibitemOpen
  \bibfield  {author} {\bibinfo {author} {\bibfnamefont {S.~J.}\ \bibnamefont {Blundell}}\ and\ \bibinfo {author} {\bibfnamefont {K.~M.}\ \bibnamefont {Blundell}},\ }\href {https://doi.org/10.1093/acprof:oso/9780199562091.001.0001} {\emph {\bibinfo {title} {Concepts in thermal physics}}}\ (\bibinfo  {publisher} {Oup Oxford},\ \bibinfo {year} {2010})\BibitemShut {NoStop}%
\bibitem [{\citenamefont {Jaynes}(1992)}]{jaynes1992gibbs}%
  \BibitemOpen
  \bibfield  {author} {\bibinfo {author} {\bibfnamefont {E.~T.}\ \bibnamefont {Jaynes}},\ }\bibinfo {title} {The {Gibbs} {Paradox}},\ in\ \href {https://doi.org/10.1007/978-94-017-2219-3_1} {\emph {\bibinfo {booktitle} {Maximum Entropy and Bayesian Methods: Seattle, 1991}}},\ \bibinfo {editor} {edited by\ \bibinfo {editor} {\bibfnamefont {C.~R.}\ \bibnamefont {Smith}}, \bibinfo {editor} {\bibfnamefont {G.~J.}\ \bibnamefont {Erickson}},\ and\ \bibinfo {editor} {\bibfnamefont {P.~O.}\ \bibnamefont {Neudorfer}}}\ (\bibinfo  {publisher} {Springer Netherlands},\ \bibinfo {address} {Dordrecht},\ \bibinfo {year} {1992})\ pp.\ \bibinfo {pages} {1--21}\BibitemShut {NoStop}%
\bibitem [{\citenamefont {Tatarin}\ and\ \citenamefont {Borodiouk}(1999)}]{tatarin_1999_entropy}%
  \BibitemOpen
  \bibfield  {author} {\bibinfo {author} {\bibfnamefont {V.}~\bibnamefont {Tatarin}}\ and\ \bibinfo {author} {\bibfnamefont {O.}~\bibnamefont {Borodiouk}},\ }\bibfield  {title} {\bibinfo {title} {Entropy {{Calculation}} of {{Reversible Mixing}} of {{Ideal Gases Shows Absence}} of {{Gibbs Paradox}}},\ }\href {https://doi.org/10.3390/e1020025} {\bibfield  {journal} {\bibinfo  {journal} {Entropy}\ }\textbf {\bibinfo {volume} {1}},\ \bibinfo {pages} {25} (\bibinfo {year} {1999})}\BibitemShut {NoStop}%
\bibitem [{\citenamefont {Allahverdyan}\ and\ \citenamefont {Nieuwenhuizen}(2006)}]{allahverdyan_2006_explanation}%
  \BibitemOpen
  \bibfield  {author} {\bibinfo {author} {\bibfnamefont {A.~E.}\ \bibnamefont {Allahverdyan}}\ and\ \bibinfo {author} {\bibfnamefont {{\relax Th}.~M.}\ \bibnamefont {Nieuwenhuizen}},\ }\bibfield  {title} {\bibinfo {title} {Explanation of the {{Gibbs}} paradox within the framework of quantum thermodynamics},\ }\href {https://doi.org/10.1103/PhysRevE.73.066119} {\bibfield  {journal} {\bibinfo  {journal} {Phys. Rev. E}\ }\textbf {\bibinfo {volume} {73}},\ \bibinfo {pages} {066119} (\bibinfo {year} {2006})}\BibitemShut {NoStop}%
\bibitem [{\citenamefont {{Ben-Naim}}(2007)}]{ben-naim_2007_socalled}%
  \BibitemOpen
  \bibfield  {author} {\bibinfo {author} {\bibfnamefont {A.}~\bibnamefont {{Ben-Naim}}},\ }\bibfield  {title} {\bibinfo {title} {On the {{So-Called Gibbs Paradox}}, and on the {{Real Paradox}}},\ }\href {https://doi.org/10.3390/e9030133} {\bibfield  {journal} {\bibinfo  {journal} {Entropy}\ }\textbf {\bibinfo {volume} {9}},\ \bibinfo {pages} {132} (\bibinfo {year} {2007})}\BibitemShut {NoStop}%
\bibitem [{\citenamefont {Lin}(2008)}]{lin_2008_gibbs}%
  \BibitemOpen
  \bibfield  {author} {\bibinfo {author} {\bibfnamefont {S.-K.}\ \bibnamefont {Lin}},\ }\bibfield  {title} {\bibinfo {title} {Gibbs {{Paradox}} and the {{Concepts}} of {{Information}}, {{Symmetry}}, {{Similarity}} and {{Their Relationship}}},\ }\href {https://doi.org/10.3390/entropy-e10010001} {\bibfield  {journal} {\bibinfo  {journal} {Entropy}\ }\textbf {\bibinfo {volume} {10}},\ \bibinfo {pages} {1} (\bibinfo {year} {2008})}\BibitemShut {NoStop}%
\bibitem [{\citenamefont {Maslov}(2008)}]{maslov_2008_solution}%
  \BibitemOpen
  \bibfield  {author} {\bibinfo {author} {\bibfnamefont {V.~P.}\ \bibnamefont {Maslov}},\ }\bibfield  {title} {\bibinfo {title} {Solution of the {{Gibbs}} paradox in the framework of classical mechanics ({{Statistical Physics}}) and crystallization of the gas {{C}} 60},\ }\href {https://doi.org/10.1134/S0001434608050167} {\bibfield  {journal} {\bibinfo  {journal} {Math Notes}\ }\textbf {\bibinfo {volume} {83}},\ \bibinfo {pages} {716} (\bibinfo {year} {2008})}\BibitemShut {NoStop}%
\bibitem [{\citenamefont {Swendsen}(2008)}]{swendsen_2008_gibbs}%
  \BibitemOpen
  \bibfield  {author} {\bibinfo {author} {\bibfnamefont {R.~H.}\ \bibnamefont {Swendsen}},\ }\bibfield  {title} {\bibinfo {title} {Gibbs' {{Paradox}} and the {{Definition}} of {{Entropy}}},\ }\href {https://doi.org/10.3390/entropy-e10010015} {\bibfield  {journal} {\bibinfo  {journal} {Entropy}\ }\textbf {\bibinfo {volume} {10}},\ \bibinfo {pages} {15} (\bibinfo {year} {2008})}\BibitemShut {NoStop}%
\bibitem [{\citenamefont {Cheng}(2009)}]{cheng_2009_thermodynamics}%
  \BibitemOpen
  \bibfield  {author} {\bibinfo {author} {\bibfnamefont {C.-H.}\ \bibnamefont {Cheng}},\ }\bibfield  {title} {\bibinfo {title} {Thermodynamics of the {{System}} of {{Distinguishable Particles}}},\ }\href {https://doi.org/10.3390/e11030326} {\bibfield  {journal} {\bibinfo  {journal} {Entropy}\ }\textbf {\bibinfo {volume} {11}},\ \bibinfo {pages} {326} (\bibinfo {year} {2009})}\BibitemShut {NoStop}%
\bibitem [{\citenamefont {Enders}(2009)}]{enders_2009_gibbs}%
  \BibitemOpen
  \bibfield  {author} {\bibinfo {author} {\bibfnamefont {P.}~\bibnamefont {Enders}},\ }\bibfield  {title} {\bibinfo {title} {Gibbs' {{Paradox}} in the {{Light}} of {{Newton}}'s {{Notion}} of {{State}}},\ }\href {https://doi.org/10.3390/e11030454} {\bibfield  {journal} {\bibinfo  {journal} {Entropy}\ }\textbf {\bibinfo {volume} {11}},\ \bibinfo {pages} {454} (\bibinfo {year} {2009})}\BibitemShut {NoStop}%
\bibitem [{\citenamefont {Nagle}(2010)}]{nagle_2010_defense}%
  \BibitemOpen
  \bibfield  {author} {\bibinfo {author} {\bibfnamefont {J.~F.}\ \bibnamefont {Nagle}},\ }\bibfield  {title} {\bibinfo {title} {In {{Defense}} of {{Gibbs}} and the {{Traditional Definition}} of the {{Entropy}} of {{Distinguishable Particles}}},\ }\href {https://doi.org/10.3390/e12081936} {\bibfield  {journal} {\bibinfo  {journal} {Entropy}\ }\textbf {\bibinfo {volume} {12}},\ \bibinfo {pages} {1936} (\bibinfo {year} {2010})}\BibitemShut {NoStop}%
\bibitem [{\citenamefont {Peters}(2013)}]{peters_2013_demonstration}%
  \BibitemOpen
  \bibfield  {author} {\bibinfo {author} {\bibfnamefont {H.}~\bibnamefont {Peters}},\ }\bibfield  {title} {\bibinfo {title} {Demonstration and resolution of the {{Gibbs}} paradox of the first kind},\ }\href {https://doi.org/10.1088/0143-0807/35/1/015023} {\bibfield  {journal} {\bibinfo  {journal} {Eur. J. Phys.}\ }\textbf {\bibinfo {volume} {35}},\ \bibinfo {pages} {015023} (\bibinfo {year} {2013})}\BibitemShut {NoStop}%
\bibitem [{\citenamefont {Saunders}(2018)}]{saunders_2018_gibbs}%
  \BibitemOpen
  \bibfield  {author} {\bibinfo {author} {\bibfnamefont {S.}~\bibnamefont {Saunders}},\ }\bibfield  {title} {\bibinfo {title} {The {{Gibbs Paradox}}},\ }\href {https://doi.org/10.3390/e20080552} {\bibfield  {journal} {\bibinfo  {journal} {Entropy}\ }\textbf {\bibinfo {volume} {20}},\ \bibinfo {pages} {552} (\bibinfo {year} {2018})}\BibitemShut {NoStop}%
\bibitem [{\citenamefont {Darrigol}(2018)}]{darrigol_2018_gibbs}%
  \BibitemOpen
  \bibfield  {author} {\bibinfo {author} {\bibfnamefont {O.}~\bibnamefont {Darrigol}},\ }\bibfield  {title} {\bibinfo {title} {The {{Gibbs Paradox}}: {{Early History}} and {{Solutions}}},\ }\href {https://doi.org/10.3390/e20060443} {\bibfield  {journal} {\bibinfo  {journal} {Entropy}\ }\textbf {\bibinfo {volume} {20}},\ \bibinfo {pages} {443} (\bibinfo {year} {2018})}\BibitemShut {NoStop}%
\bibitem [{\citenamefont {Dieks}(2018)}]{dieks_2018_gibbs}%
  \BibitemOpen
  \bibfield  {author} {\bibinfo {author} {\bibfnamefont {D.}~\bibnamefont {Dieks}},\ }\bibfield  {title} {\bibinfo {title} {The {{Gibbs Paradox}} and {{Particle Individuality}}},\ }\href {https://doi.org/10.3390/e20060466} {\bibfield  {journal} {\bibinfo  {journal} {Entropy}\ }\textbf {\bibinfo {volume} {20}},\ \bibinfo {pages} {466} (\bibinfo {year} {2018})}\BibitemShut {NoStop}%
\bibitem [{\citenamefont {Swendsen}(2018)}]{swendsen_2018_probability}%
  \BibitemOpen
  \bibfield  {author} {\bibinfo {author} {\bibfnamefont {R.~H.}\ \bibnamefont {Swendsen}},\ }\bibfield  {title} {\bibinfo {title} {Probability, {{Entropy}}, and {{Gibbs}}' {{Paradox}}(es)},\ }\href {https://doi.org/10.3390/e20060450} {\bibfield  {journal} {\bibinfo  {journal} {Entropy}\ }\textbf {\bibinfo {volume} {20}},\ \bibinfo {pages} {450} (\bibinfo {year} {2018})}\BibitemShut {NoStop}%
\bibitem [{\citenamefont {Ihnatovych}(2023)}]{ihnatovych_2023_gibbs}%
  \BibitemOpen
  \bibfield  {author} {\bibinfo {author} {\bibfnamefont {V.}~\bibnamefont {Ihnatovych}},\ }\href {https://doi.org/10.48550/arXiv.2304.11132} {\bibinfo {title} {The {{Gibbs}} paradox in classical thermodynamics is a consequence of the erroneous attribution of the entropy of an ideal gas to additive quantities}} (\bibinfo {year} {2023}),\ \Eprint {https://arxiv.org/abs/2304.11132} {arXiv:2304.11132 [cond-mat, physics:physics]} \BibitemShut {NoStop}%
\bibitem [{\citenamefont {Baker}(2024)}]{baker_2024_how}%
  \BibitemOpen
  \bibfield  {author} {\bibinfo {author} {\bibfnamefont {J.~E.}\ \bibnamefont {Baker}},\ }\bibfield  {title} {\bibinfo {title} {How muscle solved the {{Gibbs}} paradox},\ }\href {https://doi.org/10.1016/j.bpj.2023.11.1734} {\bibfield  {journal} {\bibinfo  {journal} {Biophysical Journal}\ }\textbf {\bibinfo {volume} {123}},\ \bibinfo {pages} {278a} (\bibinfo {year} {2024})}\BibitemShut {NoStop}%
\bibitem [{\citenamefont {Lairez}(2024)}]{lairez_2024_thermostatistics}%
  \BibitemOpen
  \bibfield  {author} {\bibinfo {author} {\bibfnamefont {D.}~\bibnamefont {Lairez}},\ }\bibfield  {title} {\bibinfo {title} {Thermostatistics, {{Information}}, {{Subjectivity}}: {{Why Is This Association So Disturbing}}?},\ }\href {https://doi.org/10.3390/math12101498} {\bibfield  {journal} {\bibinfo  {journal} {Mathematics}\ }\textbf {\bibinfo {volume} {12}},\ \bibinfo {pages} {1498} (\bibinfo {year} {2024})}\BibitemShut {NoStop}%
\bibitem [{\citenamefont {Yadin}\ \emph {et~al.}(2021)\citenamefont {Yadin}, \citenamefont {Morris},\ and\ \citenamefont {Adesso}}]{yadin_2021_mixing}%
  \BibitemOpen
  \bibfield  {author} {\bibinfo {author} {\bibfnamefont {B.}~\bibnamefont {Yadin}}, \bibinfo {author} {\bibfnamefont {B.}~\bibnamefont {Morris}},\ and\ \bibinfo {author} {\bibfnamefont {G.}~\bibnamefont {Adesso}},\ }\bibfield  {title} {\bibinfo {title} {Mixing indistinguishable systems leads to a quantum {{Gibbs}} paradox},\ }\href {https://doi.org/10.1038/s41467-021-21620-7} {\bibfield  {journal} {\bibinfo  {journal} {Nat Commun}\ }\textbf {\bibinfo {volume} {12}},\ \bibinfo {pages} {1471} (\bibinfo {year} {2021})}\BibitemShut {NoStop}%
\bibitem [{\citenamefont {Yoshida}\ and\ \citenamefont {Nakagawa}(2022)}]{yoshida_2022_work}%
  \BibitemOpen
  \bibfield  {author} {\bibinfo {author} {\bibfnamefont {A.}~\bibnamefont {Yoshida}}\ and\ \bibinfo {author} {\bibfnamefont {N.}~\bibnamefont {Nakagawa}},\ }\bibfield  {title} {\bibinfo {title} {Work relation for determining the mixing free energy of small-scale mixtures},\ }\href {https://doi.org/10.1103/PhysRevResearch.4.023119} {\bibfield  {journal} {\bibinfo  {journal} {Phys. Rev. Res.}\ }\textbf {\bibinfo {volume} {4}},\ \bibinfo {pages} {023119} (\bibinfo {year} {2022})}\BibitemShut {NoStop}%
\bibitem [{\citenamefont {Takakura}(2019)}]{takakura_2019_entropy}%
  \BibitemOpen
  \bibfield  {author} {\bibinfo {author} {\bibfnamefont {R.}~\bibnamefont {Takakura}},\ }\bibfield  {title} {\bibinfo {title} {Entropy of mixing exists only for classical and quantum-like theories among the regular polygon theories},\ }\href {https://doi.org/10.1088/1751-8121/ab4a2e} {\bibfield  {journal} {\bibinfo  {journal} {J. Phys. A: Math. Theor.}\ }\textbf {\bibinfo {volume} {52}},\ \bibinfo {pages} {465302} (\bibinfo {year} {2019})}\BibitemShut {NoStop}%
\bibitem [{\citenamefont {Sasa}\ \emph {et~al.}(2022)\citenamefont {Sasa}, \citenamefont {Hiura}, \citenamefont {Nakagawa},\ and\ \citenamefont {Yoshida}}]{sasa_2022_quasistatic}%
  \BibitemOpen
  \bibfield  {author} {\bibinfo {author} {\bibfnamefont {S.-i.}\ \bibnamefont {Sasa}}, \bibinfo {author} {\bibfnamefont {K.}~\bibnamefont {Hiura}}, \bibinfo {author} {\bibfnamefont {N.}~\bibnamefont {Nakagawa}},\ and\ \bibinfo {author} {\bibfnamefont {A.}~\bibnamefont {Yoshida}},\ }\bibfield  {title} {\bibinfo {title} {Quasi-static {{Decomposition}} and the {{Gibbs Factorial}} in {{Small Thermodynamic Systems}}},\ }\href {https://doi.org/10.1007/s10955-022-02991-7} {\bibfield  {journal} {\bibinfo  {journal} {J Stat Phys}\ }\textbf {\bibinfo {volume} {189}},\ \bibinfo {pages} {31} (\bibinfo {year} {2022})}\BibitemShut {NoStop}%
\bibitem [{\citenamefont {von Neumann}(2010)}]{vonNeumann1929translation}%
  \BibitemOpen
  \bibfield  {author} {\bibinfo {author} {\bibfnamefont {J.}~\bibnamefont {von Neumann}},\ }\bibfield  {title} {\bibinfo {title} {{Proof of the ergodic theorem and the H-theorem in quantum mechanics. Translation of: Beweis des Ergodensatzes und des H-Theorems in der neuen Mechanik}},\ }\href {https://doi.org/10.1140/epjh/e2010-00008-5} {\bibfield  {journal} {\bibinfo  {journal} {European Physical Journal H}\ }\textbf {\bibinfo {volume} {35}},\ \bibinfo {pages} {201} (\bibinfo {year} {2010})}\BibitemShut {NoStop}%
\bibitem [{\citenamefont {von Neumann}(1955)}]{von1955mathematical}%
  \BibitemOpen
  \bibfield  {author} {\bibinfo {author} {\bibfnamefont {J.}~\bibnamefont {von Neumann}},\ }\href {https://doi.org/10.1515/9781400889921} {\emph {\bibinfo {title} {Mathematical foundations of quantum mechanics}}}\ (\bibinfo  {publisher} {\href{http://press.princeton.edu/titles/2113.html}{Princeton university press}},\ \bibinfo {year} {1955})\BibitemShut {NoStop}%
\bibitem [{\citenamefont {{{\v{S}}afr{\'a}nek}}\ \emph {et~al.}(2019{\natexlab{a}})\citenamefont {{{\v{S}}afr{\'a}nek}}, \citenamefont {{Deutsch}},\ and\ \citenamefont {{Aguirre}}}]{safranek2019a}%
  \BibitemOpen
  \bibfield  {author} {\bibinfo {author} {\bibfnamefont {D.}~\bibnamefont {{{\v{S}}afr{\'a}nek}}}, \bibinfo {author} {\bibfnamefont {J.~M.}\ \bibnamefont {{Deutsch}}},\ and\ \bibinfo {author} {\bibfnamefont {A.}~\bibnamefont {{Aguirre}}},\ }\bibfield  {title} {\bibinfo {title} {{Quantum coarse-grained entropy and thermodynamics}},\ }\href {https://doi.org/10.1103/PhysRevA.99.010101} {\bibfield  {journal} {\bibinfo  {journal} {Phys. Rev. A}\ }\textbf {\bibinfo {volume} {99}},\ \bibinfo {eid} {010101} (\bibinfo {year} {2019}{\natexlab{a}})},\ \Eprint {https://arxiv.org/abs/1707.09722} {arXiv:1707.09722 [quant-ph]} \BibitemShut {NoStop}%
\bibitem [{\citenamefont {{{\v{S}}afr{\'a}nek}}\ \emph {et~al.}(2019{\natexlab{b}})\citenamefont {{{\v{S}}afr{\'a}nek}}, \citenamefont {{Deutsch}},\ and\ \citenamefont {{Aguirre}}}]{safranek2019b}%
  \BibitemOpen
  \bibfield  {author} {\bibinfo {author} {\bibfnamefont {D.}~\bibnamefont {{{\v{S}}afr{\'a}nek}}}, \bibinfo {author} {\bibfnamefont {J.~M.}\ \bibnamefont {{Deutsch}}},\ and\ \bibinfo {author} {\bibfnamefont {A.}~\bibnamefont {{Aguirre}}},\ }\bibfield  {title} {\bibinfo {title} {{Quantum coarse-grained entropy and thermalization in closed systems}},\ }\href {https://doi.org/10.1103/PhysRevA.99.012103} {\bibfield  {journal} {\bibinfo  {journal} {Phys. Rev. A}\ }\textbf {\bibinfo {volume} {99}},\ \bibinfo {eid} {012103} (\bibinfo {year} {2019}{\natexlab{b}})},\ \Eprint {https://arxiv.org/abs/1803.00665} {arXiv:1803.00665 [quant-ph]} \BibitemShut {NoStop}%
\bibitem [{\citenamefont {{{\v{S}}afr{\'a}nek}}\ \emph {et~al.}(2021)\citenamefont {{{\v{S}}afr{\'a}nek}}, \citenamefont {{Aguirre}}, \citenamefont {{Schindler}},\ and\ \citenamefont {{Deutsch}}}]{safranek2021brief}%
  \BibitemOpen
  \bibfield  {author} {\bibinfo {author} {\bibfnamefont {D.}~\bibnamefont {{{\v{S}}afr{\'a}nek}}}, \bibinfo {author} {\bibfnamefont {A.}~\bibnamefont {{Aguirre}}}, \bibinfo {author} {\bibfnamefont {J.}~\bibnamefont {{Schindler}}},\ and\ \bibinfo {author} {\bibfnamefont {J.~M.}\ \bibnamefont {{Deutsch}}},\ }\bibfield  {title} {\bibinfo {title} {{A Brief Introduction to Observational Entropy}},\ }\href {https://doi.org/10.1007/s10701-021-00498-x} {\bibfield  {journal} {\bibinfo  {journal} {Foundations of Physics}\ }\textbf {\bibinfo {volume} {51}},\ \bibinfo {eid} {101} (\bibinfo {year} {2021})},\ \Eprint {https://arxiv.org/abs/2008.04409} {arXiv:2008.04409 [quant-ph]} \BibitemShut {NoStop}%
\bibitem [{\citenamefont {Strasberg}\ and\ \citenamefont {Winter}(2021)}]{SW21}%
  \BibitemOpen
  \bibfield  {author} {\bibinfo {author} {\bibfnamefont {P.}~\bibnamefont {Strasberg}}\ and\ \bibinfo {author} {\bibfnamefont {A.}~\bibnamefont {Winter}},\ }\bibfield  {title} {\bibinfo {title} {First and second law of quantum thermodynamics: A consistent derivation based on a microscopic definition of entropy},\ }\href {https://doi.org/10.1103/PRXQuantum.2.030202} {\bibfield  {journal} {\bibinfo  {journal} {PRX Quantum}\ }\textbf {\bibinfo {volume} {2}},\ \bibinfo {pages} {030202} (\bibinfo {year} {2021})}\BibitemShut {NoStop}%
\bibitem [{\citenamefont {{\v{S}}afr{\'a}nek}\ and\ \citenamefont {Thingna}(2023)}]{safranek2021generalized}%
  \BibitemOpen
  \bibfield  {author} {\bibinfo {author} {\bibfnamefont {D.}~\bibnamefont {{\v{S}}afr{\'a}nek}}\ and\ \bibinfo {author} {\bibfnamefont {J.}~\bibnamefont {Thingna}},\ }\bibfield  {title} {\bibinfo {title} {Quantifying information extraction using generalized quantum measurements},\ }\href {https://doi.org/10.1103/PhysRevA.108.032413} {\bibfield  {journal} {\bibinfo  {journal} {Physical Review A}\ }\textbf {\bibinfo {volume} {108}},\ \bibinfo {pages} {032413} (\bibinfo {year} {2023})},\ \Eprint {https://arxiv.org/abs/2007.07246} {arXiv:2007.07246} \BibitemShut {NoStop}%
\bibitem [{\citenamefont {Buscemi}\ \emph {et~al.}(2022)\citenamefont {Buscemi}, \citenamefont {Schindler},\ and\ \citenamefont {{\v{S}}afr{\'a}nek}}]{buscemi2022observational}%
  \BibitemOpen
  \bibfield  {author} {\bibinfo {author} {\bibfnamefont {F.}~\bibnamefont {Buscemi}}, \bibinfo {author} {\bibfnamefont {J.}~\bibnamefont {Schindler}},\ and\ \bibinfo {author} {\bibfnamefont {D.}~\bibnamefont {{\v{S}}afr{\'a}nek}},\ }\bibfield  {title} {\bibinfo {title} {Observational entropy, coarse quantum states, and petz recovery: information-theoretic properties and bounds},\ }\bibfield  {journal} {\bibinfo  {journal} {arXiv:2209.03803}\ }\href {https://doi.org/10.48550/arXiv.2209.03803} {10.48550/arXiv.2209.03803} (\bibinfo {year} {2022})\BibitemShut {NoStop}%
\bibitem [{\citenamefont {Bai}\ \emph {et~al.}(2024)\citenamefont {Bai}, \citenamefont {{\v{S}}afr{\'{a}}nek}, \citenamefont {Schindler}, \citenamefont {Buscemi},\ and\ \citenamefont {Scarani}}]{bai2024observational}%
  \BibitemOpen
  \bibfield  {author} {\bibinfo {author} {\bibfnamefont {G.}~\bibnamefont {Bai}}, \bibinfo {author} {\bibfnamefont {D.}~\bibnamefont {{\v{S}}afr{\'{a}}nek}}, \bibinfo {author} {\bibfnamefont {J.}~\bibnamefont {Schindler}}, \bibinfo {author} {\bibfnamefont {F.}~\bibnamefont {Buscemi}},\ and\ \bibinfo {author} {\bibfnamefont {V.}~\bibnamefont {Scarani}},\ }\bibfield  {title} {\bibinfo {title} {Observational entropy with general quantum priors},\ }\href {https://doi.org/10.22331/q-2024-11-14-1524} {\bibfield  {journal} {\bibinfo  {journal} {{Quantum}}\ }\textbf {\bibinfo {volume} {8}},\ \bibinfo {pages} {1524} (\bibinfo {year} {2024})}\BibitemShut {NoStop}%
\bibitem [{\citenamefont {Riera-Campeny}\ \emph {et~al.}(2021)\citenamefont {Riera-Campeny}, \citenamefont {Sanpera},\ and\ \citenamefont {Strasberg}}]{riera2020finite}%
  \BibitemOpen
  \bibfield  {author} {\bibinfo {author} {\bibfnamefont {A.}~\bibnamefont {Riera-Campeny}}, \bibinfo {author} {\bibfnamefont {A.}~\bibnamefont {Sanpera}},\ and\ \bibinfo {author} {\bibfnamefont {P.}~\bibnamefont {Strasberg}},\ }\bibfield  {title} {\bibinfo {title} {Quantum systems correlated with a finite bath: Nonequilibrium dynamics and thermodynamics},\ }\href {https://doi.org/10.1103/PRXQuantum.2.010340} {\bibfield  {journal} {\bibinfo  {journal} {PRX Quantum}\ }\textbf {\bibinfo {volume} {2}},\ \bibinfo {pages} {010340} (\bibinfo {year} {2021})},\ \Eprint {https://arxiv.org/abs/2008.02184} {arXiv:2008.02184 [quant-ph]} \BibitemShut {NoStop}%
\bibitem [{\citenamefont {{Strasberg}}\ \emph {et~al.}(2021)\citenamefont {{Strasberg}}, \citenamefont {{D{\'\i}az}},\ and\ \citenamefont {{Riera-Campeny}}}]{strasberg2021clausius}%
  \BibitemOpen
  \bibfield  {author} {\bibinfo {author} {\bibfnamefont {P.}~\bibnamefont {{Strasberg}}}, \bibinfo {author} {\bibfnamefont {M.~G.}\ \bibnamefont {{D{\'\i}az}}},\ and\ \bibinfo {author} {\bibfnamefont {A.}~\bibnamefont {{Riera-Campeny}}},\ }\bibfield  {title} {\bibinfo {title} {{Clausius inequality for finite baths reveals universal efficiency improvements}},\ }\href {https://doi.org/10.1103/PhysRevE.104.L022103} {\bibfield  {journal} {\bibinfo  {journal} {Phys. Rev. E}\ }\textbf {\bibinfo {volume} {104}},\ \bibinfo {eid} {L022103} (\bibinfo {year} {2021})},\ \Eprint {https://arxiv.org/abs/2012.03262} {arXiv:2012.03262 [quant-ph]} \BibitemShut {NoStop}%
\bibitem [{\citenamefont {PG}\ \emph {et~al.}(2023)\citenamefont {PG}, \citenamefont {Modak},\ and\ \citenamefont {Aravinda}}]{sreeram2023witnessing}%
  \BibitemOpen
  \bibfield  {author} {\bibinfo {author} {\bibfnamefont {S.}~\bibnamefont {PG}}, \bibinfo {author} {\bibfnamefont {R.}~\bibnamefont {Modak}},\ and\ \bibinfo {author} {\bibfnamefont {S.}~\bibnamefont {Aravinda}},\ }\bibfield  {title} {\bibinfo {title} {Witnessing quantum chaos using observational entropy},\ }\href {https://doi.org/10.1103/PhysRevE.107.064204} {\bibfield  {journal} {\bibinfo  {journal} {Phys. Rev. E}\ }\textbf {\bibinfo {volume} {107}},\ \bibinfo {pages} {064204} (\bibinfo {year} {2023})}\BibitemShut {NoStop}%
\bibitem [{\citenamefont {{PG}}\ \emph {et~al.}(2024)\citenamefont {{PG}}, \citenamefont {{Bharathi Kannan}}, \citenamefont {{Harshini Tekur}},\ and\ \citenamefont {{Santhanam}}}]{sreeram2024dichotomy}%
  \BibitemOpen
  \bibfield  {author} {\bibinfo {author} {\bibfnamefont {S.}~\bibnamefont {{PG}}}, \bibinfo {author} {\bibfnamefont {J.}~\bibnamefont {{Bharathi Kannan}}}, \bibinfo {author} {\bibfnamefont {S.}~\bibnamefont {{Harshini Tekur}}},\ and\ \bibinfo {author} {\bibfnamefont {M.~S.}\ \bibnamefont {{Santhanam}}},\ }\bibfield  {title} {\bibinfo {title} {{Dichotomy in the effect of chaos on ergotropy}},\ }\href {https://doi.org/10.48550/arXiv.2409.16587} {\bibfield  {journal} {\bibinfo  {journal} {arXiv e-prints}\ ,\ \bibinfo {eid} {arXiv:2409.16587}} (\bibinfo {year} {2024})},\ \Eprint {https://arxiv.org/abs/2409.16587} {arXiv:2409.16587 [quant-ph]} \BibitemShut {NoStop}%
\bibitem [{\citenamefont {{Chakraborty}}\ \emph {et~al.}(2024)\citenamefont {{Chakraborty}}, \citenamefont {{Das}}, \citenamefont {{Ghorui}}, \citenamefont {{Hazra}},\ and\ \citenamefont {{Singh}}}]{chakraborty2024sample}%
  \BibitemOpen
  \bibfield  {author} {\bibinfo {author} {\bibfnamefont {S.}~\bibnamefont {{Chakraborty}}}, \bibinfo {author} {\bibfnamefont {S.}~\bibnamefont {{Das}}}, \bibinfo {author} {\bibfnamefont {A.}~\bibnamefont {{Ghorui}}}, \bibinfo {author} {\bibfnamefont {S.}~\bibnamefont {{Hazra}}},\ and\ \bibinfo {author} {\bibfnamefont {U.}~\bibnamefont {{Singh}}},\ }\bibfield  {title} {\bibinfo {title} {{Sample Complexity of Black Box Work Extraction}},\ }\href {https://doi.org/10.48550/arXiv.2412.02673} {\bibfield  {journal} {\bibinfo  {journal} {arXiv e-prints}\ ,\ \bibinfo {eid} {arXiv:2412.02673}} (\bibinfo {year} {2024})},\ \Eprint {https://arxiv.org/abs/2412.02673} {arXiv:2412.02673 [quant-ph]} \BibitemShut {NoStop}%
\bibitem [{\citenamefont {Nagasawa}\ \emph {et~al.}(2024)\citenamefont {Nagasawa}, \citenamefont {Kato}, \citenamefont {Wakakuwa},\ and\ \citenamefont {Buscemi}}]{nagasawa2024generic}%
  \BibitemOpen
  \bibfield  {author} {\bibinfo {author} {\bibfnamefont {T.}~\bibnamefont {Nagasawa}}, \bibinfo {author} {\bibfnamefont {K.}~\bibnamefont {Kato}}, \bibinfo {author} {\bibfnamefont {E.}~\bibnamefont {Wakakuwa}},\ and\ \bibinfo {author} {\bibfnamefont {F.}~\bibnamefont {Buscemi}},\ }\bibfield  {title} {\bibinfo {title} {Generic increase of observational entropy in isolated systems},\ }\href {https://doi.org/10.1103/PhysRevResearch.6.043327} {\bibfield  {journal} {\bibinfo  {journal} {Phys. Rev. Res.}\ }\textbf {\bibinfo {volume} {6}},\ \bibinfo {pages} {043327} (\bibinfo {year} {2024})}\BibitemShut {NoStop}%
\bibitem [{\citenamefont {Xuereb}\ \emph {et~al.}(2024)\citenamefont {Xuereb}, \citenamefont {Junior}, \citenamefont {Clivaz}, \citenamefont {Bakhshinezhad},\ and\ \citenamefont {Huber}}]{Xuereb2024}%
  \BibitemOpen
  \bibfield  {author} {\bibinfo {author} {\bibfnamefont {J.}~\bibnamefont {Xuereb}}, \bibinfo {author} {\bibfnamefont {A.~d.~O.}\ \bibnamefont {Junior}}, \bibinfo {author} {\bibfnamefont {F.}~\bibnamefont {Clivaz}}, \bibinfo {author} {\bibfnamefont {P.}~\bibnamefont {Bakhshinezhad}},\ and\ \bibinfo {author} {\bibfnamefont {M.}~\bibnamefont {Huber}},\ }\bibfield  {title} {\bibinfo {title} {{What resources do agents need to acquire knowledge in Quantum Thermodynamics?}},\ }\href {https://doi.org/10.48550/arXiv.2410.18167} {\bibfield  {journal} {\bibinfo  {journal} {arXiv e-prints}\ ,\ \bibinfo {pages} {2410.18167}} (\bibinfo {year} {2024})},\ \Eprint {https://arxiv.org/abs/2410.18167} {arXiv:2410.18167} \BibitemShut {NoStop}%
\bibitem [{\citenamefont {Meier}\ \emph {et~al.}(2025)\citenamefont {Meier}, \citenamefont {Rivlin}, \citenamefont {Debarba}, \citenamefont {Xuereb}, \citenamefont {Huber},\ and\ \citenamefont {Lock}}]{meier2025emergence}%
  \BibitemOpen
  \bibfield  {author} {\bibinfo {author} {\bibfnamefont {F.}~\bibnamefont {Meier}}, \bibinfo {author} {\bibfnamefont {T.}~\bibnamefont {Rivlin}}, \bibinfo {author} {\bibfnamefont {T.}~\bibnamefont {Debarba}}, \bibinfo {author} {\bibfnamefont {J.}~\bibnamefont {Xuereb}}, \bibinfo {author} {\bibfnamefont {M.}~\bibnamefont {Huber}},\ and\ \bibinfo {author} {\bibfnamefont {M.~P.}\ \bibnamefont {Lock}},\ }\bibfield  {title} {\bibinfo {title} {Emergence of a second law of thermodynamics in isolated quantum systems},\ }\href {https://doi.org/10.1103/PRXQuantum.6.010309} {\bibfield  {journal} {\bibinfo  {journal} {PRX Quantum}\ }\textbf {\bibinfo {volume} {6}},\ \bibinfo {pages} {010309} (\bibinfo {year} {2025})}\BibitemShut {NoStop}%
\bibitem [{\citenamefont {{Schindler}}\ \emph {et~al.}(2025)\citenamefont {{Schindler}}, \citenamefont {{Strasberg}}, \citenamefont {{Galke}}, \citenamefont {{Winter}},\ and\ \citenamefont {{Jabbour}}}]{schindler2025unification}%
  \BibitemOpen
  \bibfield  {author} {\bibinfo {author} {\bibfnamefont {J.}~\bibnamefont {{Schindler}}}, \bibinfo {author} {\bibfnamefont {P.}~\bibnamefont {{Strasberg}}}, \bibinfo {author} {\bibfnamefont {N.}~\bibnamefont {{Galke}}}, \bibinfo {author} {\bibfnamefont {A.}~\bibnamefont {{Winter}}},\ and\ \bibinfo {author} {\bibfnamefont {M.~G.}\ \bibnamefont {{Jabbour}}},\ }\bibfield  {title} {\bibinfo {title} {{Unification of observational entropy with maximum entropy principles}},\ }\href {https://doi.org/10.48550/arXiv.2503.15612} {\bibfield  {journal} {\bibinfo  {journal} {arXiv e-prints}\ ,\ \bibinfo {eid} {arXiv:2503.15612}} (\bibinfo {year} {2025})},\ \Eprint {https://arxiv.org/abs/2503.15612} {arXiv:2503.15612 [quant-ph]} \BibitemShut {NoStop}%
\bibitem [{\citenamefont {\ifmmode~\check{S}\else \v{S}\fi{}afr\'anek}\ \emph {et~al.}(2023)\citenamefont {\ifmmode~\check{S}\else \v{S}\fi{}afr\'anek}, \citenamefont {Rosa},\ and\ \citenamefont {Binder}}]{safranek2023work}%
  \BibitemOpen
  \bibfield  {author} {\bibinfo {author} {\bibfnamefont {D.}~\bibnamefont {\ifmmode~\check{S}\else \v{S}\fi{}afr\'anek}}, \bibinfo {author} {\bibfnamefont {D.}~\bibnamefont {Rosa}},\ and\ \bibinfo {author} {\bibfnamefont {F.~C.}\ \bibnamefont {Binder}},\ }\bibfield  {title} {\bibinfo {title} {Work extraction from unknown quantum sources},\ }\href {https://doi.org/10.1103/PhysRevLett.130.210401} {\bibfield  {journal} {\bibinfo  {journal} {Phys. Rev. Lett.}\ }\textbf {\bibinfo {volume} {130}},\ \bibinfo {pages} {210401} (\bibinfo {year} {2023})}\BibitemShut {NoStop}%
\bibitem [{\citenamefont {Bartlett}\ \emph {et~al.}(2007)\citenamefont {Bartlett}, \citenamefont {Rudolph},\ and\ \citenamefont {Spekkens}}]{Bartlett2007}%
  \BibitemOpen
  \bibfield  {author} {\bibinfo {author} {\bibfnamefont {S.~D.}\ \bibnamefont {Bartlett}}, \bibinfo {author} {\bibfnamefont {T.}~\bibnamefont {Rudolph}},\ and\ \bibinfo {author} {\bibfnamefont {R.~W.}\ \bibnamefont {Spekkens}},\ }\bibfield  {title} {\bibinfo {title} {{Reference frames, superselection rules, and quantum information}},\ }\href {https://doi.org/10.1103/RevModPhys.79.555} {\bibfield  {journal} {\bibinfo  {journal} {Reviews of Modern Physics}\ }\textbf {\bibinfo {volume} {79}},\ \bibinfo {pages} {555} (\bibinfo {year} {2007})},\ \Eprint {https://arxiv.org/abs/0610030v3} {arXiv:0610030v3 [arXiv:quant-ph]} \BibitemShut {NoStop}%
\bibitem [{\citenamefont {{{\v{S}}afr{\'a}nek}}(2023)}]{safranek2023ergotropic}%
  \BibitemOpen
  \bibfield  {author} {\bibinfo {author} {\bibfnamefont {D.}~\bibnamefont {{{\v{S}}afr{\'a}nek}}},\ }\bibfield  {title} {\bibinfo {title} {{Ergotropic interpretation of entanglement entropy}},\ }\href {https://doi.org/10.48550/arXiv.2306.08987} {\bibfield  {journal} {\bibinfo  {journal} {arXiv e-prints}\ ,\ \bibinfo {eid} {arXiv:2306.08987}} (\bibinfo {year} {2023})},\ \Eprint {https://arxiv.org/abs/2306.08987} {arXiv:2306.08987 [quant-ph]} \BibitemShut {NoStop}%
\bibitem [{\citenamefont {{\v S}afr{\'a}nek}\ and\ \citenamefont {Rosa}(2023)}]{safranek_2023_measuring}%
  \BibitemOpen
  \bibfield  {author} {\bibinfo {author} {\bibfnamefont {D.}~\bibnamefont {{\v S}afr{\'a}nek}}\ and\ \bibinfo {author} {\bibfnamefont {D.}~\bibnamefont {Rosa}},\ }\bibfield  {title} {\bibinfo {title} {Measuring energy by measuring any other observable},\ }\href {https://doi.org/10.1103/PhysRevA.108.022208} {\bibfield  {journal} {\bibinfo  {journal} {Phys. Rev. A}\ }\textbf {\bibinfo {volume} {108}},\ \bibinfo {pages} {022208} (\bibinfo {year} {2023})}\BibitemShut {NoStop}%
\bibitem [{Note1()}]{Note1}%
  \BibitemOpen
  \bibinfo {note} {In the example of the mixing paradox below, in the initial situation of the macroscopic non-equilibirum state with equal-sized boxes and non-dilute scenario, $\Delta S$ is about a third of $S$, so this condition is not satisfied. Including higher order terms will yield better agreement with the exact formula, Eq.~\protect \eqref {eq:work_difference}, and we employ empirical approaches to do that. For the unequal sized boxes, or sufficiently dilute gas, this condition will generally be satisfied.}\BibitemShut {Stop}%
\bibitem [{Note2()}]{Note2}%
  \BibitemOpen
  \bibinfo {note} {See reviews \cite {sagawa_2012_thermodynamics,parrondo_2014_thermodynamics,davies_2021_the} and papers for quantum~\cite {Berut_2012_experimental,koski_2014_experimental,peterson_2016_experimental,boyd_2016_maxwell,boyd_2017_transient,masuyama_2018_information,vanhorne_2020_single-atom} (Maxwell's Demon, Szilard Engine and Landauer Principle), classical~\cite {marathe_2010_cooling,tamir2018experimental,vaikuntanathan_2011_modelling}, and active matter \cite {sosuke_2013_information,malgaretti_2022_szilard,chor_2023_many-body} systems.}\BibitemShut {Stop}%
\bibitem [{\citenamefont {Yang}\ and\ \citenamefont {Lin}(2001)}]{yang2001low}%
  \BibitemOpen
  \bibfield  {author} {\bibinfo {author} {\bibfnamefont {H.}~\bibnamefont {Yang}}\ and\ \bibinfo {author} {\bibfnamefont {J.-Y.}\ \bibnamefont {Lin}},\ }\bibfield  {title} {\bibinfo {title} {Low temperature specific heat studies on the pairing states of high-tc superconductors: a brief review},\ }\href {https://doi.org/10.1016/S0022-3697(01)00118-4} {\bibfield  {journal} {\bibinfo  {journal} {Journal of Physics and Chemistry of Solids}\ }\textbf {\bibinfo {volume} {62}},\ \bibinfo {pages} {1861} (\bibinfo {year} {2001})}\BibitemShut {NoStop}%
\bibitem [{\citenamefont {Lipa}\ \emph {et~al.}(2003)\citenamefont {Lipa}, \citenamefont {Nissen}, \citenamefont {Stricker}, \citenamefont {Swanson},\ and\ \citenamefont {Chui}}]{lipa2003specific}%
  \BibitemOpen
  \bibfield  {author} {\bibinfo {author} {\bibfnamefont {J.~A.}\ \bibnamefont {Lipa}}, \bibinfo {author} {\bibfnamefont {J.~A.}\ \bibnamefont {Nissen}}, \bibinfo {author} {\bibfnamefont {D.~A.}\ \bibnamefont {Stricker}}, \bibinfo {author} {\bibfnamefont {D.~R.}\ \bibnamefont {Swanson}},\ and\ \bibinfo {author} {\bibfnamefont {T.~C.~P.}\ \bibnamefont {Chui}},\ }\bibfield  {title} {\bibinfo {title} {Specific heat of liquid helium in zero gravity very near the lambda point},\ }\href {https://doi.org/10.1103/PhysRevB.68.174518} {\bibfield  {journal} {\bibinfo  {journal} {Phys. Rev. B}\ }\textbf {\bibinfo {volume} {68}},\ \bibinfo {pages} {174518} (\bibinfo {year} {2003})}\BibitemShut {NoStop}%
\bibitem [{\citenamefont {Lashley}\ \emph {et~al.}(2003)\citenamefont {Lashley}, \citenamefont {Hundley}, \citenamefont {Migliori}, \citenamefont {Sarrao}, \citenamefont {Pagliuso}, \citenamefont {Darling}, \citenamefont {Jaime}, \citenamefont {Cooley}, \citenamefont {Hults}, \citenamefont {Morales}, \citenamefont {Thoma}, \citenamefont {Smith}, \citenamefont {Boerio-Goates}, \citenamefont {Woodfield}, \citenamefont {Stewart}, \citenamefont {Fisher},\ and\ \citenamefont {Phillips}}]{lashley2003critical}%
  \BibitemOpen
  \bibfield  {author} {\bibinfo {author} {\bibfnamefont {J.}~\bibnamefont {Lashley}}, \bibinfo {author} {\bibfnamefont {M.}~\bibnamefont {Hundley}}, \bibinfo {author} {\bibfnamefont {A.}~\bibnamefont {Migliori}}, \bibinfo {author} {\bibfnamefont {J.}~\bibnamefont {Sarrao}}, \bibinfo {author} {\bibfnamefont {P.}~\bibnamefont {Pagliuso}}, \bibinfo {author} {\bibfnamefont {T.}~\bibnamefont {Darling}}, \bibinfo {author} {\bibfnamefont {M.}~\bibnamefont {Jaime}}, \bibinfo {author} {\bibfnamefont {J.}~\bibnamefont {Cooley}}, \bibinfo {author} {\bibfnamefont {W.}~\bibnamefont {Hults}}, \bibinfo {author} {\bibfnamefont {L.}~\bibnamefont {Morales}}, \bibinfo {author} {\bibfnamefont {D.}~\bibnamefont {Thoma}}, \bibinfo {author} {\bibfnamefont {J.}~\bibnamefont {Smith}}, \bibinfo {author} {\bibfnamefont {J.}~\bibnamefont {Boerio-Goates}}, \bibinfo {author} {\bibfnamefont {B.}~\bibnamefont {Woodfield}}, \bibinfo {author} {\bibfnamefont {G.}~\bibnamefont {Stewart}}, \bibinfo {author} {\bibfnamefont {R.}~\bibnamefont
  {Fisher}},\ and\ \bibinfo {author} {\bibfnamefont {N.}~\bibnamefont {Phillips}},\ }\bibfield  {title} {\bibinfo {title} {Critical examination of heat capacity measurements made on a quantum design physical property measurement system},\ }\href {https://doi.org/https://doi.org/10.1016/S0011-2275(03)00092-4} {\bibfield  {journal} {\bibinfo  {journal} {Cryogenics}\ }\textbf {\bibinfo {volume} {43}},\ \bibinfo {pages} {369} (\bibinfo {year} {2003})}\BibitemShut {NoStop}%
\bibitem [{\citenamefont {Liang}\ \emph {et~al.}(2015)\citenamefont {Liang}, \citenamefont {Koohpayeh}, \citenamefont {Krizan}, \citenamefont {McQueen}, \citenamefont {Cava},\ and\ \citenamefont {Ong}}]{liang_2015_heat}%
  \BibitemOpen
  \bibfield  {author} {\bibinfo {author} {\bibfnamefont {T.}~\bibnamefont {Liang}}, \bibinfo {author} {\bibfnamefont {S.~M.}\ \bibnamefont {Koohpayeh}}, \bibinfo {author} {\bibfnamefont {J.~W.}\ \bibnamefont {Krizan}}, \bibinfo {author} {\bibfnamefont {T.~M.}\ \bibnamefont {McQueen}}, \bibinfo {author} {\bibfnamefont {R.~J.}\ \bibnamefont {Cava}},\ and\ \bibinfo {author} {\bibfnamefont {N.~P.}\ \bibnamefont {Ong}},\ }\bibfield  {title} {\bibinfo {title} {Heat capacity peak at the quantum critical point of the transverse ising magnet conb2o6},\ }\href {https://doi.org/10.1038/ncomms8611} {\bibfield  {journal} {\bibinfo  {journal} {Nature Communications}\ }\textbf {\bibinfo {volume} {6}},\ \bibinfo {pages} {7611} (\bibinfo {year} {2015})}\BibitemShut {NoStop}%
\bibitem [{\citenamefont {Zhang}\ \emph {et~al.}(2023)\citenamefont {Zhang}, \citenamefont {Berg},\ and\ \citenamefont {Chubukov}}]{zhang2023free}%
  \BibitemOpen
  \bibfield  {author} {\bibinfo {author} {\bibfnamefont {S.-S.}\ \bibnamefont {Zhang}}, \bibinfo {author} {\bibfnamefont {E.}~\bibnamefont {Berg}},\ and\ \bibinfo {author} {\bibfnamefont {A.~V.}\ \bibnamefont {Chubukov}},\ }\bibfield  {title} {\bibinfo {title} {Free energy and specific heat near a quantum critical point of a metal},\ }\href {https://doi.org/10.1103/PhysRevB.107.144507} {\bibfield  {journal} {\bibinfo  {journal} {Phys. Rev. B}\ }\textbf {\bibinfo {volume} {107}},\ \bibinfo {pages} {144507} (\bibinfo {year} {2023})}\BibitemShut {NoStop}%
\bibitem [{\citenamefont {{{\v{S}}afr{\'a}nek}}\ \emph {et~al.}(2020)\citenamefont {{{\v{S}}afr{\'a}nek}}, \citenamefont {{Aguirre}},\ and\ \citenamefont {{Deutsch}}}]{safranek2020classical}%
  \BibitemOpen
  \bibfield  {author} {\bibinfo {author} {\bibfnamefont {D.}~\bibnamefont {{{\v{S}}afr{\'a}nek}}}, \bibinfo {author} {\bibfnamefont {A.}~\bibnamefont {{Aguirre}}},\ and\ \bibinfo {author} {\bibfnamefont {J.~M.}\ \bibnamefont {{Deutsch}}},\ }\bibfield  {title} {\bibinfo {title} {{Classical dynamical coarse-grained entropy and comparison with the quantum version}},\ }\href {https://doi.org/10.1103/PhysRevE.102.032106} {\bibfield  {journal} {\bibinfo  {journal} {Phys. Rev. E}\ }\textbf {\bibinfo {volume} {102}},\ \bibinfo {eid} {032106} (\bibinfo {year} {2020})},\ \Eprint {https://arxiv.org/abs/1905.03841} {arXiv:1905.03841 [cond-mat.stat-mech]} \BibitemShut {NoStop}%
\bibitem [{\citenamefont {{Teixid{\'o}-Bonfill}}\ \emph {et~al.}(2025)\citenamefont {{Teixid{\'o}-Bonfill}}, \citenamefont {{Schindler}},\ and\ \citenamefont {{{\v{S}}afr{\'a}nek}}}]{teixido2025entropic}%
  \BibitemOpen
  \bibfield  {author} {\bibinfo {author} {\bibfnamefont {A.}~\bibnamefont {{Teixid{\'o}-Bonfill}}}, \bibinfo {author} {\bibfnamefont {J.}~\bibnamefont {{Schindler}}},\ and\ \bibinfo {author} {\bibfnamefont {D.}~\bibnamefont {{{\v{S}}afr{\'a}nek}}},\ }\bibfield  {title} {\bibinfo {title} {{Entropic partial orderings of quantum measurements}},\ }\href {https://doi.org/10.1088/1402-4896/ad977c} {\bibfield  {journal} {\bibinfo  {journal} {Physica Scripta}\ }\textbf {\bibinfo {volume} {100}},\ \bibinfo {eid} {015298} (\bibinfo {year} {2025})},\ \Eprint {https://arxiv.org/abs/2310.14086} {arXiv:2310.14086 [quant-ph]} \BibitemShut {NoStop}%
\bibitem [{\citenamefont {{Majidy}}\ \emph {et~al.}(2023)\citenamefont {{Majidy}}, \citenamefont {{Braasch}}, \citenamefont {{Lasek}}, \citenamefont {{Upadhyaya}}, \citenamefont {{Kalev}},\ and\ \citenamefont {{Yunger Halpern}}}]{majidy2023noncommuting}%
  \BibitemOpen
  \bibfield  {author} {\bibinfo {author} {\bibfnamefont {S.}~\bibnamefont {{Majidy}}}, \bibinfo {author} {\bibfnamefont {W.~F.}\ \bibnamefont {{Braasch}}}, \bibinfo {author} {\bibfnamefont {A.}~\bibnamefont {{Lasek}}}, \bibinfo {author} {\bibfnamefont {T.}~\bibnamefont {{Upadhyaya}}}, \bibinfo {author} {\bibfnamefont {A.}~\bibnamefont {{Kalev}}},\ and\ \bibinfo {author} {\bibfnamefont {N.}~\bibnamefont {{Yunger Halpern}}},\ }\bibfield  {title} {\bibinfo {title} {{Noncommuting conserved charges in quantum thermodynamics and beyond}},\ }\href {https://doi.org/10.1038/s42254-023-00641-9} {\bibfield  {journal} {\bibinfo  {journal} {Nature Reviews Physics}\ }\textbf {\bibinfo {volume} {5}},\ \bibinfo {pages} {689} (\bibinfo {year} {2023})},\ \Eprint {https://arxiv.org/abs/2306.00054} {arXiv:2306.00054 [quant-ph]} \BibitemShut {NoStop}%
\bibitem [{\citenamefont {{Lasek}}\ \emph {et~al.}(2024)\citenamefont {{Lasek}}, \citenamefont {{Noh}}, \citenamefont {{LeSchack}},\ and\ \citenamefont {{Yunger Halpern}}}]{lasek2024numerical}%
  \BibitemOpen
  \bibfield  {author} {\bibinfo {author} {\bibfnamefont {A.}~\bibnamefont {{Lasek}}}, \bibinfo {author} {\bibfnamefont {J.~D.}\ \bibnamefont {{Noh}}}, \bibinfo {author} {\bibfnamefont {J.}~\bibnamefont {{LeSchack}}},\ and\ \bibinfo {author} {\bibfnamefont {N.}~\bibnamefont {{Yunger Halpern}}},\ }\bibfield  {title} {\bibinfo {title} {{Numerical evidence for the non-Abelian eigenstate thermalization hypothesis}},\ }\href {https://doi.org/10.48550/arXiv.2412.07838} {\bibfield  {journal} {\bibinfo  {journal} {arXiv e-prints}\ ,\ \bibinfo {eid} {arXiv:2412.07838}} (\bibinfo {year} {2024})},\ \Eprint {https://arxiv.org/abs/2412.07838} {arXiv:2412.07838 [quant-ph]} \BibitemShut {NoStop}%
\bibitem [{\citenamefont {{\v S}afr{\'a}nek}\ \emph {et~al.}(2021)\citenamefont {{\v S}afr{\'a}nek}, \citenamefont {Aguirre}, \citenamefont {Schindler},\ and\ \citenamefont {Deutsch}}]{safranek_2021_brief}%
  \BibitemOpen
  \bibfield  {author} {\bibinfo {author} {\bibfnamefont {D.}~\bibnamefont {{\v S}afr{\'a}nek}}, \bibinfo {author} {\bibfnamefont {A.}~\bibnamefont {Aguirre}}, \bibinfo {author} {\bibfnamefont {J.}~\bibnamefont {Schindler}},\ and\ \bibinfo {author} {\bibfnamefont {J.~M.}\ \bibnamefont {Deutsch}},\ }\bibfield  {title} {\bibinfo {title} {A {{Brief Introduction}} to {{Observational Entropy}}},\ }\href {https://doi.org/10.1007/s10701-021-00498-x} {\bibfield  {journal} {\bibinfo  {journal} {Found Phys}\ }\textbf {\bibinfo {volume} {51}},\ \bibinfo {pages} {101} (\bibinfo {year} {2021})}\BibitemShut {NoStop}%
\bibitem [{\citenamefont {Myatt}\ \emph {et~al.}(1997)\citenamefont {Myatt}, \citenamefont {Burt}, \citenamefont {Ghrist}, \citenamefont {Cornell},\ and\ \citenamefont {Wieman}}]{myatt_1997_production}%
  \BibitemOpen
  \bibfield  {author} {\bibinfo {author} {\bibfnamefont {C.~J.}\ \bibnamefont {Myatt}}, \bibinfo {author} {\bibfnamefont {E.~A.}\ \bibnamefont {Burt}}, \bibinfo {author} {\bibfnamefont {R.~W.}\ \bibnamefont {Ghrist}}, \bibinfo {author} {\bibfnamefont {E.~A.}\ \bibnamefont {Cornell}},\ and\ \bibinfo {author} {\bibfnamefont {C.~E.}\ \bibnamefont {Wieman}},\ }\bibfield  {title} {\bibinfo {title} {Production of {{Two Overlapping Bose-Einstein Condensates}} by {{Sympathetic Cooling}}},\ }\href {https://doi.org/10.1103/PhysRevLett.78.586} {\bibfield  {journal} {\bibinfo  {journal} {Phys. Rev. Lett.}\ }\textbf {\bibinfo {volume} {78}},\ \bibinfo {pages} {586} (\bibinfo {year} {1997})}\BibitemShut {NoStop}%
\bibitem [{\citenamefont {Hall}\ \emph {et~al.}(1998)\citenamefont {Hall}, \citenamefont {Matthews}, \citenamefont {Ensher}, \citenamefont {Wieman},\ and\ \citenamefont {Cornell}}]{hall_1998_dynamics}%
  \BibitemOpen
  \bibfield  {author} {\bibinfo {author} {\bibfnamefont {D.~S.}\ \bibnamefont {Hall}}, \bibinfo {author} {\bibfnamefont {M.~R.}\ \bibnamefont {Matthews}}, \bibinfo {author} {\bibfnamefont {J.~R.}\ \bibnamefont {Ensher}}, \bibinfo {author} {\bibfnamefont {C.~E.}\ \bibnamefont {Wieman}},\ and\ \bibinfo {author} {\bibfnamefont {E.~A.}\ \bibnamefont {Cornell}},\ }\bibfield  {title} {\bibinfo {title} {Dynamics of {{Component Separation}} in a {{Binary Mixture}} of {{Bose-Einstein Condensates}}},\ }\href {https://doi.org/10.1103/PhysRevLett.81.1539} {\bibfield  {journal} {\bibinfo  {journal} {Phys. Rev. Lett.}\ }\textbf {\bibinfo {volume} {81}},\ \bibinfo {pages} {1539} (\bibinfo {year} {1998})}\BibitemShut {NoStop}%
\bibitem [{\citenamefont {Wheeler}\ \emph {et~al.}(2004)\citenamefont {Wheeler}, \citenamefont {Mertes}, \citenamefont {Erwin},\ and\ \citenamefont {Hall}}]{wheeler_2004_spontaneous}%
  \BibitemOpen
  \bibfield  {author} {\bibinfo {author} {\bibfnamefont {M.~H.}\ \bibnamefont {Wheeler}}, \bibinfo {author} {\bibfnamefont {K.~M.}\ \bibnamefont {Mertes}}, \bibinfo {author} {\bibfnamefont {J.~D.}\ \bibnamefont {Erwin}},\ and\ \bibinfo {author} {\bibfnamefont {D.~S.}\ \bibnamefont {Hall}},\ }\bibfield  {title} {\bibinfo {title} {Spontaneous {{Macroscopic Spin Polarization}} in {{Independent Spinor Bose-Einstein Condensates}}},\ }\href {https://doi.org/10.1103/PhysRevLett.93.170402} {\bibfield  {journal} {\bibinfo  {journal} {Phys. Rev. Lett.}\ }\textbf {\bibinfo {volume} {93}},\ \bibinfo {pages} {170402} (\bibinfo {year} {2004})}\BibitemShut {NoStop}%
\bibitem [{\citenamefont {Lin}\ \emph {et~al.}(2011)\citenamefont {Lin}, \citenamefont {{Jim{\'e}nez-Garc{\'i}a}},\ and\ \citenamefont {Spielman}}]{lin_2011_spin}%
  \BibitemOpen
  \bibfield  {author} {\bibinfo {author} {\bibfnamefont {Y.-J.}\ \bibnamefont {Lin}}, \bibinfo {author} {\bibfnamefont {K.}~\bibnamefont {{Jim{\'e}nez-Garc{\'i}a}}},\ and\ \bibinfo {author} {\bibfnamefont {I.~B.}\ \bibnamefont {Spielman}},\ }\bibfield  {title} {\bibinfo {title} {Spin--orbit-coupled {{Bose}}--{{Einstein}} condensates},\ }\href {https://doi.org/10.1038/nature09887} {\bibfield  {journal} {\bibinfo  {journal} {Nature}\ }\textbf {\bibinfo {volume} {471}},\ \bibinfo {pages} {83} (\bibinfo {year} {2011})}\BibitemShut {NoStop}%
\bibitem [{\citenamefont {Greif}\ \emph {et~al.}(2016)\citenamefont {Greif}, \citenamefont {Parsons}, \citenamefont {Mazurenko}, \citenamefont {Chiu}, \citenamefont {Blatt}, \citenamefont {Huber}, \citenamefont {Ji},\ and\ \citenamefont {Greiner}}]{greif_2016_siteresolved}%
  \BibitemOpen
  \bibfield  {author} {\bibinfo {author} {\bibfnamefont {D.}~\bibnamefont {Greif}}, \bibinfo {author} {\bibfnamefont {M.~F.}\ \bibnamefont {Parsons}}, \bibinfo {author} {\bibfnamefont {A.}~\bibnamefont {Mazurenko}}, \bibinfo {author} {\bibfnamefont {C.~S.}\ \bibnamefont {Chiu}}, \bibinfo {author} {\bibfnamefont {S.}~\bibnamefont {Blatt}}, \bibinfo {author} {\bibfnamefont {F.}~\bibnamefont {Huber}}, \bibinfo {author} {\bibfnamefont {G.}~\bibnamefont {Ji}},\ and\ \bibinfo {author} {\bibfnamefont {M.}~\bibnamefont {Greiner}},\ }\bibfield  {title} {\bibinfo {title} {Site-resolved imaging of a fermionic {{Mott}} insulator},\ }\href {https://doi.org/10.1126/science.aad9041} {\bibfield  {journal} {\bibinfo  {journal} {Science}\ }\textbf {\bibinfo {volume} {351}},\ \bibinfo {pages} {953} (\bibinfo {year} {2016})}\BibitemShut {NoStop}%
\bibitem [{\citenamefont {Parsons}\ \emph {et~al.}(2016)\citenamefont {Parsons}, \citenamefont {Mazurenko}, \citenamefont {Chiu}, \citenamefont {Ji}, \citenamefont {Greif},\ and\ \citenamefont {Greiner}}]{parsons_2016_siteresolved}%
  \BibitemOpen
  \bibfield  {author} {\bibinfo {author} {\bibfnamefont {M.~F.}\ \bibnamefont {Parsons}}, \bibinfo {author} {\bibfnamefont {A.}~\bibnamefont {Mazurenko}}, \bibinfo {author} {\bibfnamefont {C.~S.}\ \bibnamefont {Chiu}}, \bibinfo {author} {\bibfnamefont {G.}~\bibnamefont {Ji}}, \bibinfo {author} {\bibfnamefont {D.}~\bibnamefont {Greif}},\ and\ \bibinfo {author} {\bibfnamefont {M.}~\bibnamefont {Greiner}},\ }\bibfield  {title} {\bibinfo {title} {Site-resolved measurement of the spin-correlation function in the {{Fermi-Hubbard}} model},\ }\href {https://doi.org/10.1126/science.aag1430} {\bibfield  {journal} {\bibinfo  {journal} {Science}\ }\textbf {\bibinfo {volume} {353}},\ \bibinfo {pages} {1253} (\bibinfo {year} {2016})}\BibitemShut {NoStop}%
\bibitem [{Note3()}]{Note3}%
  \BibitemOpen
  \bibinfo {note} {We denote $\ket {0-}$ to express that in the first position there are no particles, and in the second position there is a red particle and so on.}\BibitemShut {Stop}%
\bibitem [{\citenamefont {{\v S}afr{\'a}nek}\ \emph {et~al.}(2023)\citenamefont {{\v S}afr{\'a}nek}, \citenamefont {Rosa},\ and\ \citenamefont {Binder}}]{safranek_2023_work}%
  \BibitemOpen
  \bibfield  {author} {\bibinfo {author} {\bibfnamefont {D.}~\bibnamefont {{\v S}afr{\'a}nek}}, \bibinfo {author} {\bibfnamefont {D.}~\bibnamefont {Rosa}},\ and\ \bibinfo {author} {\bibfnamefont {F.~C.}\ \bibnamefont {Binder}},\ }\bibfield  {title} {\bibinfo {title} {Work {{Extraction}} from {{Unknown Quantum Sources}}},\ }\href {https://doi.org/10.1103/PhysRevLett.130.210401} {\bibfield  {journal} {\bibinfo  {journal} {Phys. Rev. Lett.}\ }\textbf {\bibinfo {volume} {130}},\ \bibinfo {pages} {210401} (\bibinfo {year} {2023})}\BibitemShut {NoStop}%
\bibitem [{\citenamefont {Sagawa}(2012)}]{sagawa_2012_thermodynamics}%
  \BibitemOpen
  \bibfield  {author} {\bibinfo {author} {\bibfnamefont {T.}~\bibnamefont {Sagawa}},\ }\href {https://doi.org/10.1007/978-4-431-54168-4} {\emph {\bibinfo {title} {Thermodynamics of Information Processing in Small Systems}}},\ Vol.\ \bibinfo {volume} {127}\ (\bibinfo {year} {2012})\ pp.\ \bibinfo {pages} {1--56}\BibitemShut {NoStop}%
\bibitem [{\citenamefont {Parrondo}\ \emph {et~al.}(2015)\citenamefont {Parrondo}, \citenamefont {Horowitz},\ and\ \citenamefont {Sagawa}}]{parrondo_2014_thermodynamics}%
  \BibitemOpen
  \bibfield  {author} {\bibinfo {author} {\bibfnamefont {J.~M.~R.}\ \bibnamefont {Parrondo}}, \bibinfo {author} {\bibfnamefont {J.~M.}\ \bibnamefont {Horowitz}},\ and\ \bibinfo {author} {\bibfnamefont {T.}~\bibnamefont {Sagawa}},\ }\bibfield  {title} {\bibinfo {title} {Thermodynamics of information},\ }\href {https://doi.org/10.1038/nphys3230} {\bibfield  {journal} {\bibinfo  {journal} {Nature Physics}\ }\textbf {\bibinfo {volume} {11}},\ \bibinfo {pages} {131} (\bibinfo {year} {2015})}\BibitemShut {NoStop}%
\bibitem [{\citenamefont {Davies}\ \emph {et~al.}(2021)\citenamefont {Davies}, \citenamefont {Thomas},\ and\ \citenamefont {Zahariade}}]{davies_2021_the}%
  \BibitemOpen
  \bibfield  {author} {\bibinfo {author} {\bibfnamefont {P.~C.~W.}\ \bibnamefont {Davies}}, \bibinfo {author} {\bibfnamefont {L.}~\bibnamefont {Thomas}},\ and\ \bibinfo {author} {\bibfnamefont {G.}~\bibnamefont {Zahariade}},\ }\bibfield  {title} {\bibinfo {title} {The harmonic quantum szilárd engine},\ }\href {https://doi.org/10.1119/10.0005946} {\bibfield  {journal} {\bibinfo  {journal} {American Journal of Physics}\ }\textbf {\bibinfo {volume} {89}},\ \bibinfo {pages} {1123} (\bibinfo {year} {2021})},\ \Eprint {https://arxiv.org/abs/https://pubs.aip.org/aapt/ajp/article-pdf/89/12/1123/20099417/1123\_1\_10.0005946.pdf} {https://pubs.aip.org/aapt/ajp/article-pdf/89/12/1123/20099417/1123\_1\_10.0005946.pdf} \BibitemShut {NoStop}%
\bibitem [{\citenamefont {B{\'e}rut}\ \emph {et~al.}(2012)\citenamefont {B{\'e}rut}, \citenamefont {Arakelyan}, \citenamefont {Petrosyan}, \citenamefont {Ciliberto}, \citenamefont {Dillenschneider},\ and\ \citenamefont {Lutz}}]{Berut_2012_experimental}%
  \BibitemOpen
  \bibfield  {author} {\bibinfo {author} {\bibfnamefont {A.}~\bibnamefont {B{\'e}rut}}, \bibinfo {author} {\bibfnamefont {A.}~\bibnamefont {Arakelyan}}, \bibinfo {author} {\bibfnamefont {A.}~\bibnamefont {Petrosyan}}, \bibinfo {author} {\bibfnamefont {S.}~\bibnamefont {Ciliberto}}, \bibinfo {author} {\bibfnamefont {R.}~\bibnamefont {Dillenschneider}},\ and\ \bibinfo {author} {\bibfnamefont {E.}~\bibnamefont {Lutz}},\ }\bibfield  {title} {\bibinfo {title} {Experimental verification of landauer's principle linking information and thermodynamics},\ }\href {https://doi.org/10.1038/nature10872} {\bibfield  {journal} {\bibinfo  {journal} {Nature}\ }\textbf {\bibinfo {volume} {483}},\ \bibinfo {pages} {187} (\bibinfo {year} {2012})}\BibitemShut {NoStop}%
\bibitem [{\citenamefont {Koski}\ \emph {et~al.}(2014)\citenamefont {Koski}, \citenamefont {Maisi}, \citenamefont {Pekola},\ and\ \citenamefont {Averin}}]{koski_2014_experimental}%
  \BibitemOpen
  \bibfield  {author} {\bibinfo {author} {\bibfnamefont {J.~V.}\ \bibnamefont {Koski}}, \bibinfo {author} {\bibfnamefont {V.~F.}\ \bibnamefont {Maisi}}, \bibinfo {author} {\bibfnamefont {J.~P.}\ \bibnamefont {Pekola}},\ and\ \bibinfo {author} {\bibfnamefont {D.~V.}\ \bibnamefont {Averin}},\ }\bibfield  {title} {\bibinfo {title} {Experimental realization of a szilard engine with a single electron},\ }\href {https://doi.org/10.1073/pnas.1406966111} {\bibfield  {journal} {\bibinfo  {journal} {Proceedings of the National Academy of Sciences}\ }\textbf {\bibinfo {volume} {111}},\ \bibinfo {pages} {13786} (\bibinfo {year} {2014})},\ \Eprint {https://arxiv.org/abs/https://www.pnas.org/doi/pdf/10.1073/pnas.1406966111} {https://www.pnas.org/doi/pdf/10.1073/pnas.1406966111} \BibitemShut {NoStop}%
\bibitem [{\citenamefont {Peterson}\ \emph {et~al.}(2016)\citenamefont {Peterson}, \citenamefont {Sarthour}, \citenamefont {Souza}, \citenamefont {Oliveira}, \citenamefont {Goold}, \citenamefont {Modi}, \citenamefont {Soares-Pinto},\ and\ \citenamefont {Céleri}}]{peterson_2016_experimental}%
  \BibitemOpen
  \bibfield  {author} {\bibinfo {author} {\bibfnamefont {J.~P.~S.}\ \bibnamefont {Peterson}}, \bibinfo {author} {\bibfnamefont {R.~S.}\ \bibnamefont {Sarthour}}, \bibinfo {author} {\bibfnamefont {A.~M.}\ \bibnamefont {Souza}}, \bibinfo {author} {\bibfnamefont {I.~S.}\ \bibnamefont {Oliveira}}, \bibinfo {author} {\bibfnamefont {J.}~\bibnamefont {Goold}}, \bibinfo {author} {\bibfnamefont {K.}~\bibnamefont {Modi}}, \bibinfo {author} {\bibfnamefont {D.~O.}\ \bibnamefont {Soares-Pinto}},\ and\ \bibinfo {author} {\bibfnamefont {L.~C.}\ \bibnamefont {Céleri}},\ }\bibfield  {title} {\bibinfo {title} {Experimental demonstration of information to energy conversion in a quantum system at the landauer limit},\ }\href {https://doi.org/10.1098/rspa.2015.0813} {\bibfield  {journal} {\bibinfo  {journal} {Proceedings of the Royal Society A: Mathematical, Physical and Engineering Sciences}\ }\textbf {\bibinfo {volume} {472}},\ \bibinfo {pages} {20150813} (\bibinfo {year} {2016})},\ \Eprint
  {https://arxiv.org/abs/https://royalsocietypublishing.org/doi/pdf/10.1098/rspa.2015.0813} {https://royalsocietypublishing.org/doi/pdf/10.1098/rspa.2015.0813} \BibitemShut {NoStop}%
\bibitem [{\citenamefont {Boyd}\ and\ \citenamefont {Crutchfield}(2016)}]{boyd_2016_maxwell}%
  \BibitemOpen
  \bibfield  {author} {\bibinfo {author} {\bibfnamefont {A.~B.}\ \bibnamefont {Boyd}}\ and\ \bibinfo {author} {\bibfnamefont {J.~P.}\ \bibnamefont {Crutchfield}},\ }\bibfield  {title} {\bibinfo {title} {Maxwell demon dynamics: Deterministic chaos, the szilard map, and the intelligence of thermodynamic systems},\ }\href {https://doi.org/10.1103/PhysRevLett.116.190601} {\bibfield  {journal} {\bibinfo  {journal} {Phys. Rev. Lett.}\ }\textbf {\bibinfo {volume} {116}},\ \bibinfo {pages} {190601} (\bibinfo {year} {2016})}\BibitemShut {NoStop}%
\bibitem [{\citenamefont {Boyd}\ \emph {et~al.}(2017)\citenamefont {Boyd}, \citenamefont {Mandal}, \citenamefont {Riechers},\ and\ \citenamefont {Crutchfield}}]{boyd_2017_transient}%
  \BibitemOpen
  \bibfield  {author} {\bibinfo {author} {\bibfnamefont {A.~B.}\ \bibnamefont {Boyd}}, \bibinfo {author} {\bibfnamefont {D.}~\bibnamefont {Mandal}}, \bibinfo {author} {\bibfnamefont {P.~M.}\ \bibnamefont {Riechers}},\ and\ \bibinfo {author} {\bibfnamefont {J.~P.}\ \bibnamefont {Crutchfield}},\ }\bibfield  {title} {\bibinfo {title} {Transient dissipation and structural costs of physical information transduction},\ }\href {https://doi.org/10.1103/PhysRevLett.118.220602} {\bibfield  {journal} {\bibinfo  {journal} {Phys. Rev. Lett.}\ }\textbf {\bibinfo {volume} {118}},\ \bibinfo {pages} {220602} (\bibinfo {year} {2017})}\BibitemShut {NoStop}%
\bibitem [{\citenamefont {Masuyama}\ \emph {et~al.}(2018)\citenamefont {Masuyama}, \citenamefont {Funo}, \citenamefont {Murashita}, \citenamefont {Noguchi}, \citenamefont {Kono}, \citenamefont {Tabuchi}, \citenamefont {Yamazaki}, \citenamefont {Ueda},\ and\ \citenamefont {Nakamura}}]{masuyama_2018_information}%
  \BibitemOpen
  \bibfield  {author} {\bibinfo {author} {\bibfnamefont {Y.}~\bibnamefont {Masuyama}}, \bibinfo {author} {\bibfnamefont {K.}~\bibnamefont {Funo}}, \bibinfo {author} {\bibfnamefont {Y.}~\bibnamefont {Murashita}}, \bibinfo {author} {\bibfnamefont {A.}~\bibnamefont {Noguchi}}, \bibinfo {author} {\bibfnamefont {S.}~\bibnamefont {Kono}}, \bibinfo {author} {\bibfnamefont {Y.}~\bibnamefont {Tabuchi}}, \bibinfo {author} {\bibfnamefont {R.}~\bibnamefont {Yamazaki}}, \bibinfo {author} {\bibfnamefont {M.}~\bibnamefont {Ueda}},\ and\ \bibinfo {author} {\bibfnamefont {Y.}~\bibnamefont {Nakamura}},\ }\bibfield  {title} {\bibinfo {title} {Information-to-work conversion by maxwell's demon in a superconducting circuit quantum electrodynamical system},\ }\href {https://doi.org/10.1038/s41467-018-03686-y} {\bibfield  {journal} {\bibinfo  {journal} {Nature Communications}\ }\textbf {\bibinfo {volume} {9}},\ \bibinfo {pages} {1291} (\bibinfo {year} {2018})}\BibitemShut {NoStop}%
\bibitem [{\citenamefont {Van~Horne}\ \emph {et~al.}(2020)\citenamefont {Van~Horne}, \citenamefont {Yum}, \citenamefont {Dutta}, \citenamefont {H{\"a}nggi}, \citenamefont {Gong}, \citenamefont {Poletti},\ and\ \citenamefont {Mukherjee}}]{vanhorne_2020_single-atom}%
  \BibitemOpen
  \bibfield  {author} {\bibinfo {author} {\bibfnamefont {N.}~\bibnamefont {Van~Horne}}, \bibinfo {author} {\bibfnamefont {D.}~\bibnamefont {Yum}}, \bibinfo {author} {\bibfnamefont {T.}~\bibnamefont {Dutta}}, \bibinfo {author} {\bibfnamefont {P.}~\bibnamefont {H{\"a}nggi}}, \bibinfo {author} {\bibfnamefont {J.}~\bibnamefont {Gong}}, \bibinfo {author} {\bibfnamefont {D.}~\bibnamefont {Poletti}},\ and\ \bibinfo {author} {\bibfnamefont {M.}~\bibnamefont {Mukherjee}},\ }\bibfield  {title} {\bibinfo {title} {Single-atom energy-conversion device with a quantum load},\ }\href {https://doi.org/10.1038/s41534-020-0264-6} {\bibfield  {journal} {\bibinfo  {journal} {npj Quantum Information}\ }\textbf {\bibinfo {volume} {6}},\ \bibinfo {pages} {37} (\bibinfo {year} {2020})}\BibitemShut {NoStop}%
\bibitem [{\citenamefont {Marathe}\ and\ \citenamefont {Parrondo}(2010)}]{marathe_2010_cooling}%
  \BibitemOpen
  \bibfield  {author} {\bibinfo {author} {\bibfnamefont {R.}~\bibnamefont {Marathe}}\ and\ \bibinfo {author} {\bibfnamefont {J.~M.~R.}\ \bibnamefont {Parrondo}},\ }\bibfield  {title} {\bibinfo {title} {Cooling classical particles with a microcanonical szilard engine},\ }\href {https://doi.org/10.1103/PhysRevLett.104.245704} {\bibfield  {journal} {\bibinfo  {journal} {Phys. Rev. Lett.}\ }\textbf {\bibinfo {volume} {104}},\ \bibinfo {pages} {245704} (\bibinfo {year} {2010})}\BibitemShut {NoStop}%
\bibitem [{\citenamefont {Admon}\ \emph {et~al.}(2018)\citenamefont {Admon}, \citenamefont {Rahav},\ and\ \citenamefont {Roichman}}]{tamir2018experimental}%
  \BibitemOpen
  \bibfield  {author} {\bibinfo {author} {\bibfnamefont {T.}~\bibnamefont {Admon}}, \bibinfo {author} {\bibfnamefont {S.}~\bibnamefont {Rahav}},\ and\ \bibinfo {author} {\bibfnamefont {Y.}~\bibnamefont {Roichman}},\ }\bibfield  {title} {\bibinfo {title} {Experimental realization of an information machine with tunable temporal correlations},\ }\href {https://doi.org/10.1103/PhysRevLett.121.180601} {\bibfield  {journal} {\bibinfo  {journal} {Phys. Rev. Lett.}\ }\textbf {\bibinfo {volume} {121}},\ \bibinfo {pages} {180601} (\bibinfo {year} {2018})}\BibitemShut {NoStop}%
\bibitem [{\citenamefont {Vaikuntanathan}\ and\ \citenamefont {Jarzynski}(2011)}]{vaikuntanathan_2011_modelling}%
  \BibitemOpen
  \bibfield  {author} {\bibinfo {author} {\bibfnamefont {S.}~\bibnamefont {Vaikuntanathan}}\ and\ \bibinfo {author} {\bibfnamefont {C.}~\bibnamefont {Jarzynski}},\ }\bibfield  {title} {\bibinfo {title} {Modeling maxwell's demon with a microcanonical szilard engine},\ }\href {https://doi.org/10.1103/PhysRevE.83.061120} {\bibfield  {journal} {\bibinfo  {journal} {Phys. Rev. E}\ }\textbf {\bibinfo {volume} {83}},\ \bibinfo {pages} {061120} (\bibinfo {year} {2011})}\BibitemShut {NoStop}%
\bibitem [{\citenamefont {Ito}\ and\ \citenamefont {Sagawa}(2013)}]{sosuke_2013_information}%
  \BibitemOpen
  \bibfield  {author} {\bibinfo {author} {\bibfnamefont {S.}~\bibnamefont {Ito}}\ and\ \bibinfo {author} {\bibfnamefont {T.}~\bibnamefont {Sagawa}},\ }\bibfield  {title} {\bibinfo {title} {Information thermodynamics on causal networks},\ }\href {https://doi.org/10.1103/PhysRevLett.111.180603} {\bibfield  {journal} {\bibinfo  {journal} {Phys. Rev. Lett.}\ }\textbf {\bibinfo {volume} {111}},\ \bibinfo {pages} {180603} (\bibinfo {year} {2013})}\BibitemShut {NoStop}%
\bibitem [{\citenamefont {Malgaretti}\ and\ \citenamefont {Stark}(2022)}]{malgaretti_2022_szilard}%
  \BibitemOpen
  \bibfield  {author} {\bibinfo {author} {\bibfnamefont {P.}~\bibnamefont {Malgaretti}}\ and\ \bibinfo {author} {\bibfnamefont {H.}~\bibnamefont {Stark}},\ }\bibfield  {title} {\bibinfo {title} {Szilard engines and information-based work extraction for active systems},\ }\href {https://doi.org/10.1103/PhysRevLett.129.228005} {\bibfield  {journal} {\bibinfo  {journal} {Phys. Rev. Lett.}\ }\textbf {\bibinfo {volume} {129}},\ \bibinfo {pages} {228005} (\bibinfo {year} {2022})}\BibitemShut {NoStop}%
\bibitem [{\citenamefont {Chor}\ \emph {et~al.}(2023)\citenamefont {Chor}, \citenamefont {Sohachi}, \citenamefont {Goerlich}, \citenamefont {Rosen}, \citenamefont {Rahav},\ and\ \citenamefont {Roichman}}]{chor_2023_many-body}%
  \BibitemOpen
  \bibfield  {author} {\bibinfo {author} {\bibfnamefont {O.}~\bibnamefont {Chor}}, \bibinfo {author} {\bibfnamefont {A.}~\bibnamefont {Sohachi}}, \bibinfo {author} {\bibfnamefont {R.}~\bibnamefont {Goerlich}}, \bibinfo {author} {\bibfnamefont {E.}~\bibnamefont {Rosen}}, \bibinfo {author} {\bibfnamefont {S.}~\bibnamefont {Rahav}},\ and\ \bibinfo {author} {\bibfnamefont {Y.}~\bibnamefont {Roichman}},\ }\bibfield  {title} {\bibinfo {title} {Many-body szil\'ard engine with giant number fluctuations},\ }\href {https://doi.org/10.1103/PhysRevResearch.5.043193} {\bibfield  {journal} {\bibinfo  {journal} {Phys. Rev. Res.}\ }\textbf {\bibinfo {volume} {5}},\ \bibinfo {pages} {043193} (\bibinfo {year} {2023})}\BibitemShut {NoStop}%
\end{thebibliography}%

\end{document}